\PassOptionsToPackage{table}{xcolor} 
\newif\ifarxiv
\arxivtrue
\ifarxiv
\documentclass[acmsmall, screen, authorversion, nonacm]{acmart}
\else
\documentclass[acmsmall, screen]{acmart}
\fi

\setcopyright{rightsretained}
\acmDOI{10.1145/3704874}
\acmYear{2025}
\acmJournal{PACMPL}
\acmVolume{9}
\acmNumber{POPL}
\acmArticle{38}
\acmMonth{1}
\ifarxiv\else
\received{2024-07-11}
\received[accepted]{2024-11-07}
\fi

\usepackage{amsmath} 
\usepackage{amsfonts} 
\usepackage{amsthm} 
\usepackage{mathtools} 
\usepackage{stmaryrd} 
\usepackage{pifont} 
\usepackage{bm} 
\usepackage{multirow} 
\usepackage{wrapfig} 
\usepackage{tabularx} 
\usepackage{multirow} 

\usepackage{algorithm} 
\usepackage[noend]{algpseudocode} 

\usepackage{subcaption} 
\usepackage{thm-restate} 
\usepackage[capitalise]{cleveref} 
\usepackage{crossreftools} 

\AtBeginDocument{%
  \mathchardef\mathcomma\mathcode`\,
  \mathcode`\,="8000 
}
{\catcode`,=\active
  \gdef,{\mathcomma\nolinebreak[3]}
}

\theoremstyle{plain}
\newtheorem{theorem}{Theorem}[section]
\newtheorem{lemma}[theorem]{Lemma}

\theoremstyle{definition}
\newtheorem{definition}[theorem]{Definition}
\newtheorem{example}[theorem]{Example}

\newcommand{\alternaterowcolors}{\rowcolors{2}{}{gray!20}}

\newcommand{\cmark}{\ding{51}} 
\newcommand{\xmark}{\ding{55}} 
\newcommand{\pmark}{\ensuremath{\sim}} 

\newcommand{\defn}[1]{\textbf{\emph{#1}}}

\renewcommand{\epsilon}{\varepsilon}
\renewcommand{\phi}{\varphi}

\newcommand{\RR}{\mathbb R}
\newcommand{\nnr}{\mathbb{\RR}_{\ge 0}} 
\newcommand{\NN}{\mathbb N}
\newcommand{\ZZ}{\mathbb Z}

\DeclareMathOperator{\Meas}{Meas}
\DeclareMathOperator{\Normalize}{normalize}

\newcommand{\keyword}[1]{\textcolor{violet}{\mathsf{#1}}}
\renewcommand{\Comment}[1]{\textcolor{gray}{// \; #1}}
\newcommand{\Skip}{\keyword{skip}}

\newcommand{\Ite}[4][]{\keyword{if}\, {#2} \, \allowbreak \{ #3 \} \, #1 \allowbreak \keyword{else} \, \{ #4 \}}
\newcommand{\Whilst}[2]{\keyword{while}\, {#1} \, \allowbreak \{ #2 \}}
\newcommand{\Diverge}{\keyword{diverge}}

\newcommand{\Fail}{\keyword{fail}}
\newcommand{\Observe}{\keyword{observe} \,}

\newcommand{\Flip}{\keyword{flip}}

\newcommand{\monus}{\mathbin{\dot-}}
\newcommand{\passign}{\mathbin{{+}{=}}}

\newcommand{\massign}{\mathbin{{\dot-}{=}}}
\newcommand{\ProbBranch}[3]{\{ #1 \} \mathrel{[#2]} \{ #3 \}}
\newcommand{\True}{\keyword{true}}
\newcommand{\False}{\keyword{false}}

\newcommand{\estates}{\NN^n_{\lightning}} 

\newcommand{\distname}[1]{\textcolor{teal}{\mathsf{#1}}}
\newcommand{\Dirac}{\distname{Dirac}}
\newcommand{\Bernoulli}{\distname{Bernoulli}}

\newcommand{\Geometric}{\distname{Geometric}}

\newcommand{\Uniform}{\distname{Uniform}}

\newcommand{\egd}[2]{\distname{EGD}\mathopen{}\left(#1, #2\right)\mathclose{}}

\newcommand{\X}{{\bm X}}
\newcommand{\ExpVal}{\mathbb{E}}
\newcommand{\Prob}{\mathbb{P}}

\newcommand{\sem}[1]{\llbracket #1 \rrbracket} 
\newcommand{\esem}[1]{\llbracket #1 \rrbracket_{\lightning}} 
\newcommand{\lfp}[1]{\mathsf{lfp}(#1)}

\newcommand{\x}{{\bm x}}
\newcommand{\w}{{\bm w}}
\newcommand{\ii}{{\bm i}}
\newcommand{\oo}{{\bm o}}
\newcommand{\balpha}{{\bm \alpha}}
\newcommand{\bbeta}{{\bm \beta}}
\newcommand{\bgamma}{{\bm \gamma}}
\newcommand{\bdelta}{{\bm \delta}}
\newcommand{\blambda}{{\bm \lambda}}
\newcommand{\zero}{\mathbf{0}}
\newcommand{\one}{\mathbf{1}}

\renewcommand{\P}{\mathbf{P}}
\newcommand{\Q}{\mathbf{Q}}
\newcommand{\R}{\mathbf{R}}
\renewcommand{\S}{\mathbf{S}}
\newcommand{\T}{\mathbf{T}}

\newcommand{\semlo}[1]{\llbracket #1 \rrbracket_{\mathsf{lo}}}
\newcommand{\semres}[1]{\llbracket #1 \rrbracket_{\mathsf{res}}}

\newcommand{\evgeo}[2]{\left.#1\right|_{#2}^{\mathsf{geo}}}
\newcommand{\semgeo}[1]{\mathrel{\llbracket #1 \rrbracket^{\mathsf{geo}}}}

\newcommand{\mle}{\preceq} 
\newcommand{\mge}{\succeq} 
\newcommand{\trafole}{\sqsubseteq} 
\newcommand{\egdle}{\preceq_{\mathsf{EGD}}} 
\newcommand{\egdge}{\succeq_{\mathsf{EGD}}} 
\newcommand{\JoinRel}{\mathsf{Join}}

\newcommand{\widen}{\mathop{\nabla}}

\newcommand{\apxref}[1]{\ifarxiv \cref{#1}\else \citet{ZaiserMO24}\fi}

\allowdisplaybreaks 
\makeatletter
\renewcommand{\cite}[1]{\@latex@error{Use citet or citep instead of cite}{}}
\makeatother

\begin{document}

\title{Guaranteed Bounds on Posterior Distributions of Discrete Probabilistic Programs with Loops}

\author{Fabian Zaiser}
\orcid{0000-0001-5158-2002}
\affiliation{%
  \institution{University of Oxford}
  \city{Oxford}
  \country{United Kingdom}
}
\email{fabian.zaiser@cs.ox.ac.uk}

\author{Andrzej S. Murawski}
\orcid{0000-0002-4725-410X}
\affiliation{%
  \institution{University of Oxford}
  \city{Oxford}
  \country{United Kingdom}
}
\email{andrzej.murawski@cs.ox.ac.uk}

\author{C.-H. Luke Ong}
\orcid{0000-0001-7509-680X}
\affiliation{%
  \institution{Nanyang Technological University}
  \city{Singapore}
  \country{Singapore}
}
\email{luke.ong@ntu.edu.sg}


\begin{abstract}
We study the problem of bounding the posterior distribution of discrete probabilistic programs with unbounded support, loops, and conditioning.
Loops pose the main difficulty in this setting: even if exact Bayesian inference is possible, the state of the art requires user-provided loop invariant templates.
By contrast, we aim to find \emph{guaranteed bounds}, which sandwich the true distribution.
They are fully automated, applicable to more programs and provide more provable guarantees than approximate sampling-based inference.
Since lower bounds can be obtained by unrolling loops, the main challenge is upper bounds, and we attack it in two ways.
The first is called \emph{residual mass semantics}, which is a flat bound based on the residual probability mass of a loop.
The approach is simple, efficient, and has provable guarantees.

The main novelty of our work is the second approach, called \emph{geometric bound semantics}.
It operates on a novel family of distributions, called \emph{eventually geometric distributions} (EGDs), and can bound the distribution of loops with a new form of loop invariants called \emph{contraction invariants}.
The invariant synthesis problem reduces to a system of polynomial inequality constraints, which is a decidable problem with automated solvers.
If a solution exists, it yields an exponentially decreasing bound on the \emph{whole} distribution, and can therefore bound moments and tail asymptotics as well, not just probabilities as in the first approach.

Both semantics enjoy desirable theoretical properties.
In particular, we prove soundness and convergence, i.e. the bounds converge to the exact posterior as loops are unrolled further.
We also investigate sufficient and necessary conditions for the existence of geometric bounds.
On the practical side, we describe \emph{Diabolo}, a fully-automated implementation of both semantics, and evaluate them on a variety of benchmarks from the literature, demonstrating their general applicability and the utility of the resulting bounds.
\end{abstract}



\maketitle

\ifarxiv
\vspace{2em}
\paragraph{Note}
This is the full version of the paper with the same name published at POPL 2025 \citep{ZaiserMO25}.
This document includes an appendix with proofs and additional details omitted from the conference version.
\vspace{1em}
\fi

\section{Introduction}
\label{sec:introduction}

Probabilistic programming is a discipline that studies programming languages with probabilistic constructs \citep{BartheKS2020}.
The term is overloaded however.
At the intersection with randomized algorithms and program analysis, it usually means a programming language with a construct for probabilistic branching or sampling from probability distributions.
As such, it is simply a language to express programs with random numbers and researchers study program analysis techniques for termination probabilities, safety properties, cost analysis, and others.
At the intersection with statistics and machine learning, probabilistic programming is used to express (Bayesian) statistical models \citep{MeentPYW18}.
Bayesian inference is a very successful framework for reasoning and learning under uncertainty: it updates prior beliefs about the world with observed data to obtain posterior beliefs using Bayes' rule.
As such, the programming languages for Bayesian models provide a construct for conditioning on data in addition to sampling from distributions.
Since Bayesian inference is a difficult problem, a lot of research focuses on inference algorithms, in particular their correctness and efficiency.
This paper contributes to both areas by developing methods to bound the distributions arising from probabilistic programs, especially those with loops.

\begin{wrapfigure}{r}{0.3\textwidth}
\vspace{-2em}
\begin{minipage}{0.3\textwidth}
\small
\begin{align*}
&\mathit{Throws} := 0; \mathit{Die} := 0; \\
&\Whilst{\mathit{Die} \ne 6}{ \\
&\quad \mathit{Die} \sim \Uniform\{1, \dots, 6\}; \\
&\quad \Observe \mathit{Die} \in \{2, 4, 6\}; \\
&\quad \mathit{Throws} \passign 1 }
\end{align*}
\end{minipage}
\caption{A probabilistic program with a loop and conditioning}
\label{fig:die-paradox}
\vspace{-1em}
\end{wrapfigure}

\begin{example}
\label{ex:die-paradox}
To illustrate the concept, consider the following puzzle due to Elchanan Mossel.
\begin{quote}
You throw a fair six-sided die repeatedly until you get a 6.
You observe only even numbers during the throws.
What is the expected number of throws (including the 6) conditioned on this event?
\end{quote}
This is a surprisingly tricky problem and most people get it wrong on the first try\footnote{In a survey on Gil Kalai's blog, only 27\% of participants chose the correct answer (\url{https://gilkalai.wordpress.com/2017/09/07/tyi-30-expected-number-of-dice-throws/}).}, based on the incorrect assumption that it is equivalent to throwing a die with only the three faces 2, 4, and 6.
Probability theory and statistics abound with such counterintuitive results (e.g. the Monty-Hall problem), and probabilistic programming offers a precise way to disambiguate their description and make them amenable to automatic analysis and inference tools.
Mossel's problem can be expressed as the probabilistic program in \cref{fig:die-paradox}.
The program has a loop that samples a die until it shows 6, and conditions on the number being even.
In each iteration, the counter $\mathit{Throws}$ is incremented.
\end{example}

\subsection{Challenges}

\paragraph{Bayesian inference}
In Bayesian inference, Bayes' rule is used to update prior distributions $p(\theta)$ of model variables $\theta$ with observed data $x$ to obtain posterior distributions: $p(\theta \mid x) = \frac{p(x \mid \theta) p(\theta)}{p(x)}$.
In practice, such Bayesian statistical models are too complex for manual calculations and inferring their posterior distribution is a key challenge in Bayesian statistics.
There are two approaches: exact and approximate inference.
\emph{Exact inference} aims to find an exact representation of the posterior distribution.
Such methods impose heavy restrictions on the supported probabilistic programs and do not usually scale well.
Practitioners therefore mostly use \emph{approximate methods} that do not aim to compute this distribution exactly, but rather to produce unbiased or consistent samples from it.
If the probabilistic program does not contain conditioning, samples can simply be obtained by running the program.
But with observations, program runs that violate the observations must be rejected.
Since the likelihood of the observations is typically low, simple rejection sampling is inefficient, and thus practical samplers use more sophisticated techniques, such as Markov chain Monte Carlo.
While more scalable, these approaches typically do not provide strong guarantees on the approximation error after a finite amount of time \citep[Section 11.5]{GelmanCSDVR13}.

\paragraph{Loops}
Loops are essential to the expressiveness of programming languages but notoriously hard to analyze.
This applies even more strongly to the probabilistic setting, where deciding properties like termination is harder than in the deterministic setting \citep{KaminskiK15}.
Even if a program does not use conditioning, loops can still make sampling difficult.
For example, a program may terminate almost surely, but its expected running time may be infinite.
This prevents sampling-based approaches since they need to run the program.
Furthermore, many inference algorithms are not designed to handle unbounded loops and may return erroneous results for such programs \citep{BeutnerOZ22}.
On the formal methods side, various approaches for probabilistic loop analysis have been proposed, employing techniques such as martingales, moments, and generating functions (see \cref{sec:related-work}).
If all program variables have finite support, the program can be translated to a probabilistic transition system and techniques from probabilistic model checking can be used.

None of these analysis techniques can be applied to \cref{ex:die-paradox} however: methods from cost analysis do not support conditioning and probabilistic model checking requires finite support (but $\mathit{Throws}$ is supported on $\NN$).
The approach by \citet{KlinkenbergBCHK24} via generating functions is theoretically applicable, but requires the user to provide a loop invariant template, i.e. a loop invariant where certain parameters may be missing.
Unfortunately, such an invariant cannot always be specified in their language \citep[Example 25]{KlinkenbergBCHK24}.
Even in cases where this is possible, we argue that figuring out its shape is actually the hard part: it already requires a good understanding of the probabilistic program and its distribution, so it is not a satisfactory solution.

\subsection{Guaranteed Bounds}
To deal with the above challenges, we investigate \defn{guaranteed bounds} on the program distribution.
``Guaranteed'' here refers to a method that computes deterministic (non-stochastic) results about the mathematical denotation of a program \citep{BeutnerOZ22}.
Such bounds are applicable more often than exact inference, e.g. in the presence of loops/recursion, and provide more assurance than approximate methods, which have at best stochastic guarantees.
Why are such bounds useful?

\begin{wrapfigure}{r}{0.4\textwidth}
\vspace{-1em}
\centering
\includegraphics[width=0.4\textwidth]{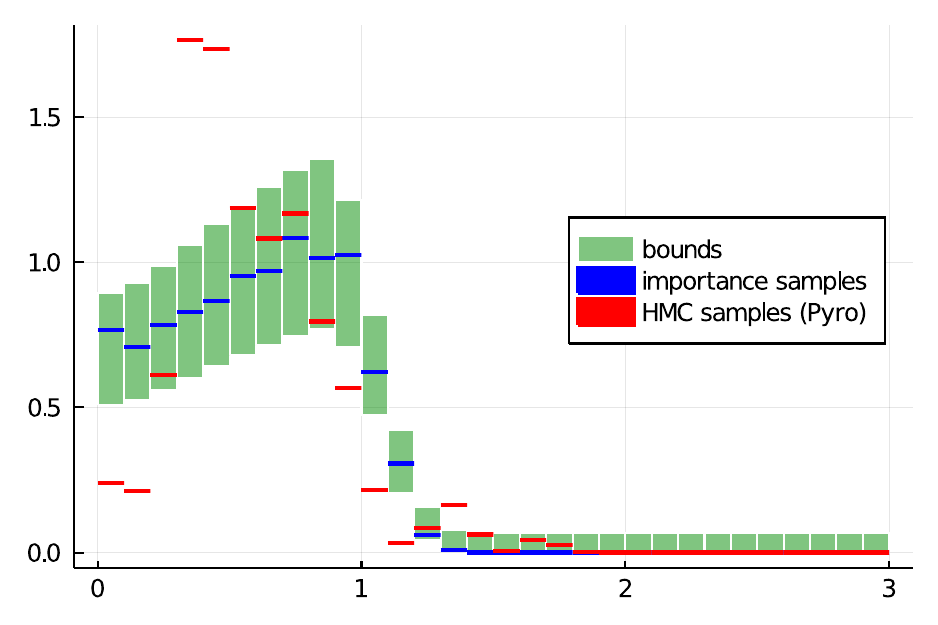}
\caption{Histogram of samples from two inference algorithms (importance sampling and Pyro's HMC), and the guaranteed bounds from \citet{BeutnerOZ22}.
The bounds show that Pyro's HMC produces wrong results.
(Source: \citet{BeutnerOZ22})}
\label{fig:pyro-wrong}
\vspace{-1em}
\end{wrapfigure}

\paragraph{Partial correctness properties}
In quantitative program analysis, one can verify safety properties by bounding the probability of reaching an unsafe state.
Bounding reachability probabilities is also a common problem in probabilistic model checking and quantitative program verification, yet it has not seen much attention in the context of probabilistic programming with conditioning, aside from the work by \citet{BeutnerOZ22} and \citet{WangYFLO24}.
Neither of those can bound moments of infinite-support distributions, whereas our work finds tight bounds on the expected value of $\mathit{Throws}$ in \cref{fig:die-paradox} (see \cref{sec:eval-comparison}).

\paragraph{Checking approximate inference}
In the context of Bayesian inference, the bounds can be useful to check and debug approximate inference algorithms and their implementations.
If the approximate results clearly contradict the bounds, the inference algorithm is likely incorrect, or some of its assumptions are violated, or it has not converged.
\Citet{BeutnerOZ22} provide an example of this: the inference tool Pyro yields wrong results for a probabilistic program with loops, but their bounds can detect this issue (\cref{fig:pyro-wrong}).%
\footnote{
The cause turned out to be an undocumented assumption in the inference algorithm.
Pyro's implementation seems to assume that the dimension (number of samples in a program run) of the problem is constant, which is violated when sampling inside probabilistic loops.
}
Another problem with approximate inference is the tail behavior of the posterior distribution, which is often crucial for the quality of the approximation \citep{LiangHM23}.
Previous work on guaranteed bounds \citep{BeutnerOZ22,WangYFLO24} does not address this aspect, but our work can bound the tail behavior as well.

\begin{table}
\centering
\footnotesize
\caption{
Comparison of our two approaches with the most relevant related work on probabilistic programs with loops.
(Cond.: supports (Bayesian) conditioning; Inf.: allows branching on variables with infinite support; Cont.: allows continuous distributions; Auto.: fully automated; Prob.: computes/bounds probabilities; Mom.: computes/bounds moments; Tails: computes/bounds tail asymptotics of this shape; A.a. (always applicable): yields a result for any program expressible in the respective language.
Partial support is denoted by ``\pmark''.)
}
\label{table:compare-related}
\alternaterowcolors
\begin{tabular}{lccccccccc}
\toprule
&\textbf{Type} &\textbf{Cond.?} &\textbf{Inf.?} &\textbf{Cont.?} &\textbf{Auto.?} &\textbf{Prob.?} &\textbf{Mom.?} &\textbf{Tails?} &\textbf{A.a.?} \\
\midrule
\citet{MoosbruggerSBK22} &exact &\xmark &\xmark &\cmark &\cmark &\pmark &\cmark &$O(n^{-k})$ &\xmark \\
\citet{BeutnerOZ22} &bounds &\cmark &\cmark &\cmark &\cmark &\cmark &\xmark &\xmark &\cmark \\
\citet{WangYFLO24} &bounds &\cmark &\cmark &\cmark &\cmark &\cmark &\xmark &\xmark &\xmark \\
\citet{KlinkenbergBCHK24} &exact &\cmark &\cmark &\xmark &\xmark &\cmark &\cmark &\xmark &\xmark \\
Resid. mass sem. (\cref{sec:unrolling-bounds}) &bounds &\cmark &\cmark &\xmark &\cmark &\cmark &\xmark &\xmark &\cmark \\
Geom. bounds (\cref{sec:geo-bounds}) &bounds &\cmark &\cmark &\xmark &\cmark &\cmark &\cmark &$O(c^n)$ &\xmark \\
\bottomrule
\end{tabular}
\vspace{-1em}
\end{table}

\paragraph{Problem statement}
Given a probabilistic program with posterior distribution $\mu$ on $\NN$, our goal is to bound:
\begin{enumerate}
  \item \emph{probability masses:} given $n \in \NN$, find $l, u \in [0, 1]$ such that $l \le \Prob_{X \sim \mu}[X = n] \le u$;
  \item \emph{moments:} given $k \in \NN$, find $l, u \in \nnr$ such that $l \le \ExpVal_{X \sim \mu}[X^k] \le u$;
  \item \emph{tail asymptotics:} find $c \in [0, 1)$ such that $\Prob_{X \sim \mu}[X = n] = O(c^n)$.
\end{enumerate}

\subsection{Contributions}

In this paper, we develop two new methods to compute guaranteed bounds on the distribution of discrete probabilistic programs with loops and conditioning.
Lower bounds can simply be found by unrolling each loop a finite number of times.
The main challenge is upper bounds and we attack it in two ways: the first is simple, always applicable, and efficient, but coarse; the second is more sophisticated and expensive, but yields much more informative bounds if applicable.
A summary of the most relevant related work is presented in \cref{table:compare-related} and a detailed account in \cref{sec:related-work}.

The first semantics, called \defn{residual mass semantics} (\cref{sec:unrolling-bounds}), is based on the simple idea of bounding the remaining probability mass after the loop unrollings, which has not previously been described, to our knowledge.
We make the following contributions:
\begin{itemize}
\item We introduce the residual mass as a simple but effective idea to bound posterior probabilities.
\item We prove soundness and convergence of the bounds to the true distribution (as loops are unrolled further and further).
\item We implement the semantics in a tool called \emph{Diabolo} and demonstrate empirically that the implementation is more efficient than previous systems (\cref{sec:eval-prev}).
\end{itemize}

The second semantics, called \defn{geometric bound semantics} (\cref{sec:geo-bounds}), is the main novelty of this paper.
The idea is to bound the distribution of loops in a more fine-grained manner with geometric tails, rather than a flat bound as in the first semantics.
\begin{itemize}
\item We present the concept of a \emph{contraction invariant} for a loop, which yields upper bounds on the distribution (\cref{sec:contraction-invariants}).
\item We introduce a family of distributions called \emph{eventually geometric distributions (EGDs)} that are used as an abstract domain to overapproximate the distribution of a loop (\cref{sec:egds}).
\item We present the \emph{geometric bound semantics} (\cref{sec:geo-bound-semantics}) which reduces the synthesis problem of such an invariant to solving a system of polynomial inequalities.
If successful, it immediately yields bounds on probability masses and, contrary to the first semantics, also on moments and tail probabilities of the program distribution.
\item We prove soundness of the semantics and convergence of the bounds, as loops are unrolled further and further (\cref{sec:geo-bound-properties}).
\item We identify necessary conditions and sufficient conditions for its applicability (\cref{sec:geo-bound-properties}).
\item We fully automate it in our tool \emph{Diabolo} (\cref{sec:impl}): contrary to previous work \citep{KlinkenbergBCHK24}, it does not rely on the user to provide a loop invariant (template).
\item We demonstrate its applicability on a large proportion of benchmarks from the literature and compare it with previous approaches and the residual mass semantics (\cref{sec:eval}).
\end{itemize}

Due to space constraints, proofs and some additional details had to be omitted.
They can be found in \ifarxiv\cref{apx:unrolling-bounds,apx:geom-bound-sem,apx:impl,apx:eval}\else the full version of this paper \citep{ZaiserMO24}\fi.

\subsection{Limitations}

Our work deals with \emph{discrete} probabilistic programs with hard conditioning.
This means that programs cannot sample or observe from continuous distributions.
Variables in our programming language take values in $\NN$; negative numbers are not supported (see \cref{sec:future-work} for possible extensions).
While our language is Turing-complete, some arithmetic operations like multiplication as well as some common infinite-support distributions (e.g. Poisson) are not directly supported (see \cref{sec:ppl} for details on our language's expressivity).
The initial values of the program variables are fixed: our methods cannot reason parametrically about these inputs.

The residual mass semantics can yield bounds on the distribution of any such probabilistic program, but convergence with increasing unrolling is only guaranteed if the program terminates almost surely.
If the program distribution has infinite support, we cannot bound the moments or tails: the bound does not contain enough information for this.

The geometric bound semantics yields EGD bounds, which allow bounding moments and tails.
On the other hand, such bounds do not exist for all programs.
Our experiments show that this is not a big concern for many probabilistic programs with loops in practice: EGD bounds exist for a majority of examples we found in the literature.
Another limitation of EGD bounds is that they cannot represent correlation of the tails of two variables, which may lead to imprecise tail bounds or failing to find bounds at all.
Finally, solving the system of polynomial inequalities arising from the semantics, while decidable, can be hard in practice and does not scale to very large programs.
It should be noted that scalability is a general issue in probabilistic program analysis owing to the hardness of the problem \citep{DagumL93} and not specific to our work.

\subsection{Notation and Conventions}
\label{sec:conventions}

We use the Iverson brackets $[\phi]$ to mean 1 if $\phi$ is satisfied and 0 otherwise.
We write variables representing vectors in bold ($\balpha$), tensors (multidimensional arrays) in uppercase and bold ($\T$), and random or program variables in uppercase ($X$).
We write $\zero$ and $\one$ for the constant zero and one functions.
We write $\zero_n$ and $\one_n$ for the zero and ones vector in $\RR^n$.
To set the $k$-th component of $\balpha$ to $v$, we write $\balpha[k \mapsto v]$.
Vectors $\balpha \in \RR^n$ are indexed as $\alpha_1, \dots, \alpha_n$.
We abbreviate $[d] := \{0, \dots, d-1\}$.
Tensors $\T \in \RR^{[d_1] \times \dots \times [d_n]}$ are indexed as $\T_{i_1, \dots, i_n}$ where $i_k$ ranges from $0$ to $d_k - 1$.
We write $\zero_{[d_1] \times \dots \times [d_n]}$ or simply $\zero$ for the zero tensor in $\RR^{[d_1] \times \dots \times [d_n]}$.
We write $|\T| = (d_1 \dots, d_n)$ for the dimensions of $\T \in \RR^{[d_1] \times \dots \times [d_n]}$, and in particular $|\T|_i = d_i$.
To index $\T$ along the $k$-th dimension, we write $\T_{k: j} \in \RR^{[d_1] \times \cdots \times [d_{k - 1}] \times [d_{k + 1}] \times \cdots \times [d_n]}$, which is defined by $(\T_{k: j})_{i_1, \dots, i_{k-1}, i_{k + 1}, \dots, i_n} = \T_{i_1, \dots, i_{k-1}, j, i_{k + 1}, \dots, i_n}$.
We often write tensor indices as $\ii := (i_1, \dots, i_n)$ for brevity.
We also abbreviate $\balpha^{\ii} := \prod_{k=1}^n \alpha_i^{i_k}$.
Other binary operations ($+$, $-$, $\min$, $\max$, etc.) work elementwise on vectors and tensors, e.g. $(\balpha + \bbeta)_j := \alpha_j + \beta_j$ and $\balpha \le \bbeta$ if and only if $\alpha_j \le \beta_j$ for all $j$.

\section{Background}
\label{sec:background}

\subsection{Probability and Measure Theory}

All probability spaces $M$ in this work are equipped with the discrete $\sigma$-algebra $\mathcal{P}(M)$.
The set of measures on a space $M$ is denoted by $\Meas(M)$.
The \defn{restriction} of a measure $\mu \in \Meas(M)$ to a set $M' \subseteq M$ is denoted by $\mu|_{M'}$.
The \defn{mass function} $m: M \to [0,1]$ of a measure $\mu \in \Meas(M)$ is defined by $m(x) := \mu(\{x\})$.
By abuse of notation, we will usually write $\mu$ for $m$ as well.
We will use the words \defn{distribution} and \defn{measure} interchangeably, but a \defn{probability distribution} is a measure $\mu$ with \defn{total mass} $\mu(M) = 1$.
The $k$-th \defn{moment} of a distribution $\mu \in \Meas(\NN)$ is defined as $\ExpVal_{X \sim \mu}[X^k] := \sum_{x \in \NN} \mu(x) \cdot x^k$.
The \defn{tail} of a distribution refers to the region far away from the mean.
A distribution $\mu$ on $\NN$ has \defn{tail asymptotics} $f(n)$ if $\mu(n) = \Theta(f(n))$.
The $\Dirac(x)$ distribution has the mass function $\Dirac(x)(y) = [x = y]$; the $\Bernoulli(\rho)$ distribution for $\rho \in [0,1]$ is given by $\Bernoulli(\rho)(0) = 1 - \rho$ and $\Bernoulli(\rho)(1) = \rho$; the $\Uniform\{a,\dots,b\}$ distribution for $a \le b \in \NN$ by $\Uniform\{a,\dots,b\}(n) = \frac{[a \le n \le b]}{b - a + 1}$; and the $\Geometric(\rho)$ distribution for $\rho \in (0,1]$ by $\Geometric(\rho)(n) = (1 - \rho)^n \cdot \rho$.

\subsection{Probabilistic Programming Language}
\label{sec:ppl}

Our programming language is a simple imperative language with a fixed number of variables $\X = (X_1, \dots, X_n)$.
Each variable only takes values in $\NN$.
We consider the following minimal programming language where $P$ denotes \defn{programs} and $E$ denotes \defn{events}.
\begin{align*}
P & ::= \Skip \mid P_1;P_2 \mid X_k := 0 \mid X_k \passign a \mid X_k \massign 1 \mid \Ite{E}{P_1}{P_2} \mid \Whilst{E}{P_1} \mid \Fail \\
E & ::= X_k = a \mid \Flip(\rho) \mid \lnot E_1 \mid E_1 \land E_2
\end{align*}
where $a \in \NN, \rho \in [0, 1]$ are literals.
We explain the constructs briefly:
$\Skip$ does nothing; $P_1;P_2$ executes $P_1$ and then $P_2$; $X_k := 0$ sets $X_k$ to 0; $X_k \passign a$ increments $X_k$ by $a$; $X_k \massign 1$ decrements $X_k$ by 1 but clamped at 0 ($X_k := \max(X_k - 1, 0)$); $\Ite{E}{P_1}{P_2}$ executes $P_1$ if the event $E$ occurs and $P_2$ otherwise; $\Whilst{E}{P_1}$ repeats $P_1$ as long as $E$ occurs; $\Fail$ states a contradictory observation, i.e. $\Observe \False$.
The $\Flip(\rho)$ event occurs with probability $\rho$, like a coin flip with bias $\rho$ coming up heads, or more formally, sampling 1 from an independent $\Bernoulli(\rho)$ distribution.%
\footnote{
This construct is not usually considered an event because it does not correspond to a subset of the state space.
We preferred it to a separate sampling statement for convenience: it avoids auxiliary variables for conditionals and loops.
}
The logical operators $\lnot, \land$ denote the complement and intersection of events.
We usually assume that all variables are set to zero with probability 1 at the beginning.

\paragraph{Syntactic sugar}
The language is kept minimal to simplify the presentation.
The following syntactic sugar for events (on the left) and statements (on the right) will be convenient:
\[
\begin{aligned}[c]
  E_1 \lor E_2 &\; \leadsto \; \lnot (\lnot E_1 \land \lnot E_2) \\
  X_k \in \{a_1, \dots, a_m\} &\; \leadsto \; X_k = a_1 \lor \dots \lor X_k = a_m \\
  X_k < a &\; \leadsto \; X_k \in \{0, \dots, a - 1\} \\
  X_k \le a &\; \leadsto \; X_k < a + 1  \\
  X_k \ge a &\; \leadsto \; \lnot (X_k < a) \\
  X_k > a &\; \leadsto \; \lnot (X_k \le a)
\end{aligned}
\;
\begin{aligned}[c]
  X_k := c &\; \leadsto \; X_k := 0; X_k \passign c \\
  \ProbBranch{P_1}{\rho}{P_2} &\; \leadsto \; \Ite{\Flip(\rho)}{P_1}{P_2} \\
  X_k \sim \Bernoulli(\rho) &\; \leadsto \; \ProbBranch{X_k := 1}{\rho}{X_k := 0} \\
  \Observe E &\; \leadsto \; \Ite{E}{\Skip}{\Fail}
\end{aligned}
\]

\paragraph{Sampling}
The above language does not include a sampling construct, but sampling from the following distributions can be encoded easily:
all finite discrete distributions (e.g. Bernoulli, Binomial, Uniform, Categorical), the Geometric and Negative Binomial distributions, as well as shifted versions thereof.
Finite discrete distributions can be expressed with branching and the $\Flip(\rho)$ event, similarly to $\Bernoulli(\rho)$ shown above.
The sampling construct $X_k \sim \Geometric(\rho)$ can be expressed as $X_k := 0; \Whilst{\lnot \Flip(\rho)}{X_k \passign 1}$.
A negative binomial distribution can be expressed similarly, as a sum of i.i.d. geometrically distributed variables.
Sampling from other distributions, such as the Poisson distribution, can also be expressed in principle since our language is Turing-complete, but we are not aware of a simple encoding.

\paragraph{Expressivity}
Our language is similar to the \texttt{cReDiP} language in \citet{KlinkenbergBCHK24}, but with more restricted sampling.
Like \texttt{cReDiP}, our language does not support negative integers, continuous distributions/variables, more complex arithmetic like multiplication, or comparisons of two variables.
Note that the non-probabilistic fragment of the language is already Turing-complete because it can simulate a three-counter machine.
Thus these missing constructs could be encoded in principle, at the expense of program complexity.

\subsection{Standard Semantics}

The semantics of programs takes an input measure on the state space and yields an output measure.
For probabilistic programs without conditioning, the state space of programs would be $\NN^n$ \citep{KlinkenbergBKKM20,ZaiserMO23}.
But in the presence of conditioning, we want to track observations failures, so we add a \defn{failure state} $\lightning$, which signifies a failed observation.
If we just represented failures by setting the measure to 0 (as in \citet{ZaiserMO23}), we would not be able to distinguish between observation failures and nontermination (see \citet{KlinkenbergBCHK24} for a detailed discussion).
So the semantics of programs transforms measures on the \defn{state space} $\estates := \NN^n \cup \{ \lightning \}$.
The intuitive idea is that when actually running a probabilistic program, $\Fail$ moves to the failure state $\lightning$ and aborts the program.

\paragraph{Orders on measures and transformers}
To define the semantics, we first need to define partial orders on measures and measure transformers, both of which are standard \citep{Kozen81}.
Given two measures $\mu, \nu$ on $M$, we say that $\mu \mle \nu$ if and only if $\mu(S) \le \nu(S)$ for all sets $S \subseteq M$.
On measure transformers $\phi, \psi: \Meas(M) \to \Meas(M)$, we define the pointwise lifting of this order: $\phi \trafole \psi$ if and only if $\phi(\mu) \mle \psi(\mu)$ for all $\mu \in \Meas(M)$.
Both orders are $\omega$-complete partial orders.

\begin{figure}
\small
\begin{subfigure}{0.33\textwidth}
\begin{minipage}{\textwidth}
\begin{align*}
  \mu|_{X_k = a}(S) &:= \mu(\{ \x \in S \mid x_k = a \}) \\
  \mu|_{\Flip(\rho)} &:= \rho \cdot \mu|_{\NN^n} \\
  \mu|_{\lnot E_1} &:= \mu|_{\NN^n} - \mu|_{E_1} \\
  \mu|_{E_1 \land E_2} &:= (\mu|_{E_1})|_{E_2} \\
\end{align*}
\end{minipage}
\caption{Standard semantics of events}
\label{fig:std-event-semantics}
\end{subfigure}
\hfill
\begin{subfigure}{0.66\textwidth}
\begin{minipage}{\textwidth}
\begin{align*}
\esem{\Skip}(\mu) &= \mu \\
\esem{P_1; P_2}(\mu) &= \esem{P_2}(\esem{P_1}(\mu)) \\
\esem{X_k := 0}(\mu)(S) &= \mu|_\lightning(S) + \mu(\{ \x \in \NN^n \mid \x[k \mapsto 0] \in S \}) \\
\esem{X_k \passign a}(\mu)(S) &= \mu|_\lightning(S) + \mu(\{ \x \in \NN^n \mid \x[k \mapsto x_k + a] \in S \}) \\
\esem{X_k \massign 1}(\mu)(S) &= \mu|_\lightning(S) + \mu(\{ \x \in \NN^n \mid \x[k \mapsto x_k \monus 1] \in S \}) \\
\esem{\Ite{E}{P_1}{P_2}}(\mu) &= \mu|_\lightning + \esem{P_1}(\mu|_{E}) + \esem{P_2}(\mu|_{\lnot E}) \\
\esem{\Fail}(\mu) &= \mu(\estates) \cdot \Dirac(\lightning) \\
\esem{\Whilst{E}{P_1}} &= \lfp{\Phi_{E, P_1}}
\end{align*}
\end{minipage}
\caption{Standard semantics of statements\vspace{-0.5em}}
\label{fig:std-stmt-semantics}
\end{subfigure}
\caption{Standard semantics for probabilistic programs\vspace{-1em}}
\label{fig:std-semantics}
\end{figure}

\paragraph{Semantics}
The \defn{standard semantics} is a straightforward adaptation of \citet{Kozen81,KlinkenbergBCHK24}.
For events $E$, it describes how a measure $\mu$ on program states is ``restricted'' to $E$ (\cref{fig:std-event-semantics}), suggestively written $\mu|_E$ like the restriction of $\mu$ to a subset of $\estates$, even though some events (e.g. $\Flip(\rho)$) do not correspond to such a subset.
For programs $P$, the statement semantics $\esem{P}: \Meas(\estates) \to \Meas(\estates)$ describes how $P$ transforms the measure on program states (\cref{fig:std-stmt-semantics}).

Regarding notation, we write $\mu|_{\lightning}$ for $\mu|_{\{\lightning\}}$, i.e. the restriction of $\mu$ to the failure state $\lightning$, and conversely $\mu|_{\NN^n}$ for the restriction of $\mu$ to $\NN^n$, excluding the failure state $\lightning$.
In \cref{sec:geo-bounds}, we will largely ignore the failure state $\lightning$ and work with distributions on $\NN^n$.
For this purpose, we introduce the \defn{simplified standard semantics} $\sem{P}: \Meas(\NN^n) \to \Meas(\NN^n)$ defined as $\sem{P}(\mu) := \esem{P}(\mu)|_{\NN^n}$.
An interesting case is the loop construct, whose semantics is ${\esem{\Whilst{E}{P_1}}} := \lfp{\Phi_{E, P_1}}$ where
\begin{align*}
&\Phi_{E, P_1}: (\Meas(\estates) \to \Meas(\estates)) \to (\Meas(\estates) \to \Meas(\estates)) \\
&\Phi_{E, P_1}(\psi)(\mu) := \mu|_\lightning + \mu|_{\lnot E} + \psi(\esem{P_1}(\mu|_{E}))
\end{align*}
is the unrolling operator and $\lfp{\Phi_{E, P_1}}$ denotes its least fixed point with respect to $\trafole$.
Another way to look at it is that $W := \esem{\Whilst{E}{P_1}}$ is the least solution to the equation:
\[ W(\mu) = \mu|_\lightning + \mu|_{\lnot E} + W(\esem{P_1}(\mu|_{E})) \quad \forall \mu \in \Meas(\estates) \]

\paragraph{Conditioning}
The $\Fail$ construct moves all the mass to the failure state $\lightning$.
Over the course of the program, some of the initial probability mass gets lost due to nontermination and some moves to the failure state due to failed observations.
Ultimately, we are interested in the distribution \emph{conditioned} on the observations.
These conditional probabilities of $\X = \x \in \NN^n$ are defined as:
$\Prob[\X = \x \mid \X \ne \lightning] = \frac{\Prob[\X = \x \land \X \ne \lightning]}{\Prob[\X \ne \lightning]} = \frac{\Prob[\X = \x]}{1 - \Prob[X = \lightning]}$.
Thus we have to remove the mass on $\lightning$ from the subprobability measure $\mu$ of the program and \defn{normalize} it, which is achieved by the function $\Normalize$ turning subprobability measures on $\estates$ into subprobability measures on $\NN^n$: $\Normalize(\mu) := \frac{\mu - \mu|_{\lightning}}{1 - \mu(\lightning)}$.
Thus the posterior probability distribution of a program $P$ with initial distribution $\mu$ (usually $\Dirac(\zero_n)$) is given by $\Normalize(\esem{P}(\mu))$.
Operationally, we can think of normalization as a rejection sampler of the posterior distribution: it runs the program repeatedly, rejects all runs that end in $\lightning$, and only yields the results of the remaining runs as samples.

\paragraph{Nontermination}
Due to the possibility of nontermination, even the normalized measure of a program may not be a probability distribution, only a subprobability distribution (see the discussion in \citet[Section 3.4]{KlinkenbergBCHK24}).
For example, the loop $\Whilst{\Flip(1)}{\Skip}$ will never terminate, so both its unnormalized and normalized semantics are always the zero measure.
This is not a problem in practice because nontermination is usually considered a bug for statistical models.
Tracking observation failures is useful for defining termination, however.
\begin{definition}
A program $P$ is \defn{almost surely terminating} on a \emph{finite} measure $\mu \in \Meas(\estates)$ if $\esem{P}(\mu)(\estates) = \mu(\estates)$.
\end{definition}

The semantics of a program $\esem{P}$ satisfies the usual properties listed below. The proof is analogous to \citet{KlinkenbergBCHK24}.

\begin{lemma}
\label{lem:stdsem-properties}
For any program $P$, the transformation $\esem{P}$ is linear and $\omega$-continuous, so in particular monotonic.
Also, the total measure does not increase: $\esem{P}(\mu)(\estates) \le \mu(\estates)$ for all $\mu \in \Meas(\estates)$.
\end{lemma}

\section{Residual Mass Semantics}
\label{sec:unrolling-bounds}

In this section, we first present the \emph{lower bound semantics} based on loop unrolling.
Since even upper bounds on the normalized posterior require lower bounds on the normalizing constant, these lower bounds are also used as part of our methods for upper bounds.
The \emph{residual mass semantics} extends the lower bound semantics to obtain upper bounds on the unnormalized distribution as well and can thus derive both lower and upper bounds on the normalized posterior.
This way, any exact inference method for the loop-free fragment can be transformed into a method to bound the distribution of programs with loops.
We present both semantics formally and prove soundness and convergence.
Full proofs are given in \apxref{apx:unrolling-bounds}.

\subsection{Lower Bounds via Unrolling}

Lower bounds can be computed by unrolling the loop a finite number of times and discarding the part of the distribution that has not exited the loop after the unrollings.
Define the \defn{$u$-fold unrolling} $P^{(u)}$ of a program $P$ by unrolling each loop $u$ times:
\begin{align*}
  (P_1; P_2)^{(u)} &:= P_1^{(u)}; P_2^{(u)} \\
  (\Ite{E}{P_1}{P_2})^{(u)} &:= \Ite{E}{P_1^{(u)}}{P_2^{(u)}} \\
  (\Whilst{E}{P})^{(u)} &:= (\Whilst{E}{P})_0^{(u)} \\
  (\Whilst{E}{P})_v^{(u)} &:= \Ite{E}{P^{(u)}; (\Whilst{E}{P})_{v+1}^{(u)}}{\Skip}; \quad \text{for } v < u \\
  (\Whilst{E}{P})_u^{(u)} &:= \Whilst{E}{P^{(u)}} \\
  P^{(u)} &:= P \quad \text{otherwise}
\end{align*}

\begin{restatable}{lemma}{UnrollingSemantics}
\label{lem:unrolling-semantics}
Unrolling does not change the semantics: $\esem{P} = \esem{P^{(u)}}$ for all programs $P$ and $u \in \NN$. 
\end{restatable}

The \defn{lower bound semantics $\semlo{P}$} is then defined just like $\esem{P}$ except that $\semlo{\Whilst{E}{P_1}}(\mu) := \zero$.
In other words, $\semlo{P} = \esem{P'}$ where $P'$ is the program $P$ where all loops are replaced by infinite loops, effectively setting the measure to zero.
Its correctness follows from the monotonicity of the standard semantics.
\begin{restatable}[Soundness of lower bounds]{theorem}{SoundnessLowerBounds}
\label{thm:soundness-lower-bounds}
For all measures $\mu \in \Meas(\estates)$ and programs $P$, we have $\semlo{P}(\mu) \mle \esem{P}(\mu)$.
\end{restatable}

The lower bounds will get better as loops are unrolled further and further.
In fact, they will converge to the true distribution.
\begin{restatable}[Convergence of lower bounds]{theorem}{ConvergenceLowerBounds}
\label{thm:convergence-lower-bounds}
For all \emph{finite} measures $\mu \in \Meas(\estates)$ and programs $P$, the lower bound $\semlo{P^{(u)}}(\mu)$ converges (in total variation distance) monotonically to the true distribution $\esem{P}(\mu)$ as $u \to \infty$.
\end{restatable}

Note that the lower bound semantics involves only finite discrete distributions (assuming all variables are initially zero), which can be represented exactly as an array of the probability masses.
In fact, any exact semantics for the loop-free fragment leads to lower bounds for programs with loops.
A related approach by \citet{JansenDKKW16} combines unrolling with probabilistic model checking to obtain lower bounds on reachability probabilities and expectations.

\subsection{Upper Bounds via Residual Mass}

Lower bounds are easy because $\zero$ is clearly a lower bound that can be used as a starting point for the fixed point iteration.
This strategy does not work for upper bounds since there is no such obvious starting point.

Instead, a simple idea is to use the fact that the total mass of the state distribution can only decrease after program execution (\cref{lem:stdsem-properties}).
In other words, for a program $P$ with initial distribution $\mu \in \Meas(\estates)$, the total mass of the distribution after running $P$ is $\esem{P}(\mu)(\estates) \le \mu(\estates)$.
So the part of the distribution that we miss in the lower bounds by cutting loops short has a probability mass bounded by the distance to the total mass $\mu(\estates)$ at the start.
We call this gap the \defn{residual mass} of the program $P$ with initial measure $\mu$ and define it as $\semres{P}(\mu) := \mu(\estates) - \semlo{P}(\mu)(\estates) \in \nnr$.

The \defn{residual mass semantics} uses the residual mass to bound probabilities $\esem{P}(\mu)(S)$ of the program distribution for a set $S \subseteq \estates$ by analyzing the gap $\esem{P}(\mu)(S) - \semlo{P}(\mu)(S)$ as follows.
The gap is maximized if $S = \estates$ and thus bounded by $\esem{P}(\mu)(\estates) - \semlo{P}(\mu)(\estates)$.
Since $\esem{P}(\mu)(\estates) \le \mu(\estates)$ by \cref{lem:stdsem-properties}, this gap is always bounded by the residual mass.
The residual mass semantics can also bound the normalizing constant $1 - \esem{P}(\mu)(\lightning)$ since $\semlo{P}(\mu)(\lightning) \le \esem{P}(\mu)(\lightning) \le \semlo{P}(\mu)(\lightning) + \semres{P}(\mu)$ where the lower bound follows from \cref{thm:soundness-lower-bounds} and the upper bound from the previous argument.
Thus it also yields bounds on the (normalized) posterior probabilities.
The above soundness argument is made fully rigorous in the following theorem.
We also have a convergence theorem for the residual mass semantics, similar to \cref{thm:convergence-lower-bounds}.

\begin{restatable}[Soundness of the residual mass semantics]{theorem}{SoundnessResidualMassBounds}
\label{thm:soundness-residual-mass-bounds}
Let $P$ be a program and $\mu \in \Meas(\estates)$.
Then for all $S \subseteq \estates$, we can bound the unnormalized probabilities:
\[ \semlo{P}(\mu)(S) \le \esem{P}(\mu)(S) \le \semlo{P}(\mu)(S) + \semres{P}(\mu) \]
If $\mu$ is a probability measure and $S \subseteq \NN^n$, we can bound the posterior probabilities:
\[ \frac{\semlo{P}(\mu)(S)}{1 - \semlo{P}(\mu)(\lightning)} \le \Normalize(\esem{P}(\mu))(S) \le \frac{\semlo{P}(\mu)(S) + \semres{P}(\mu)}{1 - \semlo{P}(\mu)(\lightning) - \semres{P}(\mu)} \]
\end{restatable}

\begin{restatable}[Convergence of the residual mass semantics]{theorem}{ConvergenceResidualMassBounds}
\label{thm:convergence-residual-mass-bounds}
Let $P$ be a program terminating almost surely on a \emph{finite} measure $\mu \in \Meas(\estates)$.
Then the residual mass $\semres{P^{(u)}}(\mu)$ converges monotonically to 0 as $u \to \infty$.
In particular, the above bounds on $\Normalize(\esem{P^{(u)}}(\mu))(S)$ converge to the true posterior probability as $u \to \infty$.
\end{restatable}

Since the residual mass semantics only depends on the lower bounds, it can be implemented using the same representations and techniques as the lower bound semantics.
In particular, it also involves only finite discrete distributions, representable as arrays of probability masses.
We illustrate the residual mass semantics with an example.

\begin{example}[Residual mass semantics]
Consider again the ``die paradox'' program (\cref{ex:die-paradox,fig:die-paradox}) where the variables are $\X := (\mathit{Throws}, \mathit{Die})$.
After three unrollings, the lower bound distribution $\mu := \semlo{P^{(3)}}(\mu_0)$ for the initial distribution $\mu_0 = \Dirac(\zero)$ is given by
\[ \mu(\x) = \begin{cases}
\frac{1}{2} + \frac{1}{3} \cdot \frac{1}{2} = \frac{2}{3} & \text{if } \x = \lightning \\
\frac{1}{6} & \text{if } \x = (1, 6) \\
\frac{1}{3} \cdot \frac{1}{6} = \frac{1}{18} & \text{if } \x = (2, 6) \\
0 & \text{otherwise}
\end{cases} \]
because the failure probability is $\frac{1}{2}$ (for $\mathit{Die} \in \{1, 3, 5\}$) and the second iteration is entered in a state other than $\lightning$ with probability $\frac{1}{3}$ (for $\mathit{Die} \in \{2, 4\}$).
The residual mass is given by $\semres{P^{(3)}}(\mu_0) = \mu_0(\estates) - \mu(\estates) = 1 - (\frac{2}{3} + \frac{1}{6} + \frac{1}{18}) = \frac{1}{9}$.
This yields the following bounds on $\esem{P}(\mu_0)(s)$ for $s \in \estates$:
\[ \esem{P}(\mu_0)(\x) \in \begin{cases}
[\frac{2}{3}, \frac{7}{9}] & \text{if } \x = \lightning \\
[\frac{1}{6}, \frac{5}{18}] & \text{if } \x = (1, 6) \\
[\frac{1}{18}, \frac{1}{6}] & \text{if } \x = (2, 6) \\
[0, \frac{1}{9}] & \text{otherwise}
\end{cases} \]
Hence the normalizing constant $1 - \esem{P}(\mu_0)$ is in $[\frac{2}{9}, \frac{1}{3}]$, which yields the following normalized bounds.
They are not very tight but can be improved by increasing the unrolling depth.
\[ \Normalize(\esem{P}(\mu_0))(\x) \in \begin{cases}
[\frac{1}{6} \cdot 3, \frac{5}{18} \cdot \frac{9}{2}] = [\frac{1}{2}, \frac{5}{4}] & \text{if } \x = (1, 6) \\
[\frac{1}{18} \cdot 3, \frac{1}{6} \cdot \frac{9}{2}] = [\frac{1}{6}, \frac{3}{4}] & \text{if } \x = (2, 6) \\
[0, \frac{1}{9} \cdot \frac{9}{2}] = [0, \frac{1}{2}] & \text{otherwise}
\end{cases} \]
\end{example}

\section{Geometric Bound Semantics}
\label{sec:geo-bounds}

A big problem with the residual mass semantics is that we can bound the total measure and individual probabilities, but not moments and tail distributions of distributions with infinite support.
This problem cannot be fully resolved in general because moments do not always exist.

\begin{example}
\label{ex:power-of-two}
The moments of the distribution of the following program are infinite:
\begin{align*}
&X := 0; Y := 1; \\
&\Whilst{\Flip(1/2)}{ \\
&\qquad \Whilst{Y > 0}{X \passign 1; Y \massign 1}; \\
&\qquad \Whilst{X > 0}{X \massign 1; Y \passign 2}; }
\end{align*}
The loop body multiplies $Y$ by 2 using the auxiliary variable $X$.
This program is almost surely terminating and with result probabilities $\Prob[X = 0, Y = 2^m] = 2^{-m}$ for $m > 0$ and zero elsewhere.
Hence the $k$-th moment of $Y$ is $\ExpVal[Y^k] = \sum_{m=1}^\infty \Prob[X = 0, Y = 2^m] \cdot 2^{km} \ge \sum_{m=1}^{\infty} 1 = \infty$ for $k \ge 1$.
\end{example}

However, in many cases, more precise upper bounds can be found that also yield bounds on moments and tails.
There are two key ideas that make this \emph{geometric bound semantics} possible: contraction invariants and eventually geometric distributions (EGDs).
Like the residual mass semantics, it builds on the lower bound semantics to obtain bounds on the normalized posterior.

\subsection{Contraction Invariants}
\label{sec:contraction-invariants}

If a loop $L = \Whilst{E}{B}$ terminates almost surely on an initial distribution $\mu$, all the probability mass that entered the loop must eventually exit the loop.
Inspired by this observation, we consider the following assumption:
\begin{equation}
\esem{B}(\mu|_E) \mle c \cdot \mu \text{ with } 0 \le c < 1 \label{eq:contraction-assumption}
\end{equation}
In words: the probability distribution at the start of the next loop iteration decreases uniformly by a factor of $c$.
This assumption looks very strong and is often violated in practice, but if we assume for a minute that it holds, what can we derive?
Recall the fixpoint equation of loops:
\begin{align*}
\esem{\Whilst{E}{B}}(\mu) &= \mu|_\lightning + \mu|_{\lnot E} + \esem{\Whilst{E}{B}}(\esem{B}(\mu|_E)) \\
&\mle \mu|_\lightning + \mu|_{\lnot E} + c \cdot \esem{\Whilst{E}{B}}(\mu) &\text{by \cref{eq:contraction-assumption} and linearity}
\end{align*}
By rearranging, we find the upper bound $\esem{\Whilst{E}{B}}(\mu) \mle \frac{\mu|_\lightning + \mu|_{\lnot E}}{1 - c}$.

Unfortunately, the initial distribution $\mu$ will usually not satisfy \cref{eq:contraction-assumption} directly.
But we may be able to increase $\mu$ to $\nu \mge \mu$ such that $\nu$ satisfies $\esem{B}(\nu|_E) \mle c \cdot \nu$ with $0 \le c < 1$.
We then call $\nu$ a \defn{contraction invariant} or \defn{$c$-contraction invariant} of $\mu$.%
\footnote{The name is inspired by contraction mappings in mathematics where they describe functions $f: X \to Y$ of metric spaces with $d(f(x), f(y)) \le c \cdot d(x,y)$ with $c < 1$.}
Surprisingly, such a contraction invariant $\nu$ can often be found in practice (see \cref{sec:eval-applicability}), and we can then derive the following upper bound on the distribution after the loop: $\esem{\Whilst{E}{B}}(\mu) \mle \esem{\Whilst{E}{B}}(\nu) \mle \frac{\nu|_\lightning + \nu|_{\lnot E}}{1 - c}$.

\subsection{Eventually Geometric Distributions (EGDs)}
\label{sec:egds}

How can we find such a contraction invariant $\nu$?
We clearly have to restrict the candidate set for $\nu$ in some way.
Specifically, we consider a class of distributions that generalize geometric distributions.
Recall that a $\Geometric(\beta)$ distribution has probability masses $p(i) = \beta (1 - \beta)^i \propto (1 - \beta)^i$ with a decay rate of $1 - \beta$.
For the multivariate setting, we consider products of independent geometric distributions, but with a twist: probabilities of small values are allowed to differ for additional flexibility.
This section makes extensive use of vector and tensor notations as described in \cref{sec:conventions}.

\begin{definition}
Let $\P \in \nnr^{[d_1 + 1] \times \cdots \times [d_n + 1]}$ be an $n$-dimensional tensor and $\balpha \in [0,1)^n$.
The $n$-dimensional \defn{eventually geometric distribution (EGD)} with \defn{initial block} $\P$ and \defn{decay rates} $\balpha$, written $\egd{\P}{\balpha} \in \Meas(\NN^n)$ is defined by the following mass function (for $\ii = (i_1,\dots,i_n) \in \NN^n$):
\[ \egd{\P}{\balpha}(\ii) = \P_{\min(\ii, |\P| - \one_n)} \cdot \balpha^{\max(\ii - |\P| + \one_n, \zero_n)} = \P_{\min(i_1, d_1), \dots, \min(i_n, d_n)} \alpha_1^{\max(i_1 - d_1, 0)} \cdots \alpha_n^{\max(i_n - d_n, 0)} \]
\end{definition}

Note that an EGD is not necessarily a probability measure.
One can think of $\egd{\P}{\balpha}$ as a scaled product of independent $\Geometric(1 - \alpha_i)$ distributions, except it may differ in a finite ``prefix'' of size $d_j$ in each dimension $j$.
Outside this initial block, the probability masses are extended like in a geometric distribution with decay rate $\alpha_j$ in each dimension $j$.
In other words: \emph{eventually}, for values $i_j > d_j$, the $j$-th component of $\egd{\P}{\balpha}$ behaves like a $\Geometric(1 - \alpha_j)$ distribution; hence the name.
EGDs were originally motivated by the shape of their generating function \citep[Appendix C.1]{Zaiser24}.
They can be seen as a subclass of multivariate discrete phase-type distributions \citep{Campillo18}.
(Discrete phase-type distributions \citep{Neuts75} describe the absorption time in a Markov chain with one absorbing state \citep[Section 1.2.6]{BladtN17}.)

\begin{example}
\label{example:egd}
Consider the two-dimensional EGD given by $\egd{\P}{(\alpha_1, \alpha_2)}$ with $\P \in \nnr^{[2] \times [2]}$, i.e. $\P$ is a matrix
$\P = \begin{pmatrix}
  \P_{0,0} & \P_{0,1} \\ \P_{1,0} & \P_{1,1}
\end{pmatrix}$.
Its probability masses are given in the following table:
\begin{center}
  \small
  \begin{tabular}{r|lllll}
    $\egd{\P}{(\alpha_1, \alpha_2)}(x_1, x_2)$ & $x_2 = 0$ & $x_2 = 1$ & $x_2 = 2$ &\dots & $x_2 = j$ \\
    \hline
    $x_1 = 0$ & $\P_{0,0}$ & $\P_{0,1}$ & $\P_{0,1} \cdot \alpha_2$ & \dots & $\P_{0,1} \cdot \alpha_2^{j - 1}$ \\
    $x_1 = 1$ & $\P_{1,0}$ & $\P_{1,1}$ & $\P_{1,1} \cdot \alpha_2$ & \dots & $\P_{1,1} \cdot \alpha_2^{j-1}$ \\
    $x_1 = 2$ & $\P_{1,0} \cdot \alpha_1$ & $\P_{1,1} \cdot \alpha_1$ & $\P_{1,1}  \cdot\alpha_1 \alpha_2$ & \dots & $\P_{1,1} \cdot \alpha_1 \cdot \alpha_2^{j-1}$ \\
    $\vdots$ & $\vdots$ & $\vdots$ & $\vdots$ & $\ddots$ & $\vdots$ \\
    $x_1 = i$ & $\P_{1,0} \cdot \alpha_1^{i-1}$ & $\P_{1,1} \cdot \alpha_1^{i-1}$ & $\P_{1,1} \cdot \alpha_1^{i-1} \cdot \alpha_2$ & \dots & $\P_{1,1} \cdot \alpha_1^{i-1} \cdot \alpha_2^{j-1}$
  \end{tabular}
\end{center}
Here we can see that after the initial block of size $2 \times 2$, the probability masses are extended with decay rate $\alpha_1$ in the first dimension and $\alpha_2$ in the second dimension.
\end{example}

Why did we pick EGDs as the shape of our bounds? There are several reasons.
\begin{enumerate}
  \item \emph{Interpretability:} EGDs can easily be understood as geometric distributions, where the probability masses for small values (up to some threshold) have been modified.
  The relationship with the geometric distribution also makes it easier to compute its moments and the tail asymptotics can be read directly off the parameter $\balpha$ (see \cref{thm:egd-stats} below).
  \item \emph{Expressiveness:} the reason we allow the ``start'' of the distribution to deviate from geometric distributions is necessary for flexibility.
  If only exact geometric distributions were allowed, typical program operations like increasing or decreasing a variable (i.e. shifting the geometric distribution) could not be represented precisely enough and the whole approach would fail.
  \item \emph{Tractability:} a sufficient condition for the order $\mle$ is easy to check for EGDs, as we will see, because the probability masses follow a simple pattern.
  In fact, this condition is equivalent to the satisfiability of a system of polynomial inequalities.
  If we had based our semantics on an extension of, say, negative binomial distributions (of which geometric distributions are a special case), deciding the order would be harder because binomial coefficients start appearing in the probability masses.
\end{enumerate}

As evidence for interpretability and tractability, we show how to marginalize EGDs and how to compute their moments.
This is also needed to extract bounds on the moments from EGD bounds.
\begin{restatable}[Marginalizing EGDs]{lemma}{EgdMarginalize}
\label{lem:egd-marginalize}
Marginalizing out the $k$-th dimension from $\egd{\P}{\balpha}$ yields $\egd{\Q}{\bbeta}$ with $\bbeta = (\alpha_1, \dots, \alpha_{k - 1}, \alpha_{k + 1}, \dots \alpha_n)$ and $\Q = \sum_{j=0}^{|\P|_k-2} \P_{k:j} + \frac{\P_{k: |\P|_k - 1}}{1 - \alpha_k}$.
\end{restatable}

\begin{restatable}[Moments and tails of EGDs]{theorem}{EgdMoments}
\label{thm:egd-stats}
Let $\egd{\P}{\alpha}$ be a one-dimensional EGD with $\P \in \nnr^{[d+1]}$.
Then its tail asymptotics is $\egd{\P}{\alpha}(n) = O(\alpha^n)$ and its $k$-th moment is
\[ \ExpVal_{X \sim \egd{\P}{\alpha}}[X^k] = \sum_{j = 0}^{d - 1} \P_j \cdot j^k + \frac{\P_d}{1 - \alpha} \sum_{i = 0}^k \binom{k}{i} d^{k - i} \ExpVal_{Y \sim \Geometric(1 - \alpha)}[Y^i] \]
and can thus be computed from the $k$-th moment of a geometric distribution.
In particular, the expected value is $\ExpVal_{X \sim \egd{\P}{\alpha}}[X] = \sum_{j = 0}^{d - 1} \P_j \cdot j + \frac{\P_d}{1 - \alpha}\left(d + \frac{\alpha}{1 - \alpha}\right)$.
\end{restatable}

\paragraph{Expansion of EGDs}
When performing binary operations (e.g. comparison or addition) on EGDs, it will often be convenient to assume that their initial blocks have the same size.
In fact, one can always expand the size of the initial block of an EGD without changing the distribution.
For example, the distribution from \cref{example:egd} can equivalently be represented as:
\begin{align*}
  \egd{\begin{pmatrix}\P_{0,0} & \P_{0,1} \\ \P_{1,0} & \P_{1,1}\end{pmatrix}}{\balpha} &= \egd{\begin{pmatrix}
    \P_{0,0} & \P_{0,1} \\ \P_{1,0} & \P_{1,1} \\ \P_{1,0} \alpha_1 & \P_{1,1} \alpha_1
  \end{pmatrix}}{\balpha} = \egd{\begin{pmatrix}
    \P_{0,0} & \P_{0,1} & \P_{0,1}\alpha_2 & \P_{0,1} \alpha_2^2 \\ \P_{1,0} & \P_{1,1} & \P_{1,1} \alpha_2 & \P_{1,1} \alpha_2^2
  \end{pmatrix}}{\balpha}
\end{align*}
More generally, given an $\egd{\P}{\balpha}$, an \defn{expansion} $\egd{\Q}{\balpha}$ with $|\P| \le |\Q|$ (meaning $|\P|_i \le |\Q|_i$ for all $i = 1, \dots, n$) is obtained by adding rows and columns with the appropriate factors of $\alpha_i$, as follows:
\begin{align*}
  \Q_{\ii} &= \P_{\min(\ii, |\P| - \one_n)} \balpha^{\max(\ii - |\P| + \one_n, \zero_n)} = \P_{\min(i_1, |\P|_1 - 1), \dots, \min(i_n, |\P|_n - 1)} \alpha_1^{\max(i_1 - |\P|_1 + 1, 0)} \cdots \alpha_n^{\max(i_n - |\P|_n + 1, 0)}
\end{align*}

\paragraph{Order of EGDs}
To decide $\egd{\P}{\balpha} \mle \egd{\Q}{\bbeta}$ for $\P$ and $\Q$ of the same size $|\P| = |\Q|$, one might hope that this would be equivalent to $\P \le \Q$ and $\balpha \le \bbeta$ (both elementwise), because such a simple conjunction of inequalities would be easy to check.
To extend this idea to EGDs $\egd{\P}{\balpha}$, $\egd{\Q}{\bbeta}$ of different sizes, we first have to expand them to the same size $\max(|\P|, |\Q|)$.
Inlining the definition of expansion, we obtain the following order.

\begin{definition}[Order of EGDs]
The order $\egd{\P}{\balpha} \egdle \egd{\Q}{\bbeta}$ is defined to hold if and only if for all $\ii < \max(|\P|, |\Q|)$ (using the notation explained in \cref{sec:conventions}):
\[ \balpha \le \bbeta \; \land \; \P_{\min(\ii, |\P| - \one_n)} \balpha^{\max(\ii - |\P| + \one_n, \zero_n)} \le \Q_{\min(\ii, |\Q| - \one_n)} \bbeta^{\max(\ii - |\Q| + \one_n, \zero_n)} \]
\end{definition}

\begin{example}
\label{example:egd-order}
Consider the $\egd{\P}{\balpha}$ from \cref{example:egd} with $|\P| = (2, 2)$ and the two-dimensional $\egd{\Q}{\bbeta}$ with $|\Q| = (1, 4)$ and $\Q = \begin{pmatrix} \Q_{0,0} & \Q_{0,1} & \Q_{0, 2} & \Q_{0,3} \end{pmatrix}$.
Then $\egd{\P}{\balpha} \egdle \egd{\Q}{\bbeta}$ is equivalent to the following system of inequalities:
\begin{align*}
\alpha_1 &\le \beta_1 & \P_{0,0} &\le \Q_{0,0} & \P_{0,1} &\le \Q_{0,1} & \P_{0,1} \cdot \alpha_2 &\le \Q_{0,2} & \P_{0,1} \cdot \alpha_2^2 &\le \Q_{0,3} \\
\alpha_2 &\le \beta_2 & \P_{1,0} &\le \Q_{0,0} \cdot \beta_1 & \P_{1,1} &\le \Q_{0,1} \cdot \beta_1 & \P_{1,1} \cdot \alpha_2 &\le \Q_{0,2} \cdot \beta_1 & \P_{1,1} \cdot \alpha_2^2 &\le \Q_{0,3} \cdot \beta_1
\end{align*}
\end{example}

Unfortunately, this order $\egdle$ is not exactly the same as $\mle$ because zero coefficients at the edge of $\P$ can make $\balpha$ irrelevant, e.g. $\egd{0}{\alpha} = \zero$ for any $\alpha \in [0,1)$.
Thus $\zero = \egd{0}{\alpha} \mle \egd{1}{0} = \Dirac(0)$ even though $\egd{0}{\alpha} \not\egdle \egd{1}{0}$ for $\alpha > 0$.
However, such counterexamples are rare in practice, and $\egdle$ does imply $\mle$, as the following lemma shows.
As a consequence, working with the simpler order $\egdle$ instead of $\mle$ is sufficient to establish bounds.

\begin{restatable}{lemma}{EgdOrder}
If $\egd{\P}{\balpha} \egdle \egd{\Q}{\bbeta}$, then $\egd{\P}{\balpha} \mle \egd{\Q}{\bbeta}$.
\end{restatable}

\subsection{Semantics}
\label{sec:geo-bound-semantics}

Using the ideas from the previous sections, we define a compositional semantics $\semgeo{P}$ for upper bounds on the unnormalized distribution of $P$.
It is called \defn{geometric bound semantics} because it operates on eventually geometric distribution bounds due to their desirable properties mentioned in \cref{sec:egds}.
It also turns out that EGDs are closed under many operations we require: restricting to events, marginalizing, and adding or subtracting constants from variables.

The semantics $\semgeo{P}$ operates on distributions on $\NN^n$, not $\estates$.
We do not track the failure state $\lightning$ in this semantics, because the geometric bound semantics is only applicable to almost surely terminating programs (see \cref{thm:necessary-runtime}), so there is no need to distinguish between observation failure and nontermination.
Like in the residual mass semantics, we can bound the normalizing constant with the help of the lower bound semantics (see \cref{thm:soundness-geo-bounds}).

\paragraph{Relational semantics}
Since there may be more than one upper bound, $\semgeo{P}$ is not a function, but a binary \emph{relation} on EGDs.
The idea is that $(\egd{\P}{\balpha}, \egd{\Q}{\bbeta}) \in {\semgeo{P}}$, also written $\egd{\P}{\balpha} \semgeo{P} \egd{\Q}{\bbeta}$, ensures that $\sem{P}(\egd{\P}{\balpha}) \mle \egd{\Q}{\bbeta}$.
If this $\egd{\Q}{\bbeta}$ is unique for a given $\egd{\P}{\balpha}$, we will use function notation: ${\semgeo{P}}(\egd{\P}{\balpha}) = \egd{\Q}{\bbeta}$.
Like $\egdle$, the relation $\semgeo{P}$ depends on the \emph{representations} $(\P, \balpha), (\Q, \bbeta)$ of the involved EGDs $\egd{\P}{\balpha}$ and $\egd{\Q}{\bbeta}$, not just their measures.

\begin{figure}
\centering
\small
\alternaterowcolors
\begin{tabularx}{0.8\textwidth}{lX}
\toprule
Event $E$ &Restriction $\evgeo{\egd{\P}{\balpha}}{E}$ of $\egd{\P}{\balpha}$ to the event $E$ \\
\midrule
$X_k = a$ &
  $\egd{\Q}{\balpha[k \mapsto 0]}$ \newline
  where $|\Q| = |\P|[k \mapsto \max(|\P|_k, a + 2)]$ \newline
  with $\Q_{k:j} = \P_{k: \min(j, |\P|_k - 1)} \cdot \alpha_k^{\max(0, j - |\P|_k + 1)} \cdot [j = a]$ for $j = 0, \dots, |\Q|_k - 1$ \\[0.2em]
$\Flip(\rho)$ &$\egd{\rho \cdot \P}{\balpha}$ \\[0.2em]
$\lnot E_1$ &
  $\egd{\Q}{\balpha}$ \newline
  where $\evgeo{\egd{\P}{\balpha}}{E_1} = \egd{\R}{\bgamma}$ \newline
  $|\Q| = |\R|$ and $\Q_{\ii} = \P_{\min(\ii, |\P| - \one_n)} \cdot \balpha^{\max(\ii - |\P| + \one_n, \zero_n)} - \R_{\ii} \quad \forall \ii \le |\Q|$ \\[0.2em]
$E_1 \land E_2$ &$\evgeo{\left(\evgeo{\egd{\P}{\balpha}}{E_1}\right)}{E_2}$ \\
\bottomrule
\end{tabularx}\vspace{-0.5em}
\caption{Geometric bound semantics for events\vspace{-1em}}
\label{fig:egd-event-semantics}
\end{figure}

\paragraph{Event semantics}
The event semantics $\evgeo{\egd{\P}{\balpha}}{E}$ computes an EGD representation of the standard event semantics $\egd{\P}{\balpha}|_E$ (\cref{fig:egd-event-semantics}).
The restriction $\egd{\Q}{\bbeta} := \evgeo{\egd{\P}{\balpha}}{E}$ is generally computed by expanding the initial block $\P$ (if necessary) and then zeroing its entries not corresponding to the event $E$.
An important property is that $\egd{\Q}{\bbeta} = \egd{\Q}{\balpha}$.

For the event $X_k = a$, we have to set the coefficients of $\P_{k:j}$ to zero for $j \ne a$.
For this, we first have to expand $\P$ to $\Q$ with $|\Q|_k = a + 2$ if necessary.
Then we set $\Q_{k:j} = \zero$ for $j \ne a$ to ensure that all coefficients different from $a$ are zero, and also set $\alpha_k = 0$.
For instance, the EGD from \cref{example:egd} restricted to the event $X_2 = 2$ is
\[ \evgeo{\egd{\begin{pmatrix}
  \P_{0,0} & \P_{0,1} \\
  \P_{1,0} & \P_{1,1}
\end{pmatrix}}{\balpha}}{X_2 = 2} = \egd{\begin{pmatrix}
  0 & 0 & \P_{0,1} \cdot \alpha_2 & 0 \\
  0 & 0 & \P_{1,1} \cdot \alpha_2 & 0
\end{pmatrix}}{(\alpha_1, 0)} \]
Next, the event $\Flip(\rho)$ is independent of the current distribution, so we can simply multiply the initial block $\P$ by the probability $\rho$.
For the complement $\lnot E_1$, we first compute the restriction $\evgeo{\egd{\P}{\balpha}}{E_1} = \egd{\R}{\bgamma}$.
Since $\egd{\R}{\bgamma} = \egd{\R}{\balpha}$ by the property mentioned above, we can compute its complement $\egd{\Q}{\balpha} := \egd{\P}{\balpha} - \egd{\R}{\balpha}$ by first expanding $\egd{\P}{\balpha}$ to $\egd{\P'}{\balpha}$ with $\P'$ of the same size as $\R$, and then subtracting $\Q := \P' - \R$.
So the result is $\egd{\Q}{\balpha}$.
Continuing the above example with the event $\lnot(X_2 = 2)$, we get:
\[ \evgeo{\egd{\P}{\balpha}}{\lnot(X_2 = 2)} = \egd{\P}{\balpha} - \evgeo{\egd{\P}{\balpha}}{X_2 = 2} = \egd{\begin{pmatrix}
  \P_{0,0} & \P_{0,1} & 0 & \P_{0,1} \cdot \alpha_2^2 \\
  \P_{1,0} & \P_{1,1} & 0 & \P_{1,1} \cdot \alpha_2^2
\end{pmatrix}}{\balpha} \]
Finally, the restriction to an intersection $E_1 \land E_2$ is computed by restricting to $E_1$ and then to $E_2$.

\begin{figure}
\small
\alternaterowcolors
\begin{tabularx}{\textwidth}{lX}
\toprule
Statement $P$ &Constraints for $\egd{\P}{\balpha} \semgeo{P} \egd{\Q}{\bbeta}$ \\
\midrule
$\Skip$ &$\bbeta = \balpha \; \land \; \Q = \P$ \\[0.2em]
$P_1; P_2$ &$ \exists \R, \bgamma \ldotp \; \egd{\P}{\balpha} \semgeo{P_1} \egd{\R}{\bgamma} \; \land \; \egd{\R}{\bgamma} \semgeo{P_2} \egd{\Q}{\bbeta}$ \\[0.2em]
$X_k := 0$ &
  $\displaystyle \bbeta = \balpha[k \mapsto 0] \; \land \; \Q_{k: 0} = \frac{\P_{k: |\P|_k - 1}}{1 - \alpha_k} + \sum_{j=0}^{|\P|_k - 2} \P_{k: j} \; \land \; \Q_{k:1} = \zero$ \newline
  where $|\Q| = |\P|[k \mapsto 2]$ \\[0.2em]
$X_k \passign a$ &
  $\displaystyle \bbeta = \balpha \; \land \; \bigwedge_{j = 0}^{|\Q|_k - 1} \Q_{k: j} = \begin{dcases}
    0 &\text{if } j < a \\
    \P_{k: j - a} &\text{else} \\
  \end{dcases}$ \newline
  where $|\Q| = |\P|[k \mapsto |\P|_k + a]$ \\[0.2em]
$X_k \massign 1$ &
  $\displaystyle \bbeta = \balpha \; \land \; \bigwedge_{j = 0}^{|\Q|_k - 1} \Q_{k: j} = \begin{dcases}
    \P_{k: 0} + \P_{k: \min(|\P|_k - 1, 1)} \cdot \alpha^{\max(2 - |\P|_k, 0)} &\text{if } j = 0 \\
    \P_{k: \min(|\P|_k - 1, j + 1)} \cdot \alpha^{\max(j + 2 - |\P|_k, 0)} &\text{else}
  \end{dcases}$ \newline
  where $|\Q| = |\P|[k \mapsto \max(|\P|_k - 1, 2)]$ \\[0.2em]
$\Ite{E}{P_1}{P_2}$ &
  $\exists \R, \S, \bgamma, \bdelta \ldotp \; \evgeo{\egd{\P}{\balpha}}{E} \semgeo{P_1} \egd{\R}{\bgamma} \; \land \; \evgeo{\egd{\P}{\balpha}}{\lnot E} \semgeo{P_2} \egd{\S}{\bdelta}$ \newline
  ${} \land \; \left(\egd{\R}{\bgamma}, \egd{\S}{\bdelta}, \egd{\Q}{\bbeta}\right) \in \JoinRel$ (see \cref{def:join-relation}) \\[0.2em]
$\Whilst{E}{P}$ &
  $\exists c, \R, \S, \bgamma, \bdelta \ldotp \; 0 \le c < 1 \;  \land \; \egd{\P}{\balpha} \egdle \egd{\R}{\bgamma}$ \newline
  ${} \land \; \evgeo{\egd{\R}{\bgamma}}{E} \semgeo{P} \egd{\S}{\bdelta} \; \land \; \egd{\S}{\bdelta} \egdle \egd{c \cdot \R}{\bgamma}$ \newline
  ${} \land \; \evgeo{\egd{\frac{\R}{1-c}}{\bgamma}}{\lnot E} = \egd{\Q}{\bbeta}$ \\[0.2em]
$\Fail$ &$\bbeta = \zero_n \; \land \; \Q = \zero_{[1] \times \dots \times [1]}$ \\
\bottomrule
\end{tabularx}\vspace{-0.5em}
\caption{Geometric bound semantics for statements\vspace{-1em}}
\label{fig:egd-statement-semantics}
\end{figure}

\paragraph{Statement semantics}
The statement semantics is a binary relation $\semgeo{P}$ on EGDs.
It defines $\egd{\P}{\balpha} \semgeo{P} \egd{\Q}{\bbeta}$ by induction on the structure of $P$ via constraints on $\P, \balpha, \Q, \bbeta$ (\cref{fig:egd-statement-semantics}) and ensures $\sem{P}(\egd{\P}{\balpha}) \mle \egd{\Q}{\bbeta}$.
The semantics of $\Skip$ is the identity relation and $P_1;P_2$ is relational composition.
Since $\Fail$ corresponds to observing a zero-probability event, its right-hand side is a zero distribution.
For $X_k := 0$, we essentially have to marginalize out the $k$-th dimension (see \cref{lem:egd-marginalize}) and then put all the probability mass on $X_k = 0$.
For instance, applying $\semgeo{X_2 := 0}$ to the $\egd{\P}{\balpha}$ from \cref{example:egd}, yields the unique right-hand side
\[ \egd{\Q}{\bbeta} = \egd{\begin{pmatrix}
\P_{0,0} + \frac{\P_{0,1}}{1 - \alpha_2} & 0 \\
\P_{1,0} + \frac{\P_{1,1}}{1 - \alpha_2} & 0
\end{pmatrix}}{(\alpha_1, 0)} \]
For $X_k \passign c$, we shift the coefficients to the right by $c$ in the $k$-th dimension, and fill up with zeros.
For instance, applying $\semgeo{X_2 \passign 2}$ to the above EGD yields the unique right-hand side
\[ \egd{\Q}{\bbeta} = \egd{\begin{pmatrix}
0 & 0 &\P_{0,0} &\P_{0,1} \\
0 & 0 &\P_{1,0} &\P_{1,1}
\end{pmatrix}}{\balpha} \]
For $X_k \massign 1$, we shift the coefficients to the left by 1 in the $k$-th dimension, except for $\P_{k:0}$, which stays at index $0$, so we get the sum of $\P_{k:0}$ and $\P_{k:1}$ at index $0$.
This special case requires $|\Q|_k \ge 2$, so we may first have to expand $\P$ to ensure $|\P|_k \ge 3$, which is done implicitly in \cref{fig:egd-statement-semantics}.
As an example, applying $\semgeo{X_2 \massign 1}$ to the above EGD yields the unique right-hand side
\[ \egd{\Q}{\bbeta} = \egd{\begin{pmatrix}
\P_{0,0} + \P_{0,1} & \P_{0,1} \cdot \alpha_2 \\
\P_{1,0} + \P_{1,1} & \P_{1,1} \cdot \alpha_2
\end{pmatrix}}{\balpha} \]

\paragraph{Conditionals}
The semantics of $\Ite{E}{P_1}{P_2}$ is more complex.
We first require the existence of bounds on both branches ($\egd{\R}{\bgamma}$ and $\egd{\S}{\bdelta}$), whose constraints are given by: $\evgeo{\egd{\P}{\balpha}}{E} \semgeo{P_1} \egd{\R}{\bgamma}$ and $\evgeo{\egd{\P}{\balpha}}{\lnot E} \semgeo{P_2} \egd{\S}{\bdelta}$.
Finally, we would like to sum the bounds on the branches to obtain a bound on the whole conditional.
However the sum of two EGDs may not be an EGD.
Instead, we define a \emph{join} $\egd{\Q}{\bbeta}$ as an upper bound on the sum $\egd{\R}{\bgamma} + \egd{\S}{\bdelta}$ of a certain shape.
\begin{definition}[Join relation]
\label{def:join-relation}
We say $\egd{\Q}{\bbeta}$ is a \defn{join} of $\egd{\R}{\bgamma}$ and $\egd{\S}{\bdelta}$, and write $(\egd{\R}{\bgamma}, \egd{\S}{\bdelta}, \egd{\Q}{\bbeta}) \in \JoinRel$, if $\bgamma \le \bbeta, \bdelta \le \bbeta$ and there are expansions $\egd{\R'}{\bgamma}$ of $\egd{\R}{\bgamma}$ and $\egd{\S'}{\bdelta}$ of $\egd{\S}{\bdelta}$, of the same size, such that $\Q = \R' + \S'$.
We also define the \defn{strict join} relation $\JoinRel^* \subset \JoinRel$ that strengthens the condition $\bbeta \ge \max(\bgamma, \bdelta)$ to $\bbeta = \max(\bgamma, \bdelta)$.
\end{definition}
As an example, the minimal-size join of $\egd{\P}{\balpha}$ and $\egd{\Q}{\bbeta}$ from \cref{example:egd-order} is given by
\[
\egd{\begin{pmatrix}
\P_{0,0} + \Q_{0,0} & \P_{0,1} + \Q_{0,1} & \P_{0,1} \cdot \alpha_2 + \Q_{0,2} & \P_{0,1} \cdot \alpha_2^2 + \Q_{0,3} \\
\P_{1,0} + \Q_{0,0} \cdot \beta_1 & \P_{1,1} + \Q_{0,1} \cdot \beta_1 & \P_{1,1} \cdot \alpha_2 + \Q_{0,2} \cdot \beta_1 & \P_{1,1} \cdot \alpha_2^2 + \Q_{0,3} \cdot \beta_1
\end{pmatrix}}{(\gamma_1, \gamma_2)}
\]
with $\alpha_1, \beta_1 \le \gamma_1 < 1$ and $\alpha_2, \beta_2 \le \gamma_2 < 1$.
It is a strict join if $\gamma_1 = \max(\alpha_1, \beta_1)$ and $\gamma_2 = \max(\alpha_2, \beta_2)$.
We use normal joins in the semantics because strict joins would introduce maxima in the constraints, making them harder to solve.
Strict joins are useful for theoretical analysis, however.

\paragraph{Loops}
Bounding a loop $\Whilst{E}{B}$ requires the existence of a contraction invariant $\egd{\R}{\bgamma}$ and a contraction factor $c \in [0, 1)$ (see \cref{sec:contraction-invariants}), satisfying the following conditions.
First, the initial distribution has to be bounded by the invariant: $\egd{\P}{\balpha} \egdle \egd{\R}{\bgamma}$.
Second, the invariant has to decrease by a factor of $c$ in each loop iteration, which is encoded as: there exists an $\egd{\S}{\bdelta}$ such that $\evgeo{\egd{\R}{\bgamma}}{E} \semgeo{B} \egd{\S}{\bdelta}$ and $\egd{\S}{\bdelta} \egdle \egd{c \cdot \R}{\bgamma}$.
Then $\egd{\Q}{\bbeta} := \evgeo{\egd{\frac{\R}{1-c}}{\bgamma}}{\lnot E}$ is an upper bound on $\sem{\Whilst{E}{B}}(\egd{\P}{\balpha})$, as discussed in \cref{sec:contraction-invariants}.

\paragraph{Nondeterminism}
Note that there are two places in the semantics where choices have to be made: conditionals and loops.
This nondeterminism is the reason why the semantics $\semgeo{-}$ is a relation instead of a function.
For $\Ite{E}{P_1}{P_2}$, the choice is in the join operation: concretely, the size of the expansion of the two distributions.
For $\Whilst{E}{P}$, the choice is in the contraction invariant $\egd{\R}{\bgamma}$ and the factor $c$.
First of all, the dimensions $|\R|$ have to be chosen.
How we do this in practice is discussed in \cref{sec:impl}.
Once the size of $\R$ is chosen, the conditions on $\R$ reduce to polynomial inequality constraints (a decidable problem).
Of course, we typically want to find a ``good'' solution to these constraints by optimizing some objective (see \cref{sec:impl-optimization}).

\subsection{Examples}
\label{sec:geo-bound-examples}

\begin{example}[Simple counter]
\label{example:geo-counter}
Consider a simple program representing a geometric distribution:
\[ \Whilst{\Flip(1/2)}{X_1 \passign 1 } \]
The starting state of the program is described by $\egd{1}{0}$.
Assume a $c$-contraction invariant $\egd{p}{\alpha}$ exists with $p \in \nnr \cong \nnr^{[1]}$ and $\alpha \in [0,1)$.
Then the constraint $\egd{1}{0} \egdle \egd{p}{\alpha}$ yields the inequalities $p \ge 1, \alpha \ge \zero$.

Let the loop body be $B := (X_1 \passign 1)$.
We find $\evgeo{\egd{p}{\alpha}}{\Flip(1/2)} \semgeo{B} \egd{\begin{pmatrix}0 &\frac{p}{2}\end{pmatrix}}{\alpha}$.
The constraint for the contraction invariant is $\egd{\begin{pmatrix}0 &\frac{p}{2}\end{pmatrix}}{\alpha} \egdle \egd{c \cdot p}{\alpha} = \egd{\begin{pmatrix}c p &c p \alpha\end{pmatrix}}{\alpha}$ and amounts to $p \ge 1$ (from above) and $\frac{1}{2} \le c \cdot \alpha$ with $c < 1$.
The bound on the whole loop is then given by $\frac{1}{1-c}\evgeo{\egd{p}{\alpha}}{\lnot\Flip(1/2)} = \egd{\frac{p}{2(1-c)}}{\alpha}$.
To get a low upper bound, it is best to set $p = 1$.
How to choose $\alpha$ and $c$ is not clear because decreasing $\alpha$ will increase $\frac{1}{1-c}$ and vice versa.

To optimize the asymptotics of the tail probabilities $\Prob[X_1 = n]$ as $n \to \infty$, we want to choose $\alpha$ as small as possible, i.e. very close to $\frac{1}{2}$, accepting that this will make $\frac{1}{1-c} \ge \frac{2\alpha}{2\alpha - 1}$ large.
This yields the bound $\sem{P}(\egd{1}{0}) \egdle \egd{\begin{pmatrix}\frac{1 + 2\epsilon}{4\epsilon}\end{pmatrix}}{\frac{1}{2} + \epsilon}$, i.e. $\Prob[X_1 = n] \le \left(\frac{1}{4\epsilon} + \frac{1}{2}\right) \left(\frac{1}{2} + \epsilon\right)^n$, see \apxref{apx:full-examples}.

To optimize the bound on the expected value of $X_1$, we want to minimize the bound $\frac{\alpha}{2(1-c)(1-\alpha)^2} \ge \ExpVal[X_1]$ from \cref{thm:egd-stats}.
This is achieved under the constraints $0 \le \alpha, c < 1$ and $\frac12 \le \alpha c$ for $\alpha = \frac{\sqrt5 - 1}{2} \approx 0.618$ and $c = \frac{\sqrt5 + 1}{4} \approx 0.809$.
At this point, the bound on the expected value is $\ExpVal[X_1] \le \frac{11 + 5 \sqrt5}{2} \approx 11.09$.
This is quite a bit more than the true value of 1, but it is a finite bound that can be found fully automatically.
It can be improved by unrolling the loop a few times before applying the above procedure.
This way, our implementation finds much better bounds (see \cref{table:moment-tail-bounds}).
\end{example}

\begin{example}[Asymmetric random walk]
\label{example:asym-rw}
Consider the program representing a biased random walk on $\NN$, starting at 1, and stopping when reaching 0:
\[ X_1 := 1; X_2 := 0; \Whilst{X_1 > 0}{X_2 \passign 1; \ProbBranch{X_1 \passign 1}{r}{X_1 \massign 1}} \]
where $X_1$ is the current position, $X_2$ is the number of steps taken, and the bias $r < \frac12$ is the probability of going right.
Denote the loop body by $B := (X_2 \passign 1; \ProbBranch{X_1 \passign 1}{r}{X_1 \massign 1})$.
We find
\begin{align*}
  \egd{1}{(0, 0)} \semgeo{X_1 := 1; X_2 := 0} \egd{\begin{pmatrix}0 \\ 1\end{pmatrix}}{(0,0)}
\end{align*}
as the distribution before the loop.
Assume a $c$-contraction invariant $\egd{\P}{\balpha}$ with $\P \in \nnr^{[2] \times [1]}$ exists.
Then one loop iteration transforms it as follows (details in \apxref{apx:full-examples}):
\begin{align*}
  \egd{\begin{pmatrix}\P_{0,0} \\ \P_{1,0}\end{pmatrix}}{\balpha} \semgeo{B} \egd{\begin{pmatrix}0 & (1-r)\P_{1,0} \\ 0 & (1-r)\alpha_1 \P_{1,0} \\ 0 & (1-r)\alpha_1^2 \P_{1,0} + r \P_{1,0} \end{pmatrix}}{\balpha}
\end{align*}
This yields the following constraints for the loop invariant:
\begin{align*}
  \egd{\begin{pmatrix}0 \\ 1\end{pmatrix}}{(0,0)} &\egdle \egd{\begin{pmatrix}\P_{0,0} \\ \P_{1,0}\end{pmatrix}}{\balpha} \\
  \egd{\begin{pmatrix}0 & (1-r)\P_{1,0} \\ 0 & (1-r)\alpha_1 \P_{1,0} \\ 0 & (1-r)\alpha_1^2 \P_{1,0} + r \P_{1,0} \end{pmatrix}}{\balpha} &\egdle \egd{\begin{pmatrix}c \P_{0,0} \\ c \P_{1,0}\end{pmatrix}}{\balpha} = \egd{\begin{pmatrix}c \P_{0,0} & c \alpha_2 \P_{0,0} \\ c \P_{1,0} & c \alpha_2 \P_{1,0} \\ c \alpha_1 \P_{1,0} & c \alpha_1 \alpha_2 \P_{1,0} \end{pmatrix}}{\balpha}
\end{align*}
which reduce to the following polynomial constraints:
\[
  0 \le \P_{0,0} \quad
  1 \le \P_{1,0} \quad
  (1-r)\P_{1,0} \le c\alpha_2 \P_{0,0}
  \quad (1-r)\alpha_1 \P_{1,0} \le c \alpha_2 \P_{1,0}
  \quad (1-r)\alpha_1^2 \P_{1,0} + r \P_{1,0} \le c \alpha_1\alpha_2 \P_{1,0}
\]
besides the obvious ones (every variable is nonnegative and $\alpha_1, \alpha_2, c \in [0, 1)$).
The projection of the solution set to $\alpha_1, \alpha_2$ is shown in \cref{fig:asym-rw-feasible} for $r = \frac14$.
The most interesting constraint is the last one, which has the solutions $\alpha_1 \in \left[ \tfrac{c\alpha_2 - \sqrt{c^2\alpha_2^2 - 4r(1-r)}}{2(1-r)} , \tfrac{c\alpha_2 + \sqrt{c^2\alpha_2^2 - 4r(1-r)}}{2(1-r)} \right]$.
It can be shown (see \apxref{apx:full-examples}) that such an $\alpha_1 < 1$ exists if and only if $r < \frac12$, which makes sense because for $r \ge \frac12$, the program has infinite expected running time (see \cref{thm:necessary-runtime}).
Then all constraints are in fact satisfiable and we get a bound on the distribution of the program:
\begin{align*}
  \egd{1}{(0,0)} &\semgeo{P} \evgeo{\egd{\begin{pmatrix}\frac{\P_{0,0}}{1-c} \\ \frac{\P_{1,0}}{1-c}\end{pmatrix}}{\balpha}}{\lnot(X_1 > 0)} = \egd{\begin{pmatrix}\frac{\P_{0,0}}{1-c} \\ 0\end{pmatrix}}{(0, \alpha_2)}
\end{align*}
The asymptotic bound is $\Prob[X_2 = n] = \frac{\P_{0,0}}{1-c} \alpha_2^n$, so it's best for $\alpha_2$ as small as possible.
Since $\alpha_2 \ge \frac{\sqrt{4r(1-r)}}{c} > 2\sqrt{r(1-r)}$, the best possible geometric tail bound for $\Pr[X_2 = n]$ is $O((2\sqrt{r(1-r)} + \epsilon)^n)$.
\end{example}

\begin{wrapfigure}{r}{0.33\textwidth}
\vspace{-1em}
\centering
\includegraphics[width=0.33\textwidth]{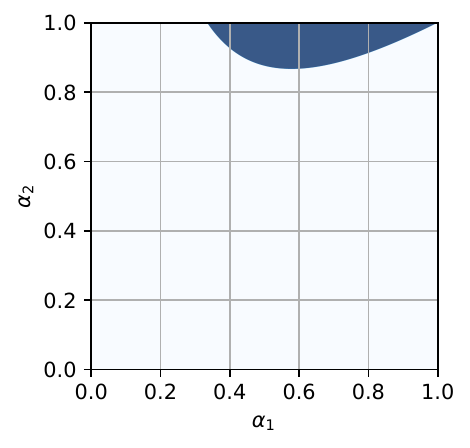}%
\vspace{-1em}%
\caption{Solution set for \cref{example:asym-rw}, projected onto $\alpha_1, \alpha_2$ for $r = \frac14$}
\label{fig:asym-rw-feasible}
\vspace{-2em}
\end{wrapfigure}

Interestingly, the exact asymptotic is $\Theta((2\sqrt{r(1-r)})^n \cdot n^{-3/2})$ (see \apxref{apx:geom-bound-sem}), so the best possible geometric bound is $O((2\sqrt{r(1-r)})^n)$ and our method can get arbitrarily close to it.
We could find a bound on $\ExpVal[X_2]$ in the same way as in the previous example, but this would be tedious to do manually.
Our implementation (again using loop unrolling before applying the above reasoning) finds the bounds shown in \cref{table:moment-tail-bounds} (for $r = \frac{1}{4}$).

\subsection{Properties}
\label{sec:geo-bound-properties}

\subsubsection{Decidability}
For any $\egd{\P}{\balpha}$ occurring in the semantics, each entry of $\P$ and $\balpha$ is a rational function of entries of initial blocks, decay rates, and contraction factors.
(In fact, it is linear in initial block entries.)
Hence the $\egdle$ constraints reduce to polynomial inequality constraints if the sizes of the EGDs are known.
If we want to find an upper bound, these constraints will contain unknowns, which we have to solve for.
Since the existential first-order theory of the reals is decidable, solving these constraints is indeed possible.
\begin{restatable}{theorem}{DecidabilityGeoBounds}
Given a program $P$ and an $\egd{\P}{\balpha}$, it is decidable whether there is an EGD $\egd{\Q}{\bbeta}$ such that $\egd{\P}{\balpha} \semgeo{P} \egd{\Q}{\bbeta}$, assuming the sizes of the intermediate EGDs for joins and contraction invariants in the semantics are fixed.
\end{restatable}
This is yet another reason to prefer EGDs over more general classes of distributions:
while it may be possible to check $\mle$ for slightly more complicated classes of distributions with concrete values, solving such a system of $\mle$-inequalities with unknown values will be much harder.
Despite the decidability results, solving our constraints is still not easy because the existential theory of the reals is NP-hard and algorithms only work for small instances in practice.

One might hope that the constraints arising from the geometric bound semantics have nice properties.
For example, one might expect the solution set $(\Q, \bbeta)$ to be convex and $\bbeta$ to be right-closed, i.e. that for any solution $\egd{\Q}{\bbeta}$ and any $\bgamma \ge \bbeta$, there is a solution $\egd{\Q'}{\bgamma}$, because this weakens the bound.
However, we found simple counterexamples to both properties (see \apxref{ex:double-geo}).
We address the problem of practical constraint solving in \cref{sec:impl-solving}.

\subsubsection{Soundness}

We prove that the geometric bound semantics is correct with respect to the standard semantics.
\begin{restatable}[Soundness]{theorem}{SoundnessGeoBounds}
\label{thm:soundness-geo-bounds}
The rules for the EGD event semantics (\cref{fig:egd-event-semantics}) agree with the standard event semantics: $\evgeo{\egd{\P}{\balpha}}{E} = \egd{\P}{\balpha}|_E$.
The rules for the EGD statement semantics $\semgeo{-}$ (\cref{fig:egd-statement-semantics}) are sound: if $\egd{\P}{\balpha} \semgeo{P} \egd{\Q}{\bbeta}$ then $\sem{P}(\egd{\P}{\balpha}) \mle \egd{\Q}{\bbeta}$.
Furthermore, if $\egd{\P}{\balpha}$ is a probability distribution, then we can bound the normalized distribution:
\[ \frac{\semlo{P}(\egd{\P}{\balpha})|_{\NN^n}}{\egd{\Q}{\bbeta}(\NN^n)} \mle \Normalize(\esem{P}(\egd{\P}{\balpha})) \mle \frac{\egd{\Q}{\bbeta}}{\semlo{P}(\egd{\P}{\balpha})(\NN^n)} \]
\end{restatable}

\subsubsection{Loop-free Fragment}
It is instructive to look at when $\semgeo{-}$ is precise, i.e. when the right-hand side of the relation is unique and equals the $\sem{-}$ semantics.
It turns out that this is the case for programs without loops, which is easy to see from the soundness proof.

\begin{restatable}[Precision for loop-free fragment]{theorem}{PrecisionLoopFree}
\label{thm:precision-loop-free}
Let $P$ be a loop-free program.
Assume that the strict $\JoinRel^*$ relation is used in the $\semgeo{-}$ semantics of $\Ite{-}{-}{-}$.
Then for all $\egd{\P}{\balpha}$, there is a unique $\egd{\Q}{\bbeta}$ such that $\egd{\P}{\balpha} \semgeo{P} \egd{\Q}{\bbeta}$ and $\sem{P}(\egd{\P}{\balpha}) = \egd{\Q}{\bbeta}$.
Furthermore, each $\beta_i \in \{0, \alpha_i\}$.
\end{restatable}

\subsubsection{Sufficient Conditions for the Existence of Bounds}
When does the constraint problem arising from our semantics actually have a solution?
This is hard to characterize in general because the feasible region can be complex, as discussed above.
Furthermore, a useful purely syntactic criterion is unlikely to exist because even for the simple asymmetric random walk $\Whilst{X_1 > 0}{\Ite{\Flip(\rho)}{X_1 \passign 1}{X_1 \massign 1}}$, a bound can be found if and only if $\rho < \frac12$.

Nevertheless, we have succeeded in identifying a sufficient criterion for the existence of an EGD loop bound:
roughly speaking, if the difference between new and old values of the variables is the same in each loop iteration, and if there is a linear ranking supermartingale \citep{ChatterjeeFNH16}, then an EGD bound exists.
The precise conditions are stated in the following theorem.

\begin{restatable}{theorem}{SufficientConditions}
\label{thm:sufficient-cond}
Suppose $L = \Whilst{E}{B}$ satisfies the following properties:
\begin{enumerate}
  \item $B$ only contains increment and decrement statements ($X_k \passign 1$, $X_k \massign 1$), probabilistic branching ($\Ite{\Flip(\dots)}{\dots}{\dots}$), and $\Fail$,
  \item $E$ occurring guarantees that $X_1 \ge a_1 \land \dots \land X_n \ge a_n$ where $a_i \in \NN$ is the maximum number of decrements of $X_i$ in any program path through $B$, and
  \item there is a conical (linear with nonnegative coefficients) combination $\sum_{i=1}^n \lambda_i X_i$ whose expected value decreases after every loop iteration for any initial assignment $\x$ that can enter the loop:
  $\ExpVal_{\X \sim \sem{B}(\Dirac(\x)|_{E})}\left[\sum_{i=1}^n \lambda_i X_i\right] < \sum_{i=1}^n \lambda_i x_i \quad \text{for all } \x \in \NN^n \text{ with } x_i \ge a_i$
\end{enumerate}
Then, for any initial $\egd{\P}{\balpha}$, there is a solution $(\Q, \bbeta)$ to $\egd{\P}{\balpha} \semgeo{L} \egd{\Q}{\bbeta}$.
\end{restatable}

Although the assumptions of this theorem may seem restrictive, they are sufficient to prove the existence of bounds for the random walk (\cref{example:asym-rw}), which is nontrivial to analyze.
In that case, $a_1 = 1, a_2 = 0$ and the expected value of $X_1$ decreases by $1 - 2r > 0$ in each loop iteration.

\subsubsection{Necessary Conditions for the Existence of Bounds}
We already mentioned that the existence of bounds is not guaranteed in general.
An obvious necessary condition is that the true program distribution actually has tails that decay exponentially fast, i.e. can be bounded by an EGD.
It turns out that the same must hold for the running time of the program as well.

\begin{restatable}[Necessary conditions on the running time]{theorem}{NecessaryRuntime}
\label{thm:necessary-runtime}
If $\egd{\P}{\balpha} \semgeo{P} \egd{\Q}{\bbeta}$ for an $\egd{\P}{\balpha}$ then the running time $T$ of $P$ on $\egd{\P}{\balpha}$ can be bounded by an EGD as well.
In particular, all its moments $\ExpVal[T^k]$ must be finite.
\end{restatable}

\subsubsection{Loop Unrolling and Convergence}

Just like for the residual mass semantics, we would expect the geometric bound semantics to yield tighter and tighter bounds as we unroll loops further and further.
It turns out that for the $u$-fold unrolling $P^{(u)}$ of a program $P$, we can find upper bounds whose distance from the true distribution decreases exponentially in $u$.

\begin{restatable}[Convergence]{theorem}{ConvergenceUpperBounds}
\label{thm:convergence}
Let $P$ be a program containing potentially nested loops and $P^{(u)}$ its $u$-fold unrolling.
Suppose $\egd{\P}{\balpha} \semgeo{P} \egd{\Q}{\bbeta}$.
Then there exist $A \in \nnr, C \in [0,1)$ and $\Q^{(u)}$ such that $\egd{\P}{\balpha} \semgeo{P^{(u)}} \egd{\Q^{(u)}}{\bbeta} \egdle \egd{\Q}{\bbeta}$ and
\[ \egd{\Q^{(u)}}{\bbeta} - \sem{P}(\egd{\P}{\balpha}) \mle A \cdot C^u \cdot \egd{\Q}{\bbeta} \]
In particular, the distribution bound $\egd{\Q^{(u)}}{\bbeta}$ converges in total variation distance to the true distribution $\sem{P}(\egd{\P}{\balpha})$, as $u \to \infty$.
Similarly, the $k$-th moment bound $\ExpVal_{\X \sim \egd{\Q^{(u)}}{\bbeta}}[X_i^k]$ converges to the true moment $\ExpVal_{\X \sim \sem{P}(\egd{\P}{\balpha})}[X_i^k]$ for any $k$.
\end{restatable}

Note that this theorem does \emph{not} prove that the parameters $\bbeta$ of the geometric tails converge to the best possible ones.
In general, unrolling does not improve the decay rate at all, except for specific cases where the loop is finite and can be fully unrolled.
In fact, the derived tail bounds may differ arbitrarily from the true tails, as shown in the next example.

\begin{example}[Nonoptimal tail bounds]
\label{ex:imprecise-tails}
As an example of nonoptimal tail bounds, consider the following program where $\rho \in (0,1)$:
\begin{align*}
  &X_1 := 0; X_2 := 0; \\
  &\Whilst{\Flip(\rho)}{X_1 \passign 1; X_2 \passign 1}; \\
  &\Whilst{X_1 > 0 \land X_2 > 0}{X_1 \massign 1; X_2 \massign 1}
\end{align*}
It is clear that at the end of the program, $X_1 = 0$ and $X_2 = 0$ almost surely.
Throughout the program, the same operations are applied to them.
However, the geometric bound semantics cannot keep track of this correlation completely because the structure of EGDs can only express correlations in the finite initial block and not in the tails.
Hence the best tail bounds we get for $X_1$ and $X_2$ are of the form $O((\rho + \delta)^n)$ for any $\delta > 0$, whereas the true tails are zero.
So the tail bounds can be arbitrarily far from the true tails.
Note, however, that the existence of geometric tails is already useful knowledge: it tells us that the distribution is not heavy-tailed and that all moments are finite.
\end{example}

\section{Implementation}
\label{sec:impl}

We implemented both the residual mass semantics and the geometric bound semantics in a tool called \defn{Diabolo} (``\textbf{Di}screte \textbf{Di}stribution \textbf{A}nalysis via \textbf{Bo}unds, supporting \textbf{L}oops and \textbf{O}bservations'').
The code is available at \href{https://github.com/fzaiser/diabolo}{github.com/fzaiser/diabolo} and on Zenodo \citep{Zaiser24artifact}.

Diabolo takes a probabilistic program as input and outputs bounds on probability masses, moments, and tail asymptotics of the posterior distribution of a specified program variable.
The output and computational effort can be controlled with various options.
The most important one is the loop unrolling limit $u$, which is common to both semantics.
The geometric bound semantics has additional options: the dimension $d$ of the contraction invariants to be synthesized, an optimization objective, as well as the solver and optimizer to be used for the polynomial constraints.
The optimization objective is the quantity whose bound should be optimized: the probability mass, the expected value, or the asymptotic decay rates.
Depending on which of these is chosen, the best EGD bound can vary significantly (see \cref{sec:geo-bound-examples}).
Diabolo is implemented in Rust \citep{MatsakisK14} for performance and uses exact rational number computations instead of floating point to ensure the correctness of the results.
The distributions in the residual mass semantics are represented as arrays of probability masses.
We describe a few practical aspects of the implementation below, most of them pertaining to the geometric bound semantics.

\paragraph{Nondeterminism in the semantics}
As pointed out in \cref{sec:geo-bound-semantics}, there are two places in the geometric bound semantics where choices have to be made: branching and loops.
To generate the polynomial constraints arising from $\keyword{while}$ loops, we need to choose the size of the contraction invariants $\egd{\R}{\bgamma}$.
It is usually best to choose the dimensions of $\R$ to be as small as possible (often 1): this reduces both the number of constraints and constraint variables, which facilitates solving.
Some programs only have larger contraction invariants, but the dimension rarely needs to be greater than 2.
In terms of quality of the bounds, increasing the unrolling limit has a much greater effect (cf. \cref{thm:convergence}).
In the semantics of $\keyword{if}$ statements, the choice is in the size of the expansion of the EGDs in the $\JoinRel$ relation.
In Diabolo, we choose the smallest expansion for speed and memory usage reasons.
The size of the EGDs in a loop depends mostly on the size of the contraction invariant, so an increase there automatically reduces the imprecision in $\JoinRel$.

\paragraph{Approximating the support}
In the constraint problem arising from a $c$-contraction invariant $\egd{\R}{\bgamma}$, the constraint variables corresponding to the entries of $\R$ are ``easy'' because they only occur linearly, whereas the constraints typically contain higher-degree polynomials in $\bgamma$.
For this reason, it is desirable to determine the decay rates $\bgamma$ in advance, if possible.
If the program variable $X_i$ has finite support $\{a_i, a_i + 1, \dots, b_i\} \subset \NN$, we can set the decay rate $\gamma_i$ to zero (since the tails of $X_i$'s distribution are zero).
But to be able to represent its distribution with $\gamma_i = 0$, we have to choose at least $|\R|_i = b_i + 1$.
(In this case, the user-provided dimension $d$ for the contraction invariant is overridden.)
Furthermore, we can infer from the support that $\R_{i: j} = \zero$ for $j < a_i$.
To reap these benefits, we need to analyze the support of the random variables occurring in the program.
This is a standard application of abstract interpretation on the interval domain.
Overapproximating the support also has the benefit that we automatically know that probabilities outside of it are zero, which can improve the results of both semantics (see \apxref{apx:approx-support}).

\subsection{Constraint Solving}
\label{sec:impl-solving}

How can we solve the polynomial constraints arising from the geometric bound semantics?
As this is a decidable problem known as the \emph{existential theory of the reals}, we first tried to use SMT solvers, such as Z3 \citep{MouraB08} and CVC5 \citep{BarbosaBBKLMMMN22}, and the dedicated tool QEPCAD \citep{Brown03}.
Out of these, Z3 performed the best, but it was only able to solve the simplest problems and scaled badly.

\paragraph{Numerical solvers}
Given that even simple programs lead to dozens and sometimes hundreds of constraints, we resorted to numerical solutions instead.
The best solver for our purposes turned out to be IPOPT: a library for large-scale nonlinear optimization \citep{WachterB06}.
One issue with IPOPT are rounding errors, so it sometimes returns solutions that are not exactly feasible.
To address this, we tighten the constraints by a small margin before handing them over to IPOPT and check the feasibility of the returned solution with exact rational arithmetic.
Another disadvantage of IPOPT is that it cannot prove infeasibility -- only exact solvers like Z3 and QEPCAD can do so.
However, its impressive scalability makes up for these shortcomings.

We also developed a custom solver that transforms the constraint problem into an unconstrained optimization problem (details in \apxref{apx:impl-solvers}) and applies the ADAM optimizer \citep{KingmaB14}, a popular gradient descent method.
We also explored other off-the-shelf solvers for polynomial constraints: IBEX, a rigorous solver based on interval arithmetic, and constrained optimization methods in scipy, the Python scientific computing library; but neither of them performed well.
Diabolo includes the solvers IPOPT (the default), our ADAM-based solver, and Z3.

\paragraph{Nested loops}
Nested loops lead to cyclic $\le$-constraints on the decay rates, e.g. $\alpha \le \beta \land \beta \le \alpha$.
Numerical solvers struggle with such indirect equality constraints, so we added a preprocessing step to detect them and replace the equal variables with a single representative.

\paragraph{Reducing the cost of unrolling}
Unrolling loops is essential to obtain tight bounds on posterior masses and moments (it does not affect tail bounds).
On the other hand, unrolling increases the complexity of the constraint problem considerably, both in terms of number of variables and constraints.
A key observation from \cref{thm:convergence} can mitigate this problem: the variables that occur nonlinearly, i.e. decay rates and contraction factors, need not be changed as the unrolling count increases.
As a consequence, we first solve the constraint problem without unrolling and need only solve a \emph{linear} constraint problem for higher unrolling limits.
Since linear programming solvers are much faster and more robust than nonlinear solvers, this approach significantly reduces numerical issues and computation time.
In fact, without this technique, several benchmarks in \cref{sec:eval-quality} would not be solvable at all.

\subsection{Optimization}
\label{sec:impl-optimization}

Since the geometric bound semantics is nondeterministic, there are many EGD bounds for a given program.
Which one is the best depends on what quantity we want to optimize.
In Diabolo, the user can specify one of the following optimization objectives to minimize: the bound on the total probability mass, on the expected value, or on the tail asymptotics (i.e. the decay rate).

IPOPT and the ADAM-based solver can also be used for optimization.
Moreover, we apply the linear programming solver CBC by the COIN-OR project to optimize the linear variables (keeping the nonlinear ones fixed), which is very fast.
By default, Diabolo runs IPOPT, the ADAM-based optimizer, and the linear solver in this order, each improving the solution of the previous one.
The ADAM-based optimizer can be slow for larger programs, in which case the user may decide to skip it.
However, it often finds better tail bounds than IPOPT, which is why it is included by default.

\section{Empirical Evaluation}
\label{sec:eval}

In this section, we evaluate our two methods in practice to answer the following four questions:
\begin{enumerate}
  \item How often is the geometric bound semantics applicable in practice? (\cref{sec:eval-applicability})
  \item How tight are the bounds in practice? (\cref{sec:eval-quality})
  \item How do our methods perform compared to previous work? (\cref{sec:eval-prev})
  \item How do our two semantics perform compared to each other? (\cref{sec:eval-comparison})
\end{enumerate}
All benchmarks and code to reproduce the results are available \citep{Zaiser24artifact}.

\subsection{Applicability of the Geometric Bound Semantics}
\label{sec:eval-applicability}

In this section, we empirically investigate the \emph{existence} of geometric bounds, i.e. how often the constraints arising from the geometric bound semantics can be solved.
If some bound can be found at all, \cref{thm:convergence} ensures that it can be made arbitrarily tight by increasing the unrolling limit.
(This aspect of the \emph{quality} of bounds is studied in \cref{sec:eval-quality,sec:eval-comparison}.)
While \cref{thm:sufficient-cond} provides sufficient conditions for existence, we want to test the applicability of the geometric bound semantics in practice in a \emph{systematic} way.

\paragraph{Benchmark selection}
For this purpose, we collected benchmarks from the Github repositories of several probabilistic programming languages: Polar \citep{MoosbruggerSBK22}, Prodigy \citep{KlinkenbergBCHK24}, and PSI \citep{GehrMV16}.
Note that a benchmark being available in a tool's repository does not mean that the tool can solve it.
We searched all benchmarks from these repositories for the keyword \verb|while|, in order to find benchmarks with loops.
We manually filtered out benchmarks whose loops are actually bounded or that make essential use of continuous distributions (15), negative integers (4), comparisons of two variables (7), or multiplication (1), because our language cannot express these.
Some benchmarks using these features could still be translated in other ways to an equivalent program in our language.

We ended up with 43 benchmarks: 9 from Polar, 11 from Prodigy, 9 from PSI, and we added 14 of our own.
They include standard probabilistic loop examples (in particular, all examples from this paper, and variations thereof), nested loops, and real-world algorithms, such as probabilistic self-stabilization protocols \citep{Herman90,IsraeliJ90,BeauquierGJ99}.

\paragraph{Symbolic inputs}
Polar and Prodigy can handle symbolic inputs and symbolic parameters to some extent, which our techniques cannot.
One benchmark (\texttt{polar/fair\_biased\_coin}) used symbolic parameters, which we replaced with concrete values.
Several benchmarks from Polar and Prodigy have symbolic inputs, i.e. they are parametric in the initial values in $\NN$ of the variables.
Since our method cannot reason parametrically about the input, we instead put a $\Geometric(1/2)$ distribution on such inputs, to cover all possible values.
Of course, this yields less information than a method computing a symbolic result that is valid for all input values.
But if an EGD bound can be found for this input distribution, then an EGD bound can be found for each possible input value because any Dirac distribution on the input can be bounded by a scaled version of the geometric distribution: $\Dirac(m) \mle 2^m \cdot \Geometric(1/2)$.
Thus, for the question of the existence of bounds, we consider this a valid approach.
Note that reasoning about such a geometric distribution on the input is nontrivial (and much harder than reasoning about fixed input values).
For instance, Polar cannot always handle it even for benchmarks where it supports the version with symbolic input values.

\paragraph{Methodology}
We ran our tool Diabolo for the geometric bound semantics on all benchmarks with a timeout of 5 minutes.
The configuration options were mostly left at the defaults for most benchmarks (invariant size: 1, solver: IPOPT).
However, we disabled unrolling for all benchmarks since it only affects the quality of bounds, not their existence.
For 4 benchmarks, the invariant size had to be increased to 2 or 3 to find bounds.
For each benchmark, we recorded the time it took to compute the bounds, or any errors.

\paragraph{Results}
Diabolo was able to solve 37 out of the 43 benchmarks (86\%); 5 failed because no EGD bound exists (at least 3 of them have infinite expected running time); and 1 timed out due to the complexity of the constraint problem.
Among the solved benchmarks, the solution time was at most 3 seconds and typically much less.
This demonstrates that our geometric bound semantics is also of practical relevance, in addition to its nice theoretical properties.
Detailed results are available in \apxref{table:benchmarks-applicability}, along with statistics about the input program and the constraint problem arising from the geometric bound semantics.

\subsection{Quality of the Geometric Bounds}
\label{sec:eval-quality}

\begin{table}
\centering
\scriptsize
\caption{Diabolo's bounds on the expected value and the tail asymptotics for the benchmarks from \cref{sec:eval-applicability} where bounds exist.
Results are rounded to 4 significant digits.
(\#U: unrolling limit; EV: expected value; ``?'': ground truth could not be determined; ``\dag'': could not be solved by an automatic tool before, in the sense of bounding the distribution's moments and tails without user intervention; timeout: 5 minutes exceeded.) \vspace{-0.5em}}
\label{table:quality-bounds}
\alternaterowcolors
\begin{tabular}{l|cccc|ccc}
\toprule
Benchmark & \#U & True EV & EV bound & Time & True tail & Tail bound & Time \\
\midrule
polar/c4B\_t303 \dag & 30 & 0.1787... & [0.1787, 0.1788] & 0.24 s & ? & $O(0.5001^n)$ & 0.12 s \\
polar/coupon\_collector2 & 30 & 2 & [1.999, 2.001] & 0.16 s & $\Theta(0.5^n)$ & $O(0.5071^n)$ & 0.11 s \\
polar/fair\_biased\_coin & 30 & 0.5 & [0.4999, 0.5001] & 0.03 s & $0$ & $0$ & 0.02 s \\
polar/las\_vegas\_search & 200 & 20 & [19.98, 26.79] & 2.29 s & $\Theta(0.9523...^n)$ & $O(0.9524^n)$ & 0.72 s \\
polar/linear01 \dag & 30 & 0.375 & [0.3749, 0.3751] & 0.03 s & $0$ & $0$ & 0.03 s \\
polar/simple\_loop & 30 & 1.3 & [1.299, 1.301] & 0.02 s & $0$ & $0$ & 0.02 s \\
prodigy/bit\_flip\_conditioning & 30 & 1.254... & [1.254, 1.255] & 0.28 s & ? & $O(0.6254^n)$ & 0.19 s \\
prodigy/brp\_obs \dag & 30 & 4.989...e-10 & [4.989e-10, 1.489e-09] & 0.61 s & $0$ & $0$ & 3.83 s \\
prodigy/condand \dag & 30 & 0.75 & [0.7499, 0.7501] & 0.15 s & ? & $O(0.5001^n)$ & 0.07 s \\
prodigy/dep\_bern \dag & 30 & 0.5 & [0.4999, 0.5124] & 1.00 s & $\Theta(0.3333...^n)$ & $O(0.34^n)$ & 0.17 s \\
prodigy/endless\_conditioning & 30 & undef & undef & 0.15 s & $0$ & $0$ & 0.13 s \\
prodigy/geometric & 30 & 2 & [1.999, 2.007] & 1.09 s & $\Theta(0.5^n)$ & $O(0.5066^n)$ & 0.60 s \\
prodigy/ky\_die & 30 & 3.5 & [3.499, 3.501] & 0.18 s & $0$ & $0$ & 0.08 s \\
prodigy/n\_geometric \dag & 30 & 1 & [0.998, 1.068] & 0.08 s & $\Theta(0.6666...^n)$ & $O(0.673^n)$ & 0.07 s \\
prodigy/trivial\_iid \dag & 30 & 3.5 & [3.499, 3.503] & 2.60 s & $\Theta(0.8318...^n)$ & $O(0.836^n)$ & 0.16 s \\
psi/beauquier-etal3 & 30 & ? & \xmark{} (timeout) & t/o & ? & $O(0.7952^n)$ & 17.64 s \\
psi/cav-example7 & 80 & 10.41... & [10.41, 10.51] & 1.71 s & ? & $O(0.8927^n)$ & 0.12 s \\
psi/dieCond (Ex. \ref{ex:die-paradox}) \dag & 40 & 1.5 & [1.499, 1.501] & 0.15 s & $\Theta(0.3333...^n)$ & $O(0.3395^n)$ & 0.13 s \\
psi/ex3 & 30 & 0.6666... & [0.6666, 0.6667] & 0.04 s & $0$ & $0$ & 0.06 s \\
psi/ex4 & 30 & 0.6666... & [0.6666, 0.6667] & 0.23 s & $0$ & $0$ & 0.27 s \\
psi/fourcards & 30 & 0.2642... & [0.264, 0.2648] & 0.46 s & $0$ & $0$ & 0.34 s \\
psi/herman3 & 30 & 1.333... & [1.333, 1.334] & 62.08 s & ? & $O(0.5002^n)$ & 0.48 s \\
psi/israeli-jalfon3 & 30 & 0.6666... & [0.6666, 0.6668] & 1.71 s & ? & $O(0.2501^n)$ & 0.15 s \\
psi/israeli-jalfon5 & 30 & ? & \xmark{} (timeout) & t/o & ? & $O(0.6583^n)$ & 7.13 s \\
ours/1d-asym-rw (Ex. \ref{example:asym-rw}) \dag & 70 & 2 & [1.999, 2.542] & 1.47 s & $\Theta(\frac{0.8660...^n}{n^{1.5}})$ & $O(0.8682^n)$ & 0.12 s \\
ours/2d-asym-rw \dag & 70 & ? & [7.394, 100.4] & 79.60 s & ? & $O(0.9385^n)$ & 4.88 s \\
ours/3d-asym-rw \dag & 30 & ? & [9.443, 6.85e+05] & 289.94 s & ? & $O(0.9812^n)$ & 6.35 s \\
ours/asym-rw-conditioning \dag & 70 & 2.444... & [2.316, 2.588] & 4.85 s & $0$ & $0$ & 0.05 s \\
ours/coupon-collector5 & 80 & 11.41... & [11.41, 11.56] & 47.93 s & $\Theta(0.8^n)$ & $O(0.8002^n)$ & 15.56 s \\
ours/double-geo \ifarxiv(Ex. \ref{ex:double-geo})\fi \dag & 30 & 2 & [1.999, 2.01] & 0.32 s & $\Theta(0.7071...^n)$ & $O(0.711^n)$ & 0.12 s \\
ours/geometric (Ex. \ref{example:geo-counter}) \dag & 30 & 1 & [0.9999, 1.006] & 0.12 s & $\Theta(0.5^n)$ & $O(0.5037^n)$ & 0.10 s \\
ours/grid \dag & 30 & 2.75 & [2.749, 2.766] & 2.39 s & ? & $O(0.5285^n)$ & 0.30 s \\
ours/imprecise\_tails (Ex. \ref{ex:imprecise-tails}) \dag & 30 & 0 & [0, 0.007054] & 2.36 s & $0$ & $O(0.5001^n)$ & 0.18 s \\
ours/israeli-jalfon4 & 30 & ? & [1.457, 1.458] & 121.04 s & ? & $O(0.5011^n)$ & 0.44 s \\
ours/nested \dag & 5 & ? & [0.9152, 10.57] & 12.59 s & ? & $O(0.4998^n)$ & 0.40 s \\
ours/sub-geom \dag & 30 & 0.6666... & [0.6666, 0.6667] & 0.06 s & ? & $O(0.5001^n)$ & 0.10 s \\
ours/sum-geos & 80 & 8 & [7.998, 8.001] & 2.91 s & $\Theta(0.875^n)$ & $O(0.8751^n)$ & 0.30 s \\
\bottomrule
\end{tabular}
\end{table}

\paragraph{Methodology}
Beyond the mere \emph{existence}, we are also interested in the \emph{quality} of the bounds.
To this end, we ran Diabolo again on each benchmark where a bound could be found, once optimizing the bound on the expected value, once optimizing the tail bound.
The configuration options were left at the default settings with an unrolling limit of 30 for most benchmarks.
For expectation bounds, we modified the settings for 23 benchmarks, to adjust the unrolling limit (reported in \cref{table:quality-bounds}), invariant size to 2, 3, or 4 (if smaller values failed or yielded bad bounds), and to skip the ADAM-based optimizer (due to it being slow for some benchmarks).
For tail bounds, we set the unrolling limit to 1 and modified the default settings for 15 benchmarks to adjust the invariant size to 2, 3, or 4 and to skip the ADAM-based optimizer (due to it being slow for some benchmarks).

\paragraph{Results}
The results are shown in \cref{table:quality-bounds}.
We mark with ``\dag'' all examples that, to our knowledge, could not be solved automatically before, in the sense of bounding the distribution's moments and tails without user intervention.
Note that no other existing tool is able to obtain \emph{exponential} tail bounds like ours on any benchmark (unless the tails are zero).
One can see that all bounds are nontrivial and the upper and lower bounds on the expected value are usually close together.
However, it seems to be harder to find good bounds for benchmarks where the tails of the distribution decay slowly, such as the asymmetric random walks (ours/*-asym-rw).
Most of the tail bounds are also very close to the theoretical optimum where the exact tail bound could be manually determined.
This demonstrates that our geometric bound semantics generally yields useful bounds.

\subsection{Comparison with Previous Work}
\label{sec:eval-prev}

\paragraph{Comparison with GuBPI}
Our residual mass semantics shares many characteristics with GuBPI (see \cref{table:compare-related}).
However, GuBPI is designed for a more general setting with continuous sampling and soft conditioning.
As a consequence, when applied to discrete probabilistic programs with only hard conditioning, there is a lot of overhead.
To demonstrate this, we ran both tools on \cref{example:geo-counter,example:asym-rw,ex:die-paradox}, configured to produce the same bounds.
The results (\cref{table:compare-gubpi}) show that Diabolo's residual mass semantics is several orders of magnitude faster than GuBPI.

\Citet{WangYFLO24} subsequently improved on GuBPI's results, but we were unable to compare with their system because they use proprietary software.
However, they report performance improvements of at most a factor of 15 over GuBPI, which is insufficient to bridge the factor of over 100 for \cref{example:geo-counter} between GuBPI and our residual mass semantics, let alone the factor of over $10^5$ for \cref{ex:die-paradox}.
This demonstrates that the residual mass semantics is more effective than existing tools for discrete probabilistic programs.

\begin{table}
\centering
\footnotesize
\caption{Comparison of the running time of GuBPI \citep{BeutnerOZ22} and our residual mass semantics (as implemented in Diabolo) to produce the same bounds\vspace{-0.5em}}
\label{table:compare-gubpi}
\alternaterowcolors
\begin{tabular}{lrrr}
\toprule
Benchmark & GuBPI & Residual mass semantics & Speedup factor \\
\midrule
Simple counter (\cref{example:geo-counter}) &0.7 s &0.006 s & $1.2 \cdot 10^2$ \\
Asymmetric random walk (\cref{example:asym-rw}) &90 s &0.002 s & $4.5 \cdot 10^4$ \\
Mossel's die paradox (\cref{ex:die-paradox}) &156 s &0.0008 s & $2 \cdot 10^5$ \\
\bottomrule
\end{tabular}
\vspace{-0.5em}
\end{table}

\begin{table}
\centering
\footnotesize
\caption{Comparison of Diabolo with Polar \citep{MoosbruggerSBK22} on the benchmarks from \cref{table:quality-bounds} where Polar can compute the expected value (EV).
Diabolo's results are the same as in \cref{table:quality-bounds}, but listed again for an easier comparison.
(t/o: timeout of 5 minutes exceeded.)
\vspace{-0.5em}}
\label{table:compare-polar}
\rowcolors{2}{gray!20}{}
\begin{tabular}{l|cc|cc}
\toprule
\multirow{2}{*}{Benchmark} & \multicolumn{2}{c|}{Polar} & \multicolumn{2}{c}{Diabolo (ours)} \\
& Exact EV & Time & EV bound & Time \\
\midrule
polar/coupon\_collector2 & 2 & 0.26 s & [1.999, 2.001] & 0.16 s \\
polar/fair\_biased\_coin & 1/2 & 0.41 s & [0.4999, 0.5001] & 0.03 s \\
polar/las\_vegas\_search & \xmark{} & t/o & [19.98, 26.79] & 2.29 s \\
polar/simple\_loop & 13/10 & 0.20 s & [1.299, 1.301] & 0.02 s \\
prodigy/geometric & 2 & 0.46 s & [1.999, 2.007] & 1.09 s \\
prodigy/ky\_die & \xmark{} & t/o & [3.499, 3.501] & 0.18 s \\
psi/beauquier-etal3 & \xmark{} & t/o & \xmark{} & t/o \\
psi/cav-example7 & \xmark{} & t/o & [10.41, 10.51] & 1.71 s \\
psi/ex3 & 2/3 & 0.58 s & [0.6666, 0.6667] & 0.04 s \\
psi/ex4 & 2/3 & 0.21 s & [0.6666, 0.6667] & 0.23 s \\
psi/fourcards & \xmark{} & t/o & [0.264, 0.2648] & 0.46 s \\
psi/herman3 & 4/3 & 19.53 s & [1.333, 1.334] & 62.08 s \\
psi/israeli-jalfon3 & 2/3 & 19.40 s & [0.6666, 0.6668] & 1.71 s \\
psi/israeli-jalfon5 & \xmark{} & t/o & \xmark{} & t/o \\
ours/coupon-collector5 & 137/12 & 74.86 s & [11.41, 11.56] & 47.93 s \\
ours/geometric & 1 & 0.59 s & [0.9999, 1.006] & 0.12 s \\
ours/israeli-jalfon4 & \xmark{} & t/o & [1.457, 1.458] & 121.04 s \\
ours/sum-geos & 8 & 0.19 s & [7.998, 8.001] & 2.91 s \\
\bottomrule
\end{tabular}
\vspace{-1em}
\end{table}

\paragraph{Comparison with Polar}
We compare the geometric bound semantics (as implemented in Diabolo) with Polar \citep{MoosbruggerSBK22}, the only other tool that can be used to bound moments.
In fact, Polar can compute moments exactly, but does not bound the tail asymptotics.\footnote{
One could derive the tail asymptotic bound $O(n^{-k})$ from the $k$-th moment via Markov's inequality.
But this asymptotic bound is very weak as it applies to any distribution with a finite $k$-th moment.}
We only included benchmarks from \cref{table:quality-bounds} where Polar can compute the expected value (EV), at least in theory.
(We had to extend Polar slightly to support a geometric distribution as the initial distribution.)
Polar was run with a timeout of 5 minutes, like Diabolo.
The results (\cref{table:compare-polar}) demonstrate that Diabolo is often faster, and sometimes much faster, than Polar and applicable to more benchmarks.

\paragraph{Comparison of tail bounds}
There are several existing program analyses that can bound the tail distribution of the running time of a probabilistic program (see \cref{sec:related-work}), but only one \citep{ChatterjeeFNH16} achieves exponential bounds $O(c^{n})$ for $c < 1$ like our bounds.
This is particularly interesting as their assumptions (bounded differences and the existence of a linear ranking supermartingale) are remarkably similar to the conditions of our \cref{thm:sufficient-cond}.
Since the code for the experiments by \citet{ChatterjeeFNH16} is not available, we do a manual comparison on \cref{example:geo-counter,example:asym-rw}.
Our geometric bound semantics yields the bounds $O((\frac12 + \epsilon)^n)$ and $O((2\sqrt{r(1-r)} + \epsilon)^n)$, respectively, whereas their method yields $O((\exp(-1/2))^{n})$ and $O((\sqrt{\exp(-(1 - 2r)^2)})^n)$, where $r$ is the bias of the random walk.
It is not hard to see that our bounds are tighter (details in \apxref{apx:eval}).

\subsection{Comparison Between Our Two Semantics}
\label{sec:eval-comparison}

\begin{table}
\centering
\footnotesize
\caption{Comparison of the quality of moment and tail bounds for various examples between the residual mass and geometric bound semantics\vspace{-0.5em}}
\label{table:moment-tail-bounds}
\begin{tabularx}{\textwidth}{X|lllll}
\toprule
Example &Method &Expected value &2nd moment &Tail &Time \\
\midrule
\multirow{3}{\hsize}{Simple counter \\ ({\cref{example:geo-counter}})} &exact &1 &3 &$\Theta(0.5^n)$ & \\
 &res. mass sem. &$\ge 1 - 5 \cdot 10^{-14}$ &$\ge 3 - 3 \cdot 10^{-12}$ &n/a &0.003 s \\
 &geom. bound sem. &$\le 1.000037$ &$\le 3.00017$ &$O(0.5037^n)$ &0.092 s \\
\midrule
\multirow{3}{\hsize}{Random walk \\ (\cref{example:asym-rw})} &exact &2 &10 &$\Theta(n^{-3/2}0.8660...^n)$ & \\
 &res. mass sem. &$\ge 1.99997$ &$\ge 9.998$ &n/a &0.067 s \\
 &geom. bound sem. &$\le 2.55$ &$\le 21.8$ &$O(0.8682^n)$ &1.47 s \\
\midrule
\multirow{3}{\hsize}{Introductory ``Die paradox'' \\ (\cref{ex:die-paradox})} &exact &1.5 &3 &$\Theta(0.3333...^n)$ & \\
 &res. mass sem. &$\ge 1.5 - 2\cdot 10^{-16}$ &$\ge 3 - 4 \cdot 10^{-16}$ &n/a &0.005 s \\
 &geom. bound sem. &$\le 1.5 + 8 \cdot 10^{-7}$ &$\le 3 + 3 \cdot 10^{-6}$ &$O(0.3395^n)$ &0.097 s \\
\midrule
\multirow{3}{\hsize}{Coupon collector \\ (5 coupons)} &exact &11.41666... &155.513888... &$\Theta(0.8^n)$ & \\
 &res. mass sem. &$\ge 11.41665$ &$\ge 155.5131$ &n/a &0.88 s \\
 &geom. bound sem. &$\le 11.56$ &$\le 158.5$ &$O(0.8002^n)$ &22.6 s \\
\midrule
\multirow{3}{\hsize}{Herman's self-stabilization \\ (3 processes)} &exact &1.333... &4.222... &? & \\
 &res. mass sem. &$\ge 1.3333332$ &$\ge 4.222220$ &n/a &0.509 s \\
 &geom. bound sem. &$\le 1.3339$ &$\le 4.226$ &$O(0.5002^n)$ &34.8 s \\
\bottomrule
\end{tabularx}
\vspace{-0.5em}
\end{table}

\begin{figure}
\begin{subfigure}{0.32\textwidth}
  \centering
  \includegraphics[width=\textwidth]{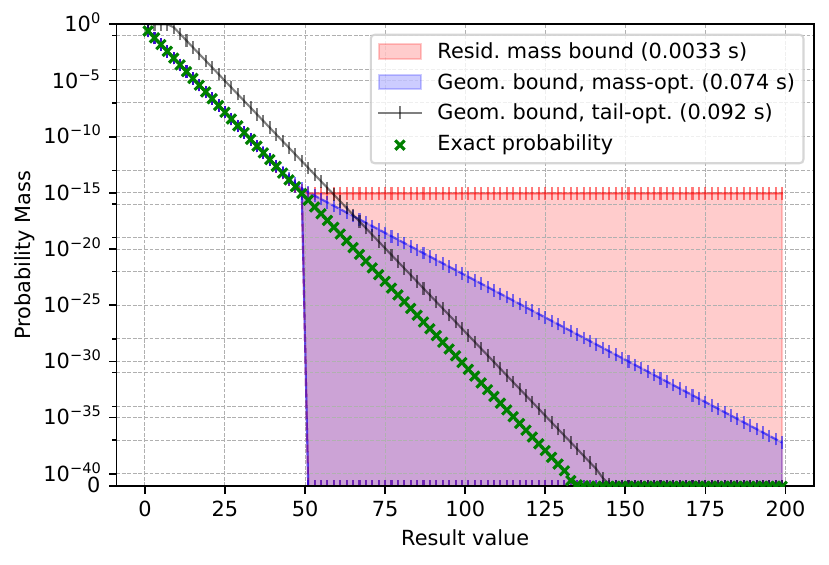}\vspace{-0.5em}
  \caption{Geom. counter (\cref{example:geo-counter})}
  \label{fig:rest-vs-geom-geo}
\end{subfigure}
\hfill
\begin{subfigure}{0.32\textwidth}
  \centering
  \includegraphics[width=\textwidth]{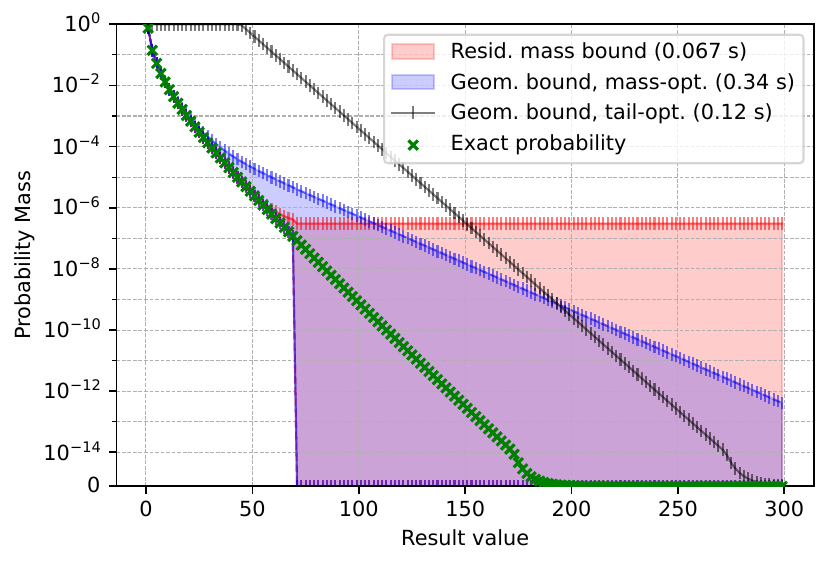}\vspace{-0.5em}
  \caption{Asym. rand. walk (\cref{example:asym-rw})}
  \label{fig:rest-vs-geom-asym-rw}
\end{subfigure}
\hfill
\begin{subfigure}{0.32\textwidth}
  \centering
  \includegraphics[width=\textwidth]{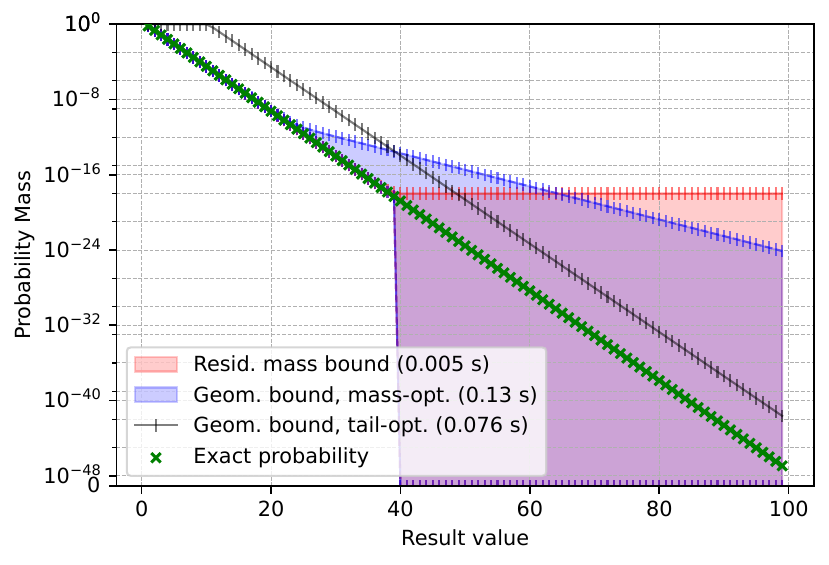}\vspace{-0.5em}
  \caption{Die paradox (\cref{ex:die-paradox})}
  \label{fig:rest-vs-geom-die-paradox}
\end{subfigure}

\hfill
\begin{subfigure}{0.49\textwidth}
  \centering
  \includegraphics[width=0.64\textwidth]{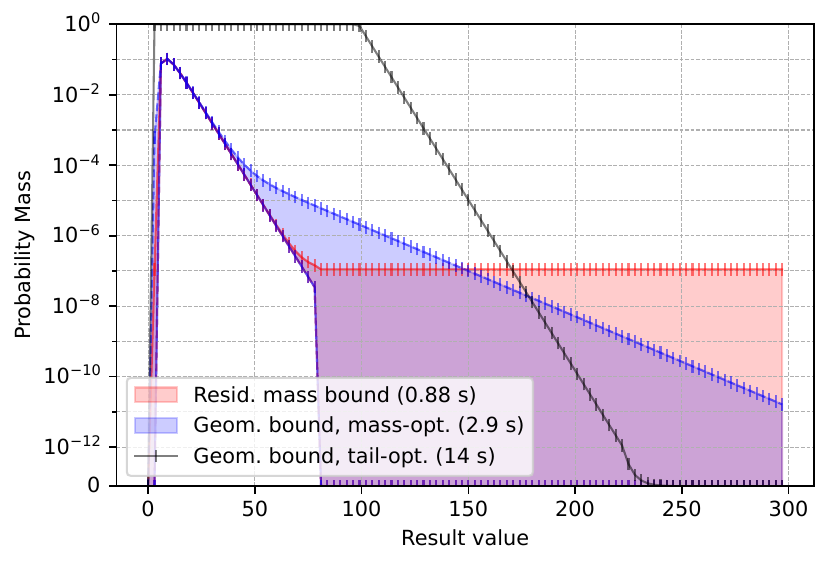}\vspace{-0.5em}
  \caption{Coupon collector problem with 5 coupons}
  \label{fig:rest-vs-geom-coupon-collector}
\end{subfigure}
\hfill
\begin{subfigure}{0.49\textwidth}
  \centering
  \includegraphics[width=0.64\textwidth]{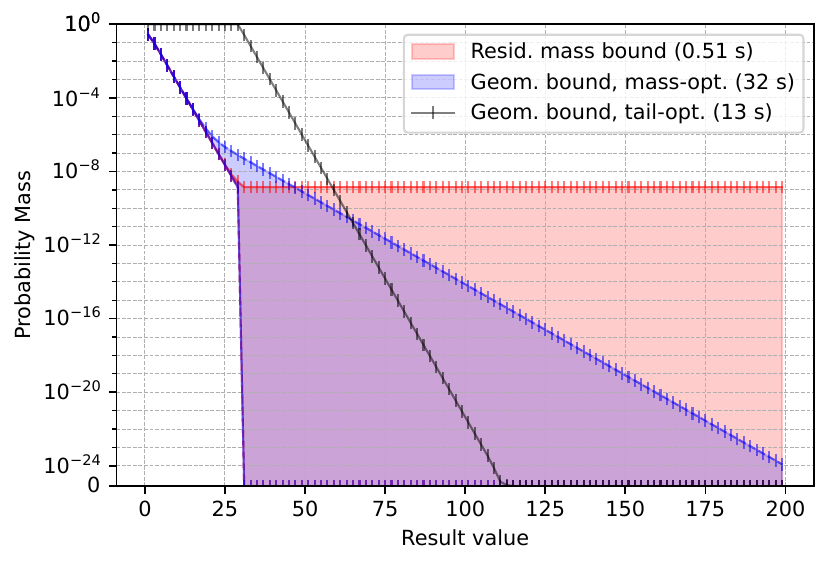}\vspace{-0.5em}
  \caption{Herman's self-stabilization with 3 processes}
  \label{fig:rest-vs-geom-herman3}
\end{subfigure}
\hfill{}\vspace{-0.5em}
\caption{Comparison of the residual mass and geometric bound semantics.
Note that the probability masses (y-axis) are on a logarithmic scale, except for the lowest part, which is linear so as to include 0.}
\label{fig:rest-vs-geom}
\vspace{-1em}
\end{figure}

We compare the \emph{quality} of the bounds from the two semantics on 5 examples: the simple counter (\cref{example:geo-counter}), the asymmetric random walk (\cref{example:asym-rw}), the introductory example (\cref{ex:die-paradox}), the coupon collector problem with 5 coupons (ours/coupon-collector5), and Herman's self-stabilization protocol with 3 processes (adapted from psi/herman3).
They were run with a loop unrolling limit of 50, 70, 40, 80, and 30 respectively, except for the geometric tail bounds where the limit was set to 1.
We report bounds on the first two (raw) moments and the tails (\cref{table:moment-tail-bounds}).
Note that the residual mass semantics (not just the lower bound semantics) is needed for lower bounds on moments because they require upper bounds on the normalizing constant.
The optimization objective for the EGD bounds was the expected value for the moment bounds and the decay rate for the tail bound.
The residual mass semantics is faster and finds good lower bounds on the moments, but cannot find upper bounds or tail bounds.
The geometric bound semantics can find upper bounds on both, but takes much more time than the residual mass semantics.

The bounds on the probability masses can be found in \cref{fig:rest-vs-geom}.
We report EGD bounds for two optimization objectives: the total probability mass (``mass-optimized'') and the decay rate (``tail-optimized'').
The residual mass semantics is faster and yields tighter bounds for the probability masses on small values, but the bound is flat, i.e. the difference between upper and lower bounds is constant.
In contrast, the geometric bound semantics finds upper bounds that keep decreasing like a geometric distribution even for large numbers beyond the unrolling limit.
On the other hand, its upper bound is somewhat worse for small values than the residual mass semantics.
This is because the geometric bounds arise from the contraction invariant, which requires a uniform decrease of the distribution.
So the geometric bounds need to lose precision either for large values (when minimizing the total mass bound), or small values (when minimizing the tail bounds).
In fact, the tail-optimized bound is much worse for small values, but eventually becomes almost parallel to the exact solution, since their asymptotics are almost equal up to a constant factor.

\section{Related Work}
\label{sec:related-work}

A summary of the most relevant work is given in \cref{table:compare-related}: guaranteed bounds \citep{BeutnerOZ22,WangYFLO24}, exact Bayesian inference with loops \citep{KlinkenbergBCHK24}, and exact loop analysis with moments \citep{MoosbruggerSBK22}.
This section presents a more detailed account.

\paragraph{Guaranteed bounds}
The two most directly related pieces of work are \citet{BeutnerOZ22} and \citet{WangYFLO24}.
\citet{BeutnerOZ22} compute guaranteed bounds on posterior probabilities: they partition the trace space (the space of all sampled values during program execution) into boxes or, more generally, polytopes, in such a way that they can get upper and lower bounds on the likelihood in each partition.
They also present an interval type system to overapproximate recursion.
For discrete probabilistic programs, the effective bounds are comparable to our residual mass semantics (\cref{sec:eval-prev}), but their computation is slower and they are not proven to converge.

\citet{WangYFLO24} introduce a new approach to bounding fixed points and can even find nontrivial bounds for ``score-recursive'' programs, i.e. programs with loops where the likelihood can increase in each iteration, and may be unbounded.

Both methods focus on continuous distributions and obtain bounds on the posterior probability of the result $X$ of a program being in a certain interval, i.e. bounds of the form $\Prob[X \in [a, b]] \in {[p, q]} \subseteq [0, 1]$.
If $X$ has infinite support, this is not enough to bound the moments or tail probability asymptotics.
Our geometric bound semantics bounds the \emph{whole distribution} and therefore does not suffer from these restrictions.
On the other hand, our programming language is more restricted and we do not support continuous distributions.

\paragraph{Verified samplers}
Another approach to ensuring correctness of approximate inference is to verify the sampler itself.
The Zar system compiles discrete probabilistic programs with loops to provably correct samplers that can be used to form simple Monte Carlo estimates of conditional probability masses and moments and is fully verified in the Coq proof assistant \citep{BagnallSB23}.

\paragraph{Exact Bayesian inference}
Exact inference is intractable in general because it requires analytical solutions to infinite sums or integrals \citep{GehrMV16}.
Thus exact inference systems either have to restrict programs to a tractable subclass or may fail on some inputs.
In the former category are Dice \citep{HoltzenBM20}, which only supports finite discrete distributions, and SPPL \citep{SaadRM21}, which supports continuous distributions, but imposes restrictions on their use, and Genfer \citep{ZaiserMO23}, which allows some infinite-support distributions but restricts operations on variables.
In the latter category are the systems PSI \citep{GehrMV16} and Hakaru \citep{NarayananCRSZ16}, which rely on computer algebra to find a closed-form solution for the posterior.

We are aware of only one piece of work that tackles exact inference for loops \citep{KlinkenbergBCHK24}.
They build on the idea to use generating functions as an exact representation of distributions with infinite support \citep{KlinkenbergBKKM20}.
\Citet{ZaiserMO23} combine this idea with automatic differentiation to perform Bayesian inference on loop-free probabilistic programs and can even support continuous sampling (but only discrete observations).
Back in the fully discrete setting, \citet{KlinkenbergBCHK24} can perform exact Bayesian inference on a \emph{single probabilistic loop without nesting}, but only if a loop invariant template is provided, which specifies the shape of the loop invariant, but may contain holes that have to be filled with real numbers.
They can synthesize these numbers and verify the resulting invariant, from which it is easy to obtain the exact posterior distribution, parameterized in the program inputs.
However, the hard problem of finding the right shape of the loop invariant falls on the user.
In fact, even if a loop invariant exists, it cannot always be specified in their language \citep[Example 25]{KlinkenbergBCHK24}.
In contrast, our method is fully automatic, so it does not suffer from these issues.
As a trade-off, our method is not exact and only works on fixed program inputs.

\paragraph{Exact loop analysis with moments}
There has been some work on studying probabilistic loop behavior by synthesizing invariants for the moments of the random variables.
\citet{BartocciKS19} consider Prob-solvable loops, which do not have a stopping condition and whose body consists of polynomial assignments with random coefficients.
They obtain invariants for the moments of the program variables by solving a system of recurrence equations arising from the fixed structure of the loop body.
\citet{AmrollahiBKKMS22} extend this method to a larger class of loops, with fewer restrictions on variable dependencies.
\citet{MoosbruggerSBK22} support (restricted) if-statements and state-dependent distribution parameters.
They use the moments to bound tail probabilities and can reconstruct the probabilities masses of the distribution if it has finite support.
\citet{KofnovMSBB24} allow Prob-solvable loops with exponential and trigonometric updates in the loop body, by leveraging the characteristic function.

Our approach differs from the above-mentioned pieces of work in several regards.
On the one hand, these approaches yield exact moments and can handle continuous distributions and more primitive operations (e.g. multiplication).
On the other hand, they do not support branching on variables with infinite support, nested loops, or conditioning.
Our method can handle these constructs, and bounds the whole distribution, not just the moments.
As a consequence, we get bounds on probability masses and exponential bounds on the tails, which are better than the polynomial bounds that can be derived from the moments in previous works.

\paragraph{Cost analysis \& martingales}
Martingales are a common technique for cost analysis of probabilistic programs, i.e. computing bounds on the total cost of a program.
The cost can be running time, resource usage, or a utility/reward.
A martingale is a sequence of random variables $X_0, X_1, \dots$ such that $\ExpVal[X_i] < \infty$ and $\ExpVal[X_{i+1} \mid X_1, \dots, X_i] = X_i$ for all $i \in \NN$.
A supermartingale only requires $\ExpVal[X_{i+1} \mid X_1, \dots, X_i] \le X_i$.
A ranking supermartingale (RSM) requires $\ExpVal[X_{i+1} \mid X_1, \dots, X_i] \le X_i - \epsilon \cdot [X_n > 0]$ for some $\epsilon > 0$.
(There are small variations in the definitions.)

\citet{ChakarovS13} introduce the concept of a RSM to prove almost sure termination of probabilistic programs.
\citet{ChatterjeeFNH16} extend RSMs to probabilistic programs with nondeterminism and use them to derive expectation and tail bounds on the running time $T$.
Then the program's expected termination time is $\ExpVal[T] \le \frac{\ExpVal(X_0)}{\epsilon}$.
Assuming the martingale is difference-bounded (i.e. $|X_{i+1} - X_i| < C$ for all $i$), they derive exponential tail bounds on the running time ($\Prob[T = n] = O(c^n)$) from Azuma's inequality for supermartingales.

\citet{KuraUH19} also deal with termination probabilities, but extend the method to higher moments $\ExpVal[T^k]$ of the running time.
Via Markov's inequality, they obtain polynomial tail bounds on the termination time ($\Prob[T = n] = O(n^{-k}))$, which are asymptotically weaker than \citet{ChatterjeeFNH16}'s exponential bounds, but more generally applicable (no need for bounded differences).
\citet{WangHR21} improve on this by bounding \emph{central} moments, from which they obtain better (but still polynomial) tail bounds than \citet{KuraUH19}.
In subsequent work by \citet{NgoCH18}, \citet{WangHGCQS19}, and \citet{ChatterjeeGMZ24}, the focus shifts to more general costs than running time and relaxing requirements around bounded updates and nonnegativity of costs.

Martingale-based methods can handle continuous distributions and more primitive program operations (e.g. multiplication).
Their bounds are also typically parameterized in the program inputs, whereas our bounds only apply to fixed inputs.
Conversely, our method yields bounds not just on one program variable (e.g. the running time), but on the distribution of all variables; we support conditioning; and our bounds can be made arbitrarily tight (\cref{thm:convergence-residual-mass-bounds,thm:convergence}).

\section{Conclusion}
\label{sec:conclusion}

We advance the concept of guaranteed bounds for discrete probabilistic programs.
Such bounds are both more automatable and generally applicable than exact methods, as well as providing more guarantees than sampling-based methods.
Our residual mass semantics finds flat tail bounds, improving upon the state of the art in terms of simplicity, speed, and provable guarantees.
Our geometric bound semantics adds two novel concepts (contraction invariants and EGDs) to the toolbox of automated probabilistic program analysis, which currently relies predominantly on martingales, moment analysis, and generating functions.
Since this semantics overapproximates the \emph{whole distribution} (as opposed to individual probabilities or moments), it can bound probability masses, moments, and tails at the same time.
Both semantics have desirable theoretical properties, such as soundness and convergence, and our empirical studies demonstrate their applicability and effectiveness on a variety of benchmarks.

\subsection{Future Work}
\label{sec:future-work}

It should be possible to handle (demonic) nondeterminism, as this corresponds to finding a bound $\nu$ on the maximum of two EGD bounds $\mu_1, \mu_2$, which can be specified using inequalities $\mu_1 \mle \nu$ and $\mu_2 \mle \nu$.
One could also try to generalize the shape of upper bounds from EGDs to a more expressive subclass of discrete phase-type distributions.
Furthermore, it would be interesting to illuminate the theoretical connections with martingale analysis and the quality of the derivable bounds, which was investigated empirically in \cref{sec:eval-prev}.

\paragraph{Negative numbers}
Our language is restricted to variables in $\NN$.
The residual mass semantics can easily be extended to $\ZZ$ since it only requires an exact semantics for the loop-free fragment, which is straightforward for distributions with finite support.
We believe that the geometric bound semantics can also be extended to $\ZZ$, but this is more involved.
Even for the simple case of a single variable, one EGD is not enough, but we need two: one for the negative part and one for the nonnegative part.
In the case of $n$ variables, we need one $\egd{\P^{\oo}}{\balpha^{\oo}}$ for each orthant $\oo \in \{ +, - \}^n$.
We do not foresee any serious conceptual difficulties extending the geometric bound semantics to operate on such a family of orthant EGDs, but the notation and proofs become quite cumbersome because $\P$ takes \emph{two} multi-indices, not just one.
In the interest of clarity and conceptual simplicity, we decided not to pursue this idea further.

\paragraph{Continuous distributions}
Variables with values in $\RR$ that can be sampled from continuous distributions are much harder to support.
The residual mass semantics can build on the exact semantics for the loop-free fragment by \citet{ZaiserMO23} to achieve partial support for continuous distributions and variables in $\nnr$ \citep[Chapter 5]{Zaiser24}.
For the geometric bound semantics, it may be possible to support continuous distributions by extending the exponential distribution to ``eventually exponential distributions'' (analogously to EGDs extending geometric distributions) which would consist of an initial block using, e.g. a polynomial bound on the probability density function on a compact interval $[-a, a]$ of $\RR$, and exponential distributions on $(-\infty, -a]$ and $[a, \infty)$.
Whether and how this could work in detail is unclear and requires further investigation.

\clearpage

\section*{Data Availability Statement}

The artifact for this paper (consisting of the Diabolo tool and benchmarks) is archived on Zenodo \citep{Zaiser24artifact}.
The latest version is available on Github: \url{https://github.com/fzaiser/diabolo}.

\begin{acks}
We would like to thank Maria Craciun, Eva Darulova, Joost-Pieter Katoen, Lutz Klinkenberg, Sam Staton, Dominik Wagner, Peixin Wang, Đorđe Žikelić, and the anonymous reviewers for their helpful comments and discussions.

This research was supported by the Engineering and Physical Sciences Research Council (studentship 2285273, grant EP/T006579) and the National Research Foundation, Singapore, under its RSS Scheme (NRF-RSS2022-009). 
\end{acks}

\bibliographystyle{ACM-Reference-Format}
\bibliography{references}

\ifarxiv

\clearpage

\appendix

\section{Supplementary Material for \Cref{sec:unrolling-bounds}}
\label{apx:unrolling-bounds}

\UnrollingSemantics*
\begin{proof}
The proof is a straightforward induction on the structure of the program $P$.
The only nontrivial case is $P = \Whilst{E}{P_1}$, where we use the fact that the unrolling operator $\Phi_{E,P_1}$ corresponds to unrolling the loop syntactically.
More specifically, we claim that $\esem{(\Whilst{E}{P_1})_v^{(u)}} = \esem{\Whilst{E}{P_1}}$ for all $v \le u$.
To see this, note that the unrolling operators agree: $\Phi_{E,P_1^{(u)}} = \Phi_{E,P_1}$ because by the induction hypothesis, we have
\[ \Phi_{E,P_1^{(u)}}(\psi)(\mu) = \mu_\lightning + \mu|_{\lnot E} + \psi(\esem{P_1^{(u)}}(\mu|_E)) = \mu_\lightning + \mu|_{\lnot E} + \psi(\esem{P_1}(\mu|_E)) = \Phi_{E,P_1}(\psi)(\mu) \]
Thus the case $v = u$ follows from
\[ \esem{(\Whilst{E}{P_1})_u^{(u)}} = \esem{\Whilst{E}{P_1^{(u)}}} = \lfp{\Phi_{E,P_1^{(u)}}} = \lfp{\Phi_{E,P_1}} = \esem{\Whilst{E}{P_1}} \]
For $v < u$, we find by induction (decreasing $v$):
\begin{align*}
  \esem{(\Whilst{E}{P_1})_{v}^{(u)}}(\mu) &= \esem{\Ite{E}{P_1^{(u)}; (\Whilst{E}{P_1})_{v+1}^{(u)}}{\Skip}}(\mu) \\
  &= \mu_\lightning + \esem{(\Whilst{E}{P_1})_{v+1}^{(u)}}(\esem{P_1^{(u)}}(\mu|_E)) + \mu|_{\lnot E} \\
  &= \Phi_{E, P_1^{(u)}}(\esem{(\Whilst{E}{P_1})_{v+1}^{(u)}})(\mu) \\
  &= \Phi_{E, P_1^{(u)}}(\esem{\Whilst{E}{P_1}})(\mu) \qquad \text{by ind. hyp.} \\
  &= \Phi_{E, P_1}(\esem{\Whilst{E}{P_1}})(\mu) \\
  &= \Phi_{E, P_1}(\lfp{\Phi_{E, P_1}})(\mu) \\
  &= \lfp{\Phi_{E, P_1}}(\mu) \\
  &= \esem{\Whilst{E}{P_1}}(\mu) \qedhere
\end{align*}
\end{proof}

\SoundnessLowerBounds*
\begin{proof}
The proof is by induction on the structure of the program $P$.
For $P = P_1; P_2$, we have
\[ \semlo{P}(\mu) = \semlo{P_2}(\semlo{P_1}(\mu)) \mle \esem{P_2}(\semlo{P_1}(\mu)) \mle \esem{P_2}(\esem{P_1}(\mu)) = \esem{P}(\mu) \]
by the induction hypothesis and the monotonicity of the standard semantics.
For conditionals $P = \Ite{E}{P_1}{P_2}$, we find
\[ \semlo{P}(\mu) = \mu_\lightning + \semlo{P_1}(\mu|_E) + \semlo{P_2}(\mu|_{\lnot E}) \mle \mu_\lightning + \esem{P_1}(\mu|_E) + \esem{P_2}(\mu|_{\lnot E}) = \esem{P}(\mu) \]
by the induction hypothesis.
For loops $P = \Whilst{E}{P_1}$, we have $\semlo{P}(\mu) = \zero \mle \esem{P}(\mu)$ by definition.
In all other cases, $\semlo{P}(\mu)$ and $\esem{P}(\mu)$ are equal by definition.
\end{proof}

\ConvergenceLowerBounds*
\begin{proof}
The proof is by induction on the structure of the program $P$.
Let $P^{(u)}_*$ denote the program $P^{(u)}$ with all loops replaced by $\Diverge := \Whilst{\Flip(1)}{\Skip}$.
Note that $P^{(u)}_*$ can be defined inductively just like $P^{(u)}$.
It is straightforward to show that $\semlo{P^{(u)}} = \esem{P_*^{(u)}}$ because $\esem{\Diverge} = \zero$.
We claim that $(\esem{P^{(u)}_*}(\mu))_{u \in \NN}$ is an $\omega$-chain (monotone sequence of measures) with supremum $\esem{P}(\mu)$.
This implies $\semlo{P^{(u)}}(\mu)(\estates) \to \esem{P}(\mu)(\estates) \le \mu(\estates) < \infty$ and the convergence is monotone.
Therefore the total variation distance between the two measures converges to 0 as well.

For simple statements with $P = P_*^{(u)}$, the claim is clear.
The remaining cases are the following.

\paragraph{Sequential composition ($P = P_1; P_2$)}
We have $\nu_u := \esem{P_*^{(u)}}(\mu) = \esem{{P_1}_*^{(u)}}(\esem{{P_2}_*^{(u)}}(\mu))$.
This is clearly monotone in $u$ and $\esem{P}(\mu)$ is an upper bound because
\[ \nu_u := \esem{P_*^{(u)}}(\mu) = \esem{{P_1}_*^{(u)}}(\esem{{P_2}_*^{(u)}}(\mu)) \mle \esem{{P_1}_*^{(u)}}(\esem{P_2}(\mu)) \mle \esem{P_1}(\esem{P_2}(\mu)) = \esem{P}(\mu) \]
To show that $\esem{P}(\mu)$ is the least upper bound, let $\nu'_{u,v} := \esem{{P_1}_*^{(u)}}(\esem{{P_2}_*^{(v)}}(\mu))$.
We claim that any $\nu$ is an upper bound on $\{ \nu'_{u,v} \mid u, v \in \NN \}$ if and only if it is an upper bound on $\nu_u$.
The forward direction is immediate.
For the reverse direction, let $\nu$ be an upper bound on $\nu_u$ for all $u \in \NN$.
Then $\nu$ is an upper bound on $\nu'_{u,v}$ for all $u, v \in \NN$ because $\nu'_{u,v} \mle \nu_{\max(u,v)} \mle \nu$.
By monotonicity, we have
\begin{align*}
  \sup_{u \in \NN} \esem{P_*^{(u)}}(\mu) &= \sup_{u \in \NN} \nu_u \\
  &= \sup_{u \in \NN, v \in \NN} \nu'_{u,v} \\
  &= \sup_{u \in \NN} \sup_{v \in \NN} \esem{{P_1}_*^{(u)}}(\esem{{P_2}_*^{(v)}}(\mu)) \\
  &= \sup_{u \in \NN} \esem{{P_1}_*^{(u)}}(\esem{P_2}(\mu)) \\
  &= \esem{P_1}(\esem{P_2}(\mu)) \\
  &= \esem{P}(\mu)
\end{align*}

\paragraph{Branching ($P = \Ite{E}{P_1}{P_2}$)}
We have $\esem{P_*^{(u)}}(\mu) = \mu_\lightning + \esem{{P_1}_*^{(u)}}(\mu|_E) + \esem{{P_2}_*^{(u)}}(\mu|_{\lnot E})$.
This is clearly monotone in $u$ by the induction hypothesis and by the linearity of suprema, we find that the supremum is in fact $\mu_\lightning + \esem{P_1}(\mu|_E) + \esem{P_2}(\mu|_{\lnot E}) = \esem{P}(\mu)$.

\paragraph{Loops ($P = \Whilst{E}{P_1}$)}
By induction hypothesis $\esem{{P_1}_*^{(u)}}(\mu)$ is an $\omega$-chain with supremum $\esem{P_1}(\mu)$.
We use the standard result that the least fixpoint $\lfp{\Phi_{E, P_1}}$ is the supremum of the $\omega$-chain $\Phi_{E, P_1}^n(\zero)$ in $n$.
We first claim that $\Phi_{E, {P_1}_*^{(u)}}(\psi)$ is an $\omega$-chain in $u$ for any $\omega$-continuous $\psi$, with supremum $\Phi_{E, P_1}(\psi)$.
This follows inductively from
\begin{align*}
\sup_{u \in \NN} \Phi_{E, {P_1}_*^{(u)}}(\psi)(\mu) &= \mu_\lightning + \mu|_{\lnot E} + \sup_{u \in \NN} \psi(\esem{{P_1}_*^{(u)}}(\mu|_E)) \\
&\mle \mu_\lightning + \mu|_{\lnot E} + \psi(\sup_{u \in \NN} \esem{{P_1}_*^{(u)}}(\mu|_E)) &&\text{by $\omega$-continuity} \\
&\mle \mu_\lightning + \mu|_{\lnot E} + \psi(\esem{P_1}(\mu|_E)) &&\text{by induction hypothesis} \\
&= \Phi_{E, P_1}(\psi)(\mu)
\end{align*}
We thus find:
\begin{align*}
\sup_{u \in \NN} \esem{{P}_*^{(u)}} &= \sup_{u \in \NN} \sup_{n \in \NN} \Phi_{E, {P_1}_*^{(u)}}^n(\zero) &&\text{characterization of least fixpoint} \\
&= \sup_{n \in \NN} \sup_{u \in \NN} \Phi_{E, {P_1}_*^{(u)}}^n(\zero) &&\text{swap suprema} \\
&= \sup_{n \in \NN} \Phi_{E, P_1}^n(\zero) &&\text{above claim} \\
&= \lfp{\Phi_{E, P_1}}(\zero) &&\text{characterization of least fixpoint} \\
&= \esem{P}(\mu) \qedhere
\end{align*}
\end{proof}

\SoundnessResidualMassBounds*
\begin{proof}
For the first part, we note that
\begin{align*}
    \esem{P}(\mu)(S) &= \semlo{P}(\mu)(S) + (\esem{P}(\mu) - \semlo{P}(\mu))(S) \\
    &\le \semlo{P}(\mu)(S) + (\esem{P}(\mu) - \semlo{P}(\mu))(\estates) \\
    &\le \semlo{P}(\mu)(S) + (\mu(\estates) - \semlo{P}(\mu)(\estates)) \\
    &= \semlo{P}(\mu)(S) + \semres{P}(\mu)
\end{align*}
by monotonicity of measures and \cref{lem:stdsem-properties}.
For the second part we use soundness of lower bounds and the residual mass bound and the inequality $\semlo{P}(\mu)(\lightning) \le \esem{P}(\mu)(\lightning) \le \semlo{P}(\mu)(\lightning) + \semres{P}(\mu)$:
\[ \frac{\semlo{P}(\mu)(S)}{1 - \semlo{P}(\mu)(\lightning)} \le \frac{\esem{P}(\mu)(S)}{1 - \esem{P}(\lightning)} = \Normalize(\esem{P}(\mu))(S) \le \frac{\semlo{P}(\mu)(S) + \semres{P}(\mu)}{1 - \semlo{P}(\mu)(\lightning) - \semres{P}(\mu)} \qedhere \]
\end{proof}

\ConvergenceResidualMassBounds*
\begin{proof}
The first part follows from:
\begin{align*}
\semres{P^{(u)}}(\mu) &= \mu(\estates) - \semlo{P^{(u)}}(\mu)(\estates) \\
&= \esem{P}(\mu)(\estates) - \semlo{P^{(u)}}(\mu)(\estates) &&\text{by almost sure termination} \\
&\to 0 &&\text{by \cref{thm:convergence-lower-bounds}}
\end{align*}
Together with $\semlo{P^{(u)}}(\mu) \to \esem{P}(\mu)$ (by \cref{thm:convergence-lower-bounds}), this immediately implies the second part:
\begin{align*}
\frac{\semlo{P^{(u)}}(\mu)(S)}{1 - \semlo{P^{(u)}}(\mu)(\lightning)} &\to \frac{\esem{P}(\mu)(S)}{1 - \esem{P}(\mu)(\lightning)} = \Normalize(\esem{P}(\mu))(S) \\
\frac{\semlo{P^{(u)}}(\mu)(S) + \semres{P^{(u)}}(\mu)}{1 - \semlo{P^{(u)}}(\mu)(\lightning) - \semres{P^{(u)}}(\mu)} &\to \frac{\esem{P}(\mu)(S)}{1 - \esem{P}(\mu)(\lightning)} = \Normalize(\esem{P}(\mu))(S) \qedhere
\end{align*}
\end{proof}

\section{Supplementary Material for \Cref{sec:geo-bounds}}
\label{apx:geom-bound-sem}

\paragraph{Remark on the measure vs representation of EGDs}
Strictly speaking, $\egdle$ and $\semgeo{P}$ operate on the \emph{representation} of EGDs, i.e. the parameters $\P$ and $\balpha$ of an EGD $\egd{\P}{\balpha}$.
So to be precise, we should write $(\P, \balpha) \egdle (\Q, \balpha)$ and $(\P, \balpha) \semgeo{P} (\Q, \balpha)$ instead of $\egd{\P}{\balpha} \egdle \egd{\Q}{\balpha}$ and $\egd{\P}{\balpha} \semgeo{P} \egd{\Q}{\balpha}$.
However, we went with the slightly imprecise notation because it provides the intuition that we ultimately use these relations to bound the measures.

\subsection{Detailed Example Derivations}
\label{apx:full-examples}

\paragraph{Simple geometric counter (\cref{example:geo-counter})}

We provide more details on the optimization step.
If we want to optimize the asymptotics of the tail probabilities $\Prob[X_1 = n]$ as $n \to \infty$, we want to choose $\alpha$ as small as possible, i.e. very close to $\frac{1}{2}$, accepting that this will make $\frac{1}{1-c} \ge \frac{2\alpha}{2\alpha - 1}$ very large.
Setting $\alpha = \frac{1}{2} + \epsilon$ yields $\frac{1}{2(1-c)} \ge \frac{\alpha}{2\alpha - 1} = \frac{1}{4\epsilon} + \frac{1}{2}$ and thus the bound
\[ \sem{P}(\egd{1}{0}) \mle \egd{\frac{1}{4\epsilon} + \frac{1}{2}}{\frac{1}{2} + \epsilon} \]
which implies $\Prob[X_1 = n] \le \left(\frac{1}{4\epsilon} + \frac{1}{2}\right) \left(\frac{1}{2} + \epsilon\right)^n$.

On the other hand, if we want to optimize the bound on the expected value of $X_1$, we want to minimize (by \cref{thm:egd-stats})
\[ \ExpVal_{\X \sim \sem{P}(\egd{1}{0})}[X_1] \le \ExpVal_{\X \sim \egd{\frac{1}{2(1-c)}}{\alpha}}[X_1] = \frac{\alpha}{2(1-c)(1-\alpha)^2} \]
This is minimized under the constraints $0 \le \alpha, c < 1$ and $\frac12 \le \alpha c$ for $\alpha = \frac{\sqrt5 - 1}{2} \approx 0.618$ and $c = \frac{\sqrt5 + 1}{4} \approx 0.809$.
At this point, the bound on the expected value is
\[ \ExpVal[X_1] \le \frac{11 + 5 \sqrt5}{2} \approx 11.09 \]

\paragraph{Asymmetric random walk example (\cref{example:asym-rw})}

Consider the program
\[ X_1 := 1; X_2 := 0; \Whilst{X_1 > 0}{X_2 \passign 1; \Ite{\Flip(r)}{X_1 \passign 1}}{X_1 \massign 1} \]
where the probability $r$ of going right is less than $\frac12$.
We find
\begin{align*}
  \semgeo{X_1 := 1; X_2 := 0}(\egd{1}{(0,0)}) &= \egd{\begin{pmatrix}0 \\ 1\end{pmatrix}}{(0,0)}
\end{align*}
as the distribution at the start of the loop.
Assume a $c$-contraction invariant $\egd{\P}{\balpha}$ with $\P \in \nnr^{[2] \times [1]}$ exists.
Then the bound on the distribution of the program is
\begin{align*}
  \sem{P}(\egd{1}{(0, 0)}) &\mle \evgeo{\egd{\frac{\P}{1-c}}{\balpha}}{\lnot(X_1 > 0)} = \egd{\begin{pmatrix}\tfrac{\P_{0,0}}{1-c} \\ 0\end{pmatrix}}{(0, \alpha_2)}
\end{align*}
One loop iteration transforms the contraction invariant as follows (using the strict join relation for the probabilistic branching):
\begin{align*}
&{\semgeo{X_2 \passign 1; \ProbBranch{X_1 \passign 1}{r}{X_1 \massign 1}}}\left(\evgeo{\egd{\begin{pmatrix}\P_{0,0} \\ \P_{1,0}\end{pmatrix}}{\balpha}}{X_1 > 0}\right) \\
&= {\semgeo{X_2 \passign 1; \ProbBranch{X_1 \passign 1}{r}{X_1 \massign 1}}}\left(\egd{\begin{pmatrix}0 \\ \P_{1,0}\end{pmatrix}}{\balpha}\right) \\
&= {\semgeo{\ProbBranch{X_1 \passign 1}{r}{X_1 \massign 1}}}\left(\semgeo{X_2 \passign 1}(\egd{\begin{pmatrix}0 \\ \P_{1,0}\end{pmatrix}}{\balpha})\right) \\
&= {\semgeo{\Ite{\Flip(r)}{X_1 \passign 1}{X_1 \massign 1}}}\left(\egd{\begin{pmatrix}0 & 0 \\ 0 & \P_{1,0}\end{pmatrix}}{\balpha}\right) \\
&= \egd{\begin{pmatrix}0 & (1-r)\P_{1,0} \\ 0 & (1-r)\alpha_1 \P_{1,0} \\ 0 & (1-r)\alpha_1^2 \P_{1,0} + r \P_{1,0} \end{pmatrix}}{\balpha}
\end{align*}
because
\begin{align*}
  {\semgeo{X_1 \passign 1}}\left(\evgeo{\egd{\begin{pmatrix}0 & 0 \\ 0 & \P_{1,0}\end{pmatrix}}{\balpha}}{\Flip(r)}\right)
  &= {\semgeo{X_1 \passign 1}}\left(\egd{\begin{pmatrix}0 & 0 \\ 0 & r \P_{1,0}\end{pmatrix}}{\balpha}\right) \\
  &= \egd{\begin{pmatrix}0 & 0 \\ 0 & 0 \\ 0 & r \P_{1,0}\end{pmatrix}}{\balpha} \\
  {\semgeo{X_1 \massign 1}}\left(\evgeo{\egd{\begin{pmatrix}0 & 0 \\ 0 & \P_{1,0}\end{pmatrix}}{\balpha}}{\lnot\Flip(r)}\right)
  &= {\semgeo{X_1 \massign 1}}\left(\egd{\begin{pmatrix}0 & 0 \\ 0 & (1-r)\P_{1,0}\end{pmatrix}}{\balpha}\right) \\
  &= \egd{\begin{pmatrix}0 & (1-r)\P_{1,0} \\ 0 & (1-r)\alpha_1 \P_{1,0}\end{pmatrix}}{\balpha} \\
  &= \egd{\begin{pmatrix}0 & (1-r)\P_{1,0} \\ 0 & (1-r)\alpha_1 \P_{1,0} \\ 0 & (1-r)\alpha_1^2 \P_{1,0}\end{pmatrix}}{\balpha}
\end{align*}
so their strict join is
\[ \egd{\begin{pmatrix}0 & 0 \\ 0 & 0 \\ 0 & r \P_{1,0}\end{pmatrix}}{\balpha} + \egd{\begin{pmatrix}0 & (1-r)\P_{1,0} \\ 0 & (1-r)\alpha_1 \P_{1,0} \\ 0 & (1-r)\alpha_1^2 \P_{1,0}\end{pmatrix}}{\balpha} = \egd{\begin{pmatrix}0 & (1-r)\P_{1,0} \\ 0 & (1-r)\alpha_1 \P_{1,0} \\ 0 & (1-r)\alpha_1^2 \P_{1,0} + r \P_{1,0} \end{pmatrix}}{\balpha} \]

Hence the contraction invariant has to satisfy the following requirements:
\begin{align*}
  \egd{\begin{pmatrix}0 \\ 1\end{pmatrix}}{(0, 0)} &\egdle \egd{\begin{pmatrix}\P_{0,0} \\ \P_{1,0}\end{pmatrix}}{\balpha} \\
  \egd{\begin{pmatrix}0 & (1-r)\P_{1,0} \\ 0 & (1-r)\alpha_1 \P_{1,0} \\ 0 & (1-r)\alpha_1^2 \P_{1,0} + r \P_{1,0} \end{pmatrix}}{\balpha} &\egdle \egd{\begin{pmatrix}c \P_{0,0} \\ c \P_{1,0}\end{pmatrix}}{\balpha} \\
  &= \egd{\begin{pmatrix}c \P_{0,0} & c \alpha_2 \P_{0,0} \\ c \P_{1,0} & c \alpha_2 \P_{1,0} \\ c \alpha_1 \P_{1,0} & c \alpha_1 \alpha_2 \P_{1,0} \end{pmatrix}}{\balpha}
\end{align*}
which reduce to the following polynomial constraints:
\[
  0 \le \P_{0,0} \quad
  1 \le \P_{1,0} \quad
  (1-r)\P_{1,0} \le c\alpha_2 \P_{0,0}
  \quad (1-r)\alpha_1 \P_{1,0} \le c \alpha_2 \P_{1,0}
  \quad (1-r)\alpha_1^2 \P_{1,0} + r \P_{1,0} \le c \alpha_1\alpha_2 \P_{1,0}
\]
besides the obvious ones (every variable is nonnegative and $\alpha_1, \alpha_2, c \in [0, 1)$).
The most interesting constraint is the last one, which is equivalent to
\[ (1-r) \cdot \alpha_1^2 - c\alpha_2 \cdot \alpha_1 + r \le 0 \]
Solving this for $\alpha_1$ yields:
\[ \alpha_1 \in \left[ \frac{c\alpha_2 - \sqrt{c^2\alpha_2^2 - 4r(1-r)}}{2(1-r)} , \frac{c\alpha_2 + \sqrt{c^2\alpha_2^2 - 4r(1-r)}}{2(1-r)} \right] \]
For $\alpha_1$ to exist, we need to have $1 > c\alpha_2 \ge \sqrt{4r(1-r)}$.
For $\alpha_1$ to be satisfiable is equivalent to the lower bound of the interval being less than 1, i.e. $c\alpha_2 - 2(1-r) < \sqrt{c^2\alpha_2^2 - 4r(1-r)}$.
This can only hold if $c\alpha_2 - 2(1-r) \le 0$ or
\begin{align*}
  &&(c\alpha_2 - 2(1-r))^2 &< c^2\alpha_2^2 - 4r(1-r) \\
  \iff &&-2c\alpha_2 + 4(1-r) &< -4r \\
  \iff &&c\alpha_2 &> 2
\end{align*}
a contradiction.
Hence $c\alpha_2 \le 2(1-r)$ must hold, and together with $c\alpha_2 \ge \sqrt{4r(1-r)}$ we obtain:
\begin{align*}
  \sqrt{4r(1-r)} \le 2(1-r) \iff 4r \le 4(1-r) \iff r \le \frac{1}{2}
\end{align*}
For $r = \frac12$, the constraint $\sqrt{4r(1-r)} \le c\alpha_2 < 1$ is violated.
So $r < \frac12$ is a necessary condition for the existence of a solution.
It turns out that this is sufficient: we can pick $c, \alpha_2 \in [0,1)$ with $c\alpha_2 = \sqrt{4r(1-r)}$.
Then $\alpha_1 = \frac{c\alpha_2}{2(1-r)} < 1$ is a valid solution.

In the satisfiable case $r < \frac12$, we get the following bound on the distribution of the program:
\begin{align*}
  \egd{1}{(0,0)} &\semgeo{P} \evgeo{\egd{\begin{pmatrix}\frac{\P_{0,0}}{1-c} \\ \frac{\P_{1,0}}{1-c}\end{pmatrix}}{\balpha}}{\lnot(X_1 > 0)} = \egd{\begin{pmatrix}\frac{\P_{0,0}}{1-c} \\ 0\end{pmatrix}}{(0, \alpha_2)}
\end{align*}
The asymptotic bound is $\Prob[X_2 = n] = \frac{\P_{0,0}}{1-c} \alpha_2^n$, so it is best for $\alpha_2$ as small as possible.
Since $\alpha_2 \ge \frac{\sqrt{4r(1-r)}}{c} > 2\sqrt{r(1-r)}$, the best possible geometric bound for $\Pr[X_2 = n]$ is $O((2\sqrt{r(1-r)} + \epsilon)^n)$.
We could find a bound on $\ExpVal[X_2]$ in the same way as in the previous example, but this would be very tedious to do manually.
Instead we will automate the semantics in \cref{sec:impl}.

Interestingly, our method's tail bound can get arbitrarily close (asymptotically) to the best possible geometric bound.
It is a known fact about random walks and hitting times that
\[ \Prob[X_2 = n] = \frac{1}{n}\binom{n}{(n-1)/2}\cdot r^{(n-1)/2} (1-r)^{(n+1)/2} \]
Using well-known approximations for Catalan numbers $\frac{1}{2n+1}\binom{2n+1}{n} \sim \frac{4^n}{n^{3/2}\sqrt{\pi}}$, we can compute their asymptotics for the exact probabilities (see \cref{example:asym-rw}) as follows:
\begin{align*}
  \Prob[X_2 = n] &= \frac{1}{n}\binom{n}{(n-1)/2}\cdot r^{(n-1)/2} (1-r)^{(n+1)/2} \\
  &\sim \frac{4^{(n-1)/2}}{((n-1)/2)^{3/2}\sqrt{\pi}} \cdot \frac{\sqrt{1-r}}{\sqrt{r}} \cdot (\sqrt{r(1-r)})^n \\
  &= (2\sqrt{r(1-r)})^n \cdot \frac{\sqrt{1-r}}{2 \sqrt{\pi r} ((n-1)/2)^{3/2}} \\
  &= \Theta((2\sqrt{r(1-r)})^n \cdot n^{-3/2})
\end{align*}
Thus the best possible geometric bound is $O((2\sqrt{r(1-r)})^n)$, and our method's bound can get arbitrarily close.

\subsection{Full Proofs}
\label{apx:geom-bound-proofs}

\subsubsection{Properties of EGDs}
We make explicit some properties of EGDs that were mentioned in the main text.

\begin{lemma}[Expansion preserves the measure]
\label{lem:egd-expansion}
The expansion $\egd{\Q}{\balpha}$ of $\egd{\P}{\balpha}$ with $|\P| \le |\Q|$ describes the same measure as $\egd{\P}{\balpha}$.
\end{lemma}
\begin{proof}
We expand the definitions:
\begin{align*}
\egd{\Q}{\balpha}(\x) &= \Q_{\min(\x, |\Q| - \one_n)} \balpha^{\max(\zero_n, \x - |\Q| + \one_n)} \\
&= \P_{\min(\x, |\P| - \one_n)} \balpha^{\max(\zero_n, \min(\x, |\Q| - \one_n) - |\P| + \one_n)} \balpha^{\max(\zero_n, \x - |\Q| + \one_n)} \\
&= \P_{\min(\x, |\P| - \one_n)} \cdot \balpha^{\max(\zero_n, \x - |\P| + \one_n)}
\end{align*}
because
\begin{align*}
  &\max(\zero_n, \min(\x, |\Q| - \one_n) - |\P| + \one_n) + \max(\zero_n, \x - |\Q| + \one_n) \\
  &= \max(\zero_n, \min(\x - |\Q| + \one_n, \zero_n) + |\Q| - |\P|) + \max(\zero_n, \x - |\Q| + \one_n) \\
  &= \max(\max(\zero_n, \x - |\Q| + \one_n), \max(\zero_n, \x - |\Q| + \one_n) + \min(\x - |\Q| + \one_n, \zero_n) + |\Q| - |\P|) \\
  &= \max(\max(\zero_n, \x - |\Q| + \one_n), \x - |\P| + \one_n) \\
  &= \max(\zero_n, \x - |\P| + \one_n)
\end{align*}
since $|\P| \le |\Q|$ implies $\x - |\Q| + \one_n \le \x - |\P| + \one_n$.
\end{proof}

\EgdOrder*
\begin{proof}
Let $\egd{\P'}{\balpha}$ and $\egd{\Q'}{\bbeta}$ be expansions of sizes $|\P'| = |\Q'| = \max(|\P|, |\Q|)$.
By \cref{lem:egd-expansion}, we have $\egd{\P}{\balpha} \mle \egd{\Q}{\bbeta}$ if and only if $\egd{\P'}{\balpha} \mle \egd{\Q'}{\bbeta}$.

By the definition of expansions, $\P'_\ii \le \Q'_\ii$ is equivalent to the condition from the definition of $\egdle$:
\[ \P_{\min(\ii, |\P| - \one_n)} \balpha^{\max(\ii - |\P| + \one_n)} \le \Q_{\min(\ii, |\Q| - \one_n)} \bbeta^{\max(\ii - |\Q| + \one_n)} \]
Finally, the conditions $\P'_\ii \le \Q'_\ii$ for $\ii < |\P'| = |\Q'|$ and $\balpha \le \bbeta$ imply
\begin{align*}
  \egd{\P'}{\balpha}(\ii)
  = \P'_{\min(\ii, |\P'| - 1)} \balpha^{\max(\ii - |\P'| + \one_n)}
  \le \Q'_{\min(\ii, |\Q'| - 1)} \bbeta^{\max(\ii - |\Q'| + \one_n)}
  = \egd{\Q'}{\bbeta}(\ii)
\end{align*}
for all $\ii \in \NN^n$, as desired.
\end{proof}

\EgdMarginalize*
\begin{proof}
For a measure $\mu$ on $\NN^n$, marginalizing out the $k$-th dimension yields the measure $\mu'$ with
\[ \mu'(\x[\hat{k}]) := \sum_{x_k \in \NN} \mu(x_1, \dots, x_n) \]
where $\x[\hat{k}]$ denotes the omission of $x_k$ in $\x$.
For EGDs, we find, as desired (using $|\Q| = |\P|[\hat{k}]$):
\begin{align*}
  &\sum_{j \in \NN} \egd{\P}{\balpha}(\x[k \mapsto j]) \\
  &= \sum_{j \in \NN} \P_{\min(\x[k \mapsto j], |\P| - \one_n)} \balpha^{\max(\zero_n, \x[k \mapsto j] - |\P| + \one_n)} \\
  &= \sum_{j=0}^{\infty} \P_{\min(\x[k \mapsto j], |\P| - \one_n)} \alpha_k^{\max(0, j - |\P|_k + 1)} \balpha[\hat{k}]^{\max(\zero_{n-1}, \x[\hat k] - |\Q| + \one_{n-1})} \\
  &= \left(\sum_{j=0}^{|\P|_k-2} \P_{\min(\x[k \mapsto j], |\P| - \one_n)} + \sum_{j'=0}^\infty \P_{\min(\x[k \mapsto |\P|_k - 1], |\P| - \one_n)} \alpha_k^{j'}\right) \balpha[\hat{k}]^{\max(\zero_{n-1}, \x[\hat k] - |\Q| + \one_{n-1})} \\
  &= \left(\sum_{j=0}^{|\P|_k-2} \P_{\min(\x[k \mapsto j], |\P| - \one_n)} + \frac{\P_{\min(\x[k \mapsto |\P|_k - 1], |\P| - \one_n)}}{1 - \alpha_k}\right) \balpha[\hat{k}]^{\max(\zero_{n-1}, \x[\hat k] - |\Q| + \one_{n-1})} \\
  &= \left(\sum_{j=0}^{|\P|_k-2} (\P_{k:j})_{\min(\x[\hat k], |\Q| - \one_n)} + \frac{(\P_{k: |\P|_k - 1})_{\min(\x[\hat k], |\Q| - \one_{n-1})}}{1 - \alpha_k}\right) \balpha[\hat{k}]^{\max(\zero_{n-1}, \x[\hat k] - |\Q| + \one_{n-1})} \\
  &= \Q_{\min(\x[\hat k], |\Q| - \one_n)} \balpha[\hat{k}]^{\max(\zero_{n-1}, \x[\hat k] - |\Q| + \one_{n-1})} \\
  &= \egd{\Q}{\bbeta}(\x[\hat{k}]) \qedhere
\end{align*}
\end{proof}

\EgdMoments*
\begin{proof}
The tail asymptotics are immediate from the definition of EGDs.
For the $k$-th moment, we find:
\begin{align*}
  \ExpVal_{X \sim \egd{\P}{\alpha}}[X^k] &= \sum_{j = 0}^\infty \egd{\P}{\alpha}(j) \cdot j^k \\
  &= \sum_{j=0}^{d - 1} \P_j \cdot j^k + \sum_{j = d}^\infty \P_d \cdot \alpha^{j - d} \cdot j^k \\
  &= \sum_{j=0}^{d - 1} \P_j \cdot j^k + \frac{\P_d}{1 - \alpha}\sum_{j = 0}^\infty (1 - \alpha)\alpha^j \cdot (j + d)^k \\
  &= \sum_{j=0}^{d - 1} \P_j \cdot j^k + \frac{\P_d}{1 - \alpha}\ExpVal_{Y \sim \Geometric(1 - \alpha)}[(Y + d)^k] \\
  &= \sum_{j=0}^{d - 1} \P_j \cdot j^k + \frac{\P_d}{1 - \alpha}\ExpVal_{Y \sim \Geometric(1 - \alpha)}\left[\sum_{i=0}^k \binom{k}{i} d^{k - i} Y^i\right] \\
  &= \sum_{j=0}^{d - 1} \P_j \cdot j^k + \frac{\P_d}{1 - \alpha} \sum_{i=0}^k \binom{k}{i} d^{k - i} \ExpVal_{Y \sim \Geometric(1 - \alpha)}[Y^i] \qedhere
\end{align*}
\end{proof}

\subsubsection{Complexity of the Solution Set}

\begin{figure}
\centering
\includegraphics[width=0.32\textwidth]{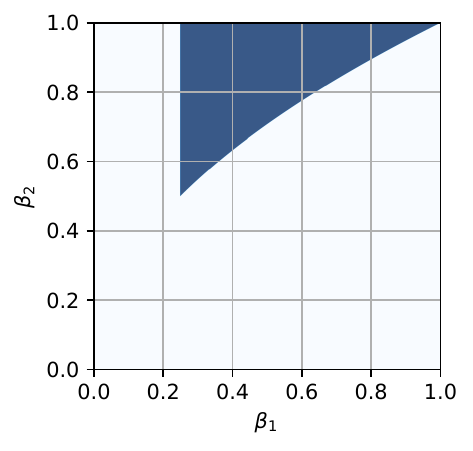}
\caption{Solution set for \cref{ex:double-geo}, projected onto $\beta_1, \beta_2$.}
\label{fig:double-feasible}
\end{figure}
One might expect that the solution set of $\bbeta$ is right-closed, i.e. that for any solution $\egd{\Q}{\bbeta}$ and any $\bgamma \ge \bbeta$, there is a solution $\egd{\Q'}{\bgamma}$, because this loosens the bound by increasing the decay rates.
Similarly, one might expect that the solution set is convex.
Both properties would be desirable because they would make the constraint solving and optimization easier.
However, the next example show that these properties do not always hold.

\begin{example}[Doubling a geometric variable]
\label{ex:double-geo}
Consider the doubling of a geometric variable with initial distribution $X_1 \sim \Geometric(\tfrac34)$ and $X_2 = 0$, in other words $\X \sim \egd{1}{(\tfrac14, 0)}$:
\[ \Whilst{X_1 > 0}{X_2 \passign 2; X_1 \massign 1 } \]
Assuming a contraction invariant $\egd{p}{\bbeta}$ with $p \in \RR \cong \RR^{[1] \times [1]}$ exists, we find:
\begin{align*}
  \semgeo{X_2 \passign 2; X_1 \massign 1}\left(\evgeo{\egd{p}{\bbeta}}{X_1 > 0}\right) &= \semgeo{X_2 \passign 2; X_1 \massign 1}\left(\egd{\begin{pmatrix} 0 \\ p \beta_1 \end{pmatrix}}{\bbeta}\right) \\
  &= \semgeo{X_1 \massign 1}\left(\egd{\begin{pmatrix} 0 & 0 & 0 \\ 0 & 0 & p \beta_1 \end{pmatrix}}{\bbeta}\right) \\
  &= \egd{\begin{pmatrix} 0 & 0 & p \beta_1 \end{pmatrix}}{\bbeta}
\end{align*}
The conditions for the contraction invariant are:
\begin{align*}
  \egd{1}{(\tfrac14, 0)} &\egdle \egd{p}{\bbeta} \\
  \egd{\begin{pmatrix} 0 & 0 & p \beta_1 \end{pmatrix}}{\bbeta} &\egdle \egd{cp}{\bbeta} = \egd{\begin{pmatrix} cp & cp\beta_2 & cp\beta_2^2 \end{pmatrix}}{\bbeta}
\end{align*}
which reduce to the following polynomial inequalities:
\begin{align*}
p \ge 1 \land \beta_1 \ge \tfrac14 \land \beta_1 < 1 \land \beta_2 \ge 0 \land \beta_2 < 1 \land cp\beta_1 \le cp\beta_2^2 \land c \ge 0 \land c < 1
\end{align*}
The constraints on $\bbeta$, after eliminating all other unknowns are equivalent to:
\begin{align*}
  \beta_1 \ge \tfrac14 \land \beta_1 < 1 \land \beta_2 > 0 \land \beta_2 < 1 \land \beta_1 < \beta_2^2
\end{align*}
The feasible values of $\beta$ are plotted in \cref{fig:double-feasible}.
Clearly the set is neither convex nor right-closed.
\end{example}

\subsubsection{Properties of the Geometric Bound Semantics}

\DecidabilityGeoBounds*
\begin{proof}
First, note that the sizes of all intermediate EGDs in each rule are determined by the input $\egd{\P}{\balpha}$, except for contraction invariants and joins, whose size needs to be fixed for this reason.
Next, we note that for $\egd{\Q}{\bbeta} := \evgeo{\egd{\P}{\balpha}}{E}$, the elements of $\Q$ and $\bbeta$ are polynomials of elements of $\P$ and $\balpha$.
In fact, only $\balpha$ occurs nonlinearly.
All this follows easily from a structural induction on $E$.

We claim that existence of $\egd{\Q}{\bbeta}$ with $\egd{\P}{\balpha} \semgeo{P} \egd{\Q}{\bbeta}$ is equivalent to a formula of the form
\[ \exists x_1,\dots,x_m \ldotp \bigwedge_{i=1}^k p_i(x_1,\dots,x_m) \bowtie_i 0 \]
where $x_1, \dots, x_m \in \nnr$, $\bowtie_i \in \{{\le}, {<}\}$ and each $p_i$ is a polynomial.
The proof is again by structural induction.

For simple statements $\Skip, \Fail, X_k \passign a, X_k \massign 1, X_k := 0$, the claim is clear because the elements of $\Q$ and $\bbeta$ are specified exactly as rational functions of elements of $\P$ and $\balpha$, so they can be encoded as
\[ \exists \Q, \bbeta \ldotp \bigwedge_{i=1}^n \beta_i = f_i(\balpha) \land \bigwedge_{\ii < |\Q|} \Q_\ii = g_\ii(\P, \balpha) \]
where $f_i$ and $g_\ii$ are rational functions.
This formula can clearly be transformed into the desired shape.
For $P = P_1; P_2$, the claim is clear by induction hypothesis.

For conditionals, we assume that the size $|\Q|$ of the join $\egd{\Q}{\bbeta}$ is fixed.
The constraints for the two branches can be expressed in the desired form by the induction hypothesis.
The constraint $(\egd{\R}{\bgamma}, \egd{\S}{\bdelta}, \egd{\Q}{\bbeta}) \in \JoinRel$ can be encoded as polynomial constraints in $\R, \S, \Q, \bgamma, \bdelta, \bbeta$, using the definition of $\JoinRel$ and expansions:
\begin{align*}
  &\bigwedge_{i=1}^n (\gamma_i \le \beta_i \land \delta_i \le \beta_i \land \beta_i < 1) \\
  &\land \bigwedge_{\ii < |\Q|} \Q_\ii = \R_{\min(\ii, |\R| - \one_n)} \bgamma^{\max(\zero_n, \ii - |\R| + \one_n)} + \S_{\min(\ii, |\S| - \one_n)} \bdelta^{\max(\zero_n, \ii - |\S| + \one_n)}
\end{align*}
which is a system of polynomial equalities.

For loop invariants, we assume that the size $|\R|$ of the contraction invariant $\egd{\R}{\bgamma}$ is fixed.
Let ${\bm s} = \max(|\P|, |\R|)$.
Then the $\egdle$-constraint $\egd{\P}{\balpha} \egdle \egd{\R}{\bgamma}$ from the semantics corresponds to the following polynomial constraints:
\[ \bigwedge_{i=1}^n \alpha_i \le \gamma_i \land \bigwedge_{\ii < {\bm s}} \P_{\min(\ii, |\P| - \one_n)} \balpha^{\max(\zero_n, \ii - |\P| + \one_n)} \le \R_{\min(\ii, |\R| - \one_n)} \bgamma^{\max(\zero_n, \ii - |\R| + \one_n)} \]
The constraint $\egd{\S}{\bdelta} \egdle \egd{c \cdot \R}{\bgamma}$ can be encoded analogously.
The remaining constraints can be encoded by the inductive hypothesis and composed into a big existential conjunction.
This proves the claim.
Note that in all the constraints, elements of initial blocks occur linearly, which is exploited in the implementation (see \cref{sec:impl-solving}).
The decidability follows from the well-known fact that the existential theory of the reals is decidable.
\end{proof}

The join relations $\JoinRel$ and $\JoinRel^*$ and the geometric bound semantics enjoy nice properties.
In particular they are monotone and multiplicative, as the following lemmas show.
As a consequence, we also show that the semantics is sound.
\begin{lemma}[Well-behavedness of the join relation]
\label{lem:join-well-behaved}
Suppose $\left(\egd{\R}{\bgamma}, \egd{\S}{\bdelta}, \egd{\Q}{\bbeta}\right) \in \JoinRel$.
Then
\begin{enumerate}
  \item if $\egd{\tilde \R}{\tilde \bgamma} \egdle \egd{\R}{\bgamma}$ and $\egd{\tilde \S}{\tilde \bdelta} \egdle \egd{\S}{\bdelta}$ then there is a $\egd{\tilde \Q}{\tilde \bbeta}$ such that $\left(\egd{\tilde \R}{\tilde \bgamma}, \egd{\tilde \S}{\tilde \bdelta}, \egd{\tilde \Q}{\tilde \bbeta}\right) \in \JoinRel$ and $\egd{\tilde \Q}{\tilde \bbeta} \egdle \egd{\Q}{\bbeta}$.
  \item for any $c > 0$, we have $\left(\egd{c \cdot \R}{\bgamma}, \egd{c \cdot \S}{\bdelta}, \egd{c \cdot \Q}{\bbeta}\right) \in \JoinRel$.
\end{enumerate}
Analogous properties hold for the strict join relation $\JoinRel^*$.
\end{lemma}
\begin{proof}
The proof is immediate from the definition of the join relations.
\end{proof}

\begin{lemma}[Well-behavedness of the geometric bound semantics]
\label{lem:geo-sem-well-behaved}
Suppose $\egd{\P}{\balpha} \semgeo{P} \egd{\Q}{\bbeta}$.
Then
\begin{enumerate}
  \item if $\egd{\tilde \P}{\tilde \balpha} \egdle \egd{\P}{\balpha}$ then there is an $\egd{\tilde \Q}{\tilde \bbeta}$ such that $\egd{\tilde \P}{\tilde \balpha} \semgeo{P} \egd{\tilde \Q}{\tilde \bbeta}$ and $\egd{\tilde \Q}{\tilde \bbeta} \egdle \egd{\Q}{\bbeta}$.
  \item for any $c > 0$, we have $\egd{c \cdot \P}{\balpha} \semgeo{P} \egd{c \cdot \Q}{\bbeta}$.
\end{enumerate}
\end{lemma}
\begin{proof}
The proof is a straightforward induction on the structure of $P$.
\end{proof}

\SoundnessGeoBounds*
\begin{proof}
We first prove soundness of the event semantics, that is $\egd{\Q}{\bbeta} := \evgeo{\egd{\P}{\balpha}}{E} = \egd{\P}{\alpha}|_E$.
First, we note that $|\Q| \ge |\P|$ because no rule reduces the size of the initial block and $\beta_i \in \{0, \alpha_i\}$ because rules can only change the decay rates by setting them to 0.
We strengthen the inductive hypothesis to include $\egd{\Q}{\bbeta} = \egd{\Q}{\balpha}$.

For $E = (X_k = a)$, we note that $\mu|_E(\x) = \mu(\x) \cdot [x_k = a]$.
Let $\P'$ be the expansion of $\P$ to size $\max(|\P|_k, a + 2)$ in dimension $k$.
In particular, $|\P'| = |\Q|$.
Then we find:
\begin{align*}
  \egd{\P}{\balpha}|_E(\x)
  &= \egd{\P'}{\balpha}|_E(\x)\\
  &= \P'_{\min(\x, |\P'| - \one_n)} \cdot \balpha^{\max(\zero_n, \x - |\P'| + \one_n)} [x_k = a] \\
  &= \left([x_k = a] \P'_{\x[k\mapsto \min(x_k, |\P'|_k - 1)]}\right) \cdot \balpha[k\mapsto 0]^{\max(\zero_n, \x - |\P'| + \one_n)} \\
  &= \Q_{\min(\x, |\Q| - \one_n)} \balpha^{\max(\zero_n, \x - |\Q| + \one_n)} \\
  &= \egd{\Q}{\bbeta}(\x)
\end{align*}
because $x_k = a$ implies $x_k < |\P'|_k$ and $\Q_\ii$ equals $[i_k = a] \P'_\ii = [i_k = a] \P_{\min(\ii, |\P| - \one_n)} \cdot \alpha_k^{\max(0, i_k - |\P|_k + 1)}$ by definition.
Furthermore, $\egd{\Q}{\bbeta} = \egd{\Q}{\balpha}$ follows from $\Q_{k:|\Q|_k - 1} = \zero$.

For $E = \lnot E_1$, let $\mu = \egd{\P}{\balpha}$ and $\egd{\R}{\bgamma} := \evgeo{\egd{\P}{\balpha}}{E_1} = \mu|_{E_1}$ by induction hypothesis.
By the inductive hypothesis, we have $\egd{\R}{\bgamma} = \egd{\R}{\balpha}$.
Let $\P'$ be the expansion of $\P$ to size $|\Q| = |\R|$.
We find:
\begin{align*}
\mu|_{\lnot E_1}(\x) &= \mu - \mu|_{E_1}(\x) \\
&= \egd{\P'}{\balpha}(\x) - \egd{\R}{\balpha}(\x) \\
&= \P'_{\min(\x, |\P'| - \one_n)} \balpha^{\max(\zero_n, \x - |\P'| + \one_n)} - \R_{\min(\x, |\R| - 1)} \cdot \balpha^{\max(\zero_n, \x - |\R| + \one_n)} \\
&= \left(\P'_{\min(\x, |\Q| - \one_n)} - \R_{\min(\x, |\Q| - \one_n)}\right) \cdot \balpha^{\max(\zero_n, \x - |\Q| + \one_n)} \\
&= \Q_{\min(\x, |\Q| - \one_n)} \balpha^{\max(\zero_n, \x - |\Q| + \one_n)} \\
&= \Q_{\min(\x, |\Q| - \one_n)} \bbeta^{\max(\zero_n, \x - |\Q| + \one_n)} \\
&= \egd{\Q}{\bbeta}(\x)
\end{align*}
since $\Q_\ii := \P'_\ii - \R_\ii = \P_{\min(\ii, |\P| - \one_n)} \cdot \balpha^{\max(\zero_n, \ii - |\P| + \one_n)} - \R_\ii$ by the definition of the expansion.

For $E = \Flip(\rho)$, the claim is clear by linearity.
For $E = E_1 \land E_2$, the correctness follows immediately from $\mu|_{E_1 \land E_2} = \mu|_{E_1}|_{E_2}$.

For the statement semantics, we prove $\egd{\P}{\balpha} \semgeo{P} \egd{\Q}{\bbeta}$ implies $\sem{P}(\egd{\P}{\bbeta}) \mle \egd{\Q}{\bbeta}$.
The cases $\Skip$, $P_1;P_2$, and $\Fail$ are obvious.
For $\Skip$ and $\Fail$, even equality holds.

The case $X_k := 0$ follows from \cref{lem:egd-marginalize}, except that the variable $X_k$ is not removed, but set to 0.
We have $\sem{P}(\mu)(\x) = \sum_{n \in \NN} \mu(\x[k \mapsto n]) \cdot [x_k = 0]$.
Hence we find
\begin{align*}
&\sem{X_k := 0}(\mu)(\x) \\
&= \sum_{n \in \NN} \mu(\x[k \mapsto n]) \cdot [x_k = 0] \\
&= \left(\frac{\P_{\min(\x[k \mapsto |\P|_k - 1], |\P| - \one_n)}}{1 - \alpha_k} + \sum_{j=0}^{|\P|_k - 2} \P_{\min(\x[k \mapsto j], |\P| - \one_n)}\right) \cdot [x_k = 0] \cdot \balpha^{\max(\zero_n, \x - |\P| + \one_n)} \\
&= [x_k = 0] \left(\frac{\P_{\min(\x[k \mapsto |\P|_k - 1], |\P| - \one_n)}}{1 - \alpha_k} + \sum_{j=0}^{|\P|_k - 2} \P_{\min(\x[k \mapsto j], |\P| - \one_n)}\right) \cdot \balpha[k\mapsto 0]^{\max(\zero_n, \x - |\P|[k \mapsto 2] + \one_n)} \\
&= \Q_{\min(\x, |\Q| - \one_n)} \cdot \bbeta^{\max(\zero_n, \x - |\Q| + \one_n)} \\
\end{align*}
by the definition of $\Q$.

The case $X_k \passign a$ is easier because $\sem{P}(\mu)(\x) = [x_k \ge a] \mu(\x[k \mapsto x_k - a])$.
We find:
\begin{align*}
\sem{X_k \passign a}(\egd{\P}{\balpha})(\x) &= [x_k \ge a] \egd{\P}{\balpha}(\x[k \mapsto x_k - a]) \\
&= [x_k \ge a] \P_{\min(\x[k \mapsto x_k - a], |\P| - \one_n)} \balpha^{\max(\zero_n, \x[k \mapsto x_k - a] - |\P| + \one_n)} \\
&= [x_k \ge a] \P_{\min(\x[k \mapsto x_k - a], |\P| - \one_n)} \balpha^{\max(\zero_n, \x - |\P|[k \mapsto \P_k + a] + \one_n)} \\
&= \Q_{\min(\x, |\Q| - \one_n)} \balpha^{\max(\zero_n, \x - |\Q| + \one_n)} \\
&= \egd{\Q}{\bbeta}(\x)
\end{align*}

The case $X_k \massign 1$ is similar: $\sem{P}(\mu)(\x) = \mu(\x[k \mapsto x_k + 1]) + [x_k = 0]\mu(\x)$.
First we consider the case $|\P|_k \ge 3$.
Then $|\Q|_k = \max(|\P|_k - 1, 2) = |\P|_k - 1$ and we find:
\begin{align*}
&\sem{X_k \massign 1}(\egd{\P}{\balpha})(\x) \\
&= [x_k = 0]\egd{\P}{\balpha}(\x) + \egd{\P}{\balpha}(\x[k \mapsto x_k + 1]) \\
&= [x_k = 0] \P_{\min(\x, |\P| - \one_n)} \balpha^{\max(\zero_n, \x - |\P| + \one_n)} + \P_{\min(\x[k \mapsto x_k + 1], |\P| - \one_n)} \balpha^{\max(\zero_n, \x[k \mapsto x_k + 1] - |\P| + \one_n)} \\
&= \left([x_k = 0] \P_{\min(\x[k \mapsto 0], |\P| - \one_n)} + \P_{\min(\x[k \mapsto x_k + 1], |\P| - \one_n)}\right) \balpha^{\max(\zero_n, \x - |\Q| + \one_n)} \\
&= \Q_{\min(\x, |\Q| - \one_n)} \balpha^{\max(\zero_n, \x - |\Q| + \one_n)} \\
&= \egd{\Q}{\bbeta}(\x)
\end{align*}
For $|\P|_k \le 2$, we first extend $\egd{\P}{\balpha}$ to $\egd{\P'}{\balpha}$ with $|\P'|_k = 3$.
By the above argument, we find $\sem{X_k \massign 1}(\egd{\P'}{\balpha}) = \egd{\Q}{\bbeta}$ with
\begin{align*}
\Q_{\ii} &= [i_k = 0] \P'_{\ii[k \mapsto 0]} + \P'_{\ii[k \mapsto i_k + 1]} \\
&= [i_k = 0] \P_{\ii[k \mapsto 0]} + \P_{\ii[k \mapsto i_k + 1]} \alpha_k^{\max(0, i_k + 1 - |\P|_k + 1)} \\
&= [i_k = 0] \P_{\ii[k \mapsto 0]} + \P_{\ii[k \mapsto i_k + 1]} \alpha_k^{\max(0, 2 + i_k - |\P|_k)}
\end{align*}
for all $\ii < |\Q|$.
It turns out that for $|\P|_k \ge 3$, we have  $i_k \le |\Q|_k - 1 \le |\P|_k - 2$ and thus $\max(0, 2 + i_k - |\P|_k) = 0$.
Hence the above equation for $\Q_{\ii}$ holds generally, as desired.
Note that we have to ensure $|\Q|_k \ge 2$ due to the special case for $i_k = 0$, which necessitates the case split on $|\P|_k \ge 3$.

For the case $\Ite{E}{P_1}{P_2}$, we have $\sem{P}(\mu)(\x) = \sem{P_1}(\mu|_E)(\x) + \sem{P_2}(\mu|_{\lnot E})(\x)$.
Suppose $\evgeo{\egd{\P}{\balpha}}{E} \semgeo{P_1} \egd{\R}{\bgamma}$ and $\evgeo{\egd{\P}{\balpha}}{\lnot E} \semgeo{P_2} \egd{\S}{\bdelta}$ and $\egd{\Q}{\bbeta}$ is a join of $\egd{\R}{\bgamma}$ and $\egd{\S}{\bdelta}$.
By the definition of the join relation, we have $\egd{\R}{\bgamma} + \egd{\S}{\bdelta} \mle \egd{\Q}{\bbeta}$.
Then we find with the inductive hypothesis:
\begin{align*}
\sem{\Ite{E}{P_1}{P_2}}(\egd{\P}{\balpha})
&= \sem{P_1}\left(\evgeo{\egd{\P}{\balpha}}{E}\right) + \sem{P_2}\left(\evgeo{\egd{\P}{\balpha}}{\lnot E}\right) \\
&\mle \egd{\R}{\bgamma} + \egd{\S}{\bdelta} \\
&\mle \egd{\Q}{\bbeta}
\end{align*}

Finally, for loops $\Whilst{E}{B}$, suppose that there is a $c$-contraction invariant $\egd{\R}{\bgamma}$ such that $\egd{\P}{\balpha} \egdle \egd{\R}{\bgamma}$ and $\evgeo{\egd{\R}{\bgamma}}{E} \semgeo{B} \egd{\S}{\bdelta}$ with $\egd{\S}{\bdelta} \egdle \egd{c \cdot \R}{\bgamma}$.
By the induction hypothesis, we know $\sem{B}(\egd{\R}{\bgamma}|_E) \mle \egd{\S}{\bdelta}$.
The fixed point equation of $\Whilst{E}{B}$ yields by monotonicity of the standard semantics:
\begin{align*}
\sem{\Whilst{E}{B}}\left(\egd{\R}{\bgamma}\right) &= \egd{\R}{\bgamma}|_E + \sem{\Whilst{E}{B}}\left(\sem{B}\left(\egd{\R}{\bgamma}|_E\right)\right) \\
&\mle \evgeo{\egd{\R}{\bgamma}}{\lnot E} + \sem{\Whilst{E}{B}}\left(\egd{\S}{\bdelta}\right) \\
&\mle \evgeo{\egd{\R}{\bgamma}}{\lnot E} + \sem{\Whilst{E}{B}}\left(c \cdot \egd{\R}{\bgamma}\right) \\
&= \evgeo{\egd{\R}{\bgamma}}{\lnot E} + c \cdot \sem{\Whilst{E}{B}}\left(\egd{\R}{\bgamma}\right)
\end{align*}
which implies that $\sem{\Whilst{E}{B}}\left(\egd{\R}{\bgamma}\right) \mle \frac{1}{1-c}\evgeo{\egd{\R}{\bgamma}}{\lnot E}$.
By monotonicity, we also get $\sem{\Whilst{E}{B}}\left(\egd{\P}{\balpha}\right) \mle \sem{\Whilst{E}{B}}\left(\egd{\R}{\bgamma}\right) \mle \frac{1}{1-c}\evgeo{\egd{\R}{\bgamma}}{\lnot E} = \egd{\Q}{\bbeta}$.

For bounds on the normalized distribution, note that the normalizing constant equals
\begin{align*}
  1 - \esem{P}(\mu)(\lightning) &= \esem{P}(\mu)(\estates) - \esem{P}(\mu)(\lightning) = \esem{P}(\mu)(\NN^n) = \sem{P}(\mu)(\NN^n)
\end{align*}
by almost sure termination, the fact that $\mu$ is a probability distribution, and the definition of $\sem{P}$.
The lower bound on $\sem{P}(\mu)(\NN^n)$ follows from \cref{thm:soundness-lower-bounds} and the upper bound $\egd{\Q}{\bbeta}(\NN^n)$ from the above proof.
As a consequence, we obtain the desired bounds on the normalized distribution.
\end{proof}

\PrecisionLoopFree*
\begin{proof}
The proof is by induction on the structure of $P$.
We strengthen the inductive hypothesis to state that $\egd{\Q}{\bbeta} = \egd{\Q}{\balpha}$ and each $\beta_i \in \{0, \alpha_i\}$.
In the soundness proof, it was shown that the restriction $\evgeo{\egd{\P}{\balpha}}{E}$ is correct and that the semantics of $\Skip$, $X_k := 0$, $X_k \passign c$, $X_k \massign 1$, and $\Fail$ are exact.
Only for $\Fail$ and $X_k := 0$ can it happen that $\bbeta \ne \balpha$ (because some $\beta_k$'s are set to zero).
In both cases, $\egd{\Q}{\balpha} = \egd{\Q}{\bbeta}$ because $\Q_{k:|\Q|_k-1}$ for all $k$ where $\beta_k = 0 < \alpha_k$.

For $P = P_1; P_2$, the inductive hypothesis ensures that $\semgeo{P_1}$ and $\semgeo{P_2}$ are precise and thus $\semgeo{P}$ is precise as well.
The stronger claim $\beta_i \in \{0, \alpha_i\}$ also clearly holds for them because they can only ever set a component of $\balpha$ to zero.

The only construct that is generally not precise is the conditional $\Ite{E}{P_1}{P_2}$.
Suppose that $\evgeo{\egd{\P}{\balpha}}{E} \semgeo{P_1} \egd{\R}{\bgamma}$ and $\evgeo{\egd{\P}{\balpha}}{\lnot E} \semgeo{P_2} \egd{\S}{\bdelta}$, which are precise by inductive hypothesis and $\gamma_i, \delta_i \in \{0, \alpha_i\}$ for all $i$.

We claim that $\left(\egd{\R}{\bgamma}, \egd{\S}{\bdelta}, \egd{\Q}{\bbeta}\right) \in \JoinRel^*$ implies $\egd{\R}{\bgamma} + \egd{\S}{\bdelta} = \egd{\Q}{\bbeta}$.
By the definition of the strict join relation, $\bbeta = \max(\bgamma, \bdelta) \le \balpha$.
Thus each $\gamma_i, \delta_i \in \{0, \beta_i\} \subseteq \{0, \alpha_i\}$.
If either of them is zero, then our strengthened induction hypothesis ensures that $\egd{\R}{\bgamma} = \egd{\R}{\bbeta}$ and $\egd{\S}{\bdelta} = \egd{\S}{\bbeta}$ and similarly for any expansions $\R', \S'$ of them.
Then by the definition of the strict join relation, $\Q = \R' + \S'$ for some expansions $\R', \S'$.
As a consequence:
\begin{align*}
  \egd{\Q}{\bbeta} &= \egd{\R'}{\bbeta} + \egd{\S'}{\bbeta} \\
  &= \egd{\R'}{\bgamma} + \egd{\S'}{\bdelta} \\
  &= \egd{\R}{\bgamma} + \egd{\S}{\bdelta} \\
  &= \sem{P_1}\left(\evgeo{\egd{\P}{\balpha}}{E}\right) + \sem{P_2}\left(\evgeo{\egd{\P}{\balpha}}{\lnot E}\right) \\
  &= \sem{\Ite{E}{P_1}{P_2}}(\egd{\P}{\balpha})
\end{align*}
Furthermore $\egd{\R}{\bgamma} = \egd{\R}{\balpha}$ and $\egd{\S}{\bdelta} = \egd{\S}{\balpha}$ by the inductive hypothesis.
As a consequence, $\egd{\Q}{\balpha} = \egd{\R'}{\balpha} + \egd{\S'}{\balpha} = \egd{\R}{\bgamma} + \egd{\S}{\bdelta} = \egd{\Q}{\bbeta}$.
\end{proof}

\SufficientConditions*
\begin{proof}
Since the loop body is loop-free, we have ${\semgeo{B}} = \sem{B}$ \cref{thm:precision-loop-free} (if the strict join relation is used in the semantics of $\Ite{-}{-}{-}$).
Conditions (1) and (2) ensure that the distribution of the change is the same in every loop iteration.
(The technical condition (2) is needed to deal with the fact that decrements are clamped at zero.)
We will show that a contraction invariant of the form $\egd{\R}{\bgamma}$ exists, with $\R \in \RR$ viewed as a $(1, \dots, 1)$-tensor.

We will see that $\egd{\R}{\bgamma}$ with $\bgamma = \one_n$ is ``almost'' a solution to the constraints of a $c$-contraction invariant.
``Almost'' in the sense that we relax some $<$-constraints to $\le$-constraints, for example $\gamma_i < 1$ and $c < 1$.
Via a gradient argument, we can show that taking an $\epsilon$-small step from $\bgamma = \one_n$ in the direction of $-\blambda = (-\lambda_1, \dots, -\lambda_n)$, i.e. the negated coefficients from the linear combination, improves all of the $\le$-constraints.
The full proof works as follows.

Let $b_i$ be the maximum number of increments of $X_i$ in any program path through $B$.
We claim that there is a finite subprobability distribution $D$ on $\Delta := \bigtimes_{i=1}^n \{-a_i, \dots, b_i\}$ such that $\sem{B}(\Dirac(\x)) = \sum_{\bdelta \in \Delta} D(\bdelta) \cdot \Dirac(\x + \bdelta)$ for all $\x$ with $x_i \ge a_i$.
The latter restriction is necessary because for $X_i = 0$, the decrement $X_i \massign 1$ has no effect.
Since we assume that $E$ guarantees that $X_i \ge a_i$, we find $\sem{B}(\Dirac(\x)|_E)(\w) \le [\w - \x \in \Delta] \cdot D(\w - \x)$ for all $\w \in \NN^n$.
By a weighted sum of all $\x$, we can construct any measure $\mu$ on $\NN^n$ and obtain
\begin{align*}
  \sem{B}(\mu|_E)(\w) &\le \sum_{\x \in \NN^n} \mu|_E(\x) \cdot [\w - \x \in \Delta] \cdot D(\w - \x) \\
  &= \sum_{\bdelta \in \NN^n} \mu|_E(\w - \bdelta) \cdot [\bdelta \in \Delta] \cdot D(\bdelta) \\
  &= \sum_{\bdelta \in \Delta} D(\bdelta) \cdot \mu|_E(\w - \bdelta)
\end{align*}

Let's see what one loop iteration on $\mu = \egd{\R}{\bgamma}$ does, where $\R \in \RR$ is viewed as a $(1, \dots, 1)$-tensor.
The restriction $\mu|_E$ is represented by $\egd{\R'}{\bgamma}$ with $\R'_{\x} \le [x_1 \ge a_1] \cdots [x_n \ge a_n] \cdot \R \bgamma^{\x}$ since $E$ implies $\x \ge {\bm a}$.
The semantics of $B$ on $\mu|_E$ is then $\egd{\S}{\bgamma'} := {\semgeo{B}}\left(\evgeo{\egd{\R}{\bgamma}}{E}\right)$.
By \cref{thm:precision-loop-free}, $\egd{\S}{\bgamma'}$ is precise and we find for $\x < |\S|$:
\begin{align*}
  \S_{\x} &:= \sum_{\bdelta \in \Delta} D(\bdelta) \cdot \egd{\R'}{\bgamma}(\x - \bdelta) \\
  &\le \sum_{\bdelta \in \Delta} D(\bdelta) \cdot [\x - \bdelta \ge {\bm a}] \cdot \R \cdot \bgamma^{\x - \bdelta} \\
  &\le \R \sum_{\bdelta \in \Delta} D(\bdelta) \cdot \bgamma^{\x - \bdelta}
\end{align*}
By the invariant of the proof of \cref{thm:precision-loop-free}, we have $\bgamma' \le \bgamma$ and $\egd{\S}{\bgamma'} = \egd{\S}{\bgamma}$.
The constraint $\egd{\S}{\bgamma'} \egdle c \cdot \egd{\R}{\bgamma}$ is thus implied by the following inequality:
\begin{align*}
  \R \sum_{\bdelta \in \Delta} D(\bdelta) \cdot \bgamma^{\x - \bdelta} &\le c \cdot \R \cdot \bgamma^{\x} \\
\iff L(\bgamma) := \sum_{\bdelta \in \Delta} D(\bdelta) \cdot \bgamma^{-\bdelta} &\le c < 1
\end{align*}

The left-hand side is $\le 1$ at the point $\bgamma = \one_n$ because $D$ is a subprobability distribution.
We will show that in any neighborhood of $\one_n$, there is a $\bgamma < \one_n$ such that the LHS is $L(\bgamma) < 1$.
Then choosing $c := L(\bgamma)$ ensures $\egd{\S}{\bgamma'} \egdle c \cdot \egd{\R}{\bgamma}$.

To show that such a point $\bgamma < \one_n$ exists, we exhibit a direction in which $L(\bgamma)$ decreases, starting at $\bgamma = \one_n$.
For this it is enough to find $\blambda \in \RR_{\ge 0}^n$ such that $\blambda^\top \nabla L(\one_n) > 0$.
We compute the derivatives:
\begin{align*}
  \frac{\partial}{\partial \gamma_k} L(\bgamma)\big|_{\bgamma = \one_n} &= -\sum_{\bdelta \in \Delta} D(\bdelta) \cdot \delta_k \cdot \gamma_1^{-\delta_1} \cdots \gamma_k^{-\delta_k - 1} \cdots \gamma_n^{-\delta_n} \bigg|_{\bgamma = \one_n} \\
  &= -\sum_{\bdelta \in \Delta} D(\bdelta) \cdot \delta_k \\
  &= -\sum_{\X \in \NN^n} [\X - \x \in \Delta] \cdot D(\X - \x) \cdot (X_k - x_k) \\
  &= -\sum_{\X \in \NN^n} \sem{B}(\Dirac(\x)|_E)(\X) \cdot (X_k - x_k) \\
  &= -\ExpVal_{\X \sim \sem{B}(\Dirac(\x)|_E)}[X_k - x_k]
\end{align*}
where we picked $\x \ge {\bm a}$.
In other words, this is the negative expected change of $X_k$ after running one loop iteration on $\x \ge {\bm a}$.
Therefore, we have
\begin{align*}
  \blambda^\top \nabla L(\one_n) &\ge -\sum_{k=1}^n \lambda_k \ExpVal_{\X \sim \sem{B}(\Dirac(\x)|_E)}[X_k - x_k] \\
  &= -\ExpVal_{\X \sim \sem{B}(\Dirac(\x)|_E)}\left[\sum_{k=1}^n \lambda_k (X_k - x_k)\right] \\
  &> 0
\end{align*}
Without loss of generality, we can assume that $\lambda_k > 0$ for all $k$ because increasing each $\lambda_k$ by a small enough amount will not invalidate the above inequality since $X_k - x_k$ is bounded.
This means that there is an $\epsilon > 0$ such that for all $t \in (0, \epsilon)$, we have $L(\one_n - t \blambda) < 1$.
Then all the constraints $L(\bgamma) < 1$ are satisfied at $\bgamma = 1 - t \blambda < \one_n$ for all $t \in (0, \epsilon)$.
Having found $\bgamma$, we can choose $c := L(\bgamma)$ and all the constraints arising from $\egd{\S}{\bgamma'} \egdle \egd{c \cdot \R}{\bgamma}$ are satisfied.

The only remaining constraint is $\egd{\P}{\balpha} \egdle \egd{\R}{\bgamma}$.
Concretely, this gives rise to the constraints $\balpha \le \bgamma$ and to the constraints of the form (for all $\x \le |\P|$):
\[ \P_{\x} \le \R \cdot \gamma_1^{x_1} \cdots \gamma_n^{x_n} \]
Since $\bgamma$ can be chosen to be strictly positive, this can be satisfied by choosing $\R \ge \max_{\x \le |\P|} \frac{\P_{\x}}{\gamma_1^{x_1} \cdots \gamma_n^{x_n}}$.
The constraints $\balpha \le \bgamma$ can clearly be satisfied by choosing $\bgamma$ close enough to $\one_n$.

From the existence of a contraction invariant $\egd{\R}{\bgamma}$, it directly follows that $\egd{\P}{\balpha} \semgeo{L} \egd{\Q}{\bbeta} := \evgeo{\egd{\frac{\R}{1-c}}{\bgamma}}{\lnot E}$.
\end{proof}

\NecessaryRuntime*
\begin{proof}
We introduce an extra program variable $T := X_{n+1}$ to keep track of the running time.
We transform the original program $P$ to a program $P^T$ which adds the statement $T \passign 1$ after each statement.

We claim that whenever $\egd{\P}{\balpha} \semgeo{P} \egd{\Q}{\bbeta}$, then there is a $\Q^T$ and $\beta^* \in (0,1)$ such that for all $\beta' \in [\beta^*,1)$, there is a $C(\beta') \in \nnr$ such that $C(\beta')$ tends to 1 as $\beta' \to 1$, and we have
\[ \egd{[\P]}{(\balpha, \beta')} \semgeo{P^T} \egd{\Q^T}{(\bbeta, \beta')} \egdle \egd{C(\beta') \cdot [\Q]}{(\bbeta, \beta')} \]
where the notation $[\P]$ denotes the $(|\P|,1)$-tensor with entry $[\P]_{n+1: 0} = \P$.
This claim implies the statement of the theorem because then $\egd{[\P]}{(\balpha, 0)} \egdle \egd{[\P]}{(\balpha, \beta')} \semgeo{P^T} \egd{\Q^T}{(\bbeta, \beta')}$ for all $\beta' \in [\beta^*, 1)$.
By \cref{lem:geo-sem-well-behaved}, there is also an $\egd{\Q'}{(\bbeta, \beta')}$ such that $\egd{[\P]}{(\balpha, 0)} \semgeo{P^T} \egd{\Q'}{(\bbeta, \beta')}$.

We observe that increasing $T = X_{n+1}$ transforms $\egd{[\P]}{(\balpha,\alpha')}$ to
\[ \egd{[\zero, \P]}{(\balpha, \alpha')} \egdle \egd{\frac{1}{\alpha'}[\P]}{(\balpha, \alpha')} \]
where $[\zero, \P]$ is the $(|\P|, 2)$-tensor with $[\zero, \P]_0 = \zero$ and $[\zero, \P]_1 = \P$.
Using this observation, we prove the above claim by induction on the structure of $P$.

For simple statements $P$ (i.e. not containing other statements), the claim is proven as follows.
Since $P^T = (P; T \passign 1)$, we have
\[ \egd{[\P]}{(\balpha, \beta')} \semgeo{P^T} \egd{[\zero, \Q]}{(\bbeta, \beta')} \egdle \egd{\frac{1}{\beta'} [\Q]}{(\bbeta, \beta')} \]
and by setting $C(\beta') = \frac{1}{\beta'}$, the claim is proven.

For statements $P = P_1; P_2$, there is an $\egd{\R}{\bgamma}$ such that $\egd{\P}{\balpha} \semgeo{P_1} \egd{\R}{\bgamma}$ and $\egd{\R}{\bgamma} \semgeo{P_2} \egd{\Q}{\bbeta}$.
By inductive hypothesis, there is a $\R^T$ and $\gamma^* \in (0, 1)$ such that $\egd{[\P]}{(\balpha, \gamma')} \semgeo{P_1^T} \egd{\R^T}{(\bgamma, \gamma')}$ and $\egd{\R^T}{(\bgamma, \gamma')} \egdle \egd{C_1(\gamma') \cdot [\R]}{(\bgamma, \gamma')}$ for all $\gamma' \in [\gamma^*, 1)$.
By inductive hypothesis, there is also a $\Q^T$ and $\beta^* \in [\gamma', 1)$ such that $\egd{[\R]}{(\bgamma, \beta')} \semgeo{P_2^T} \egd{\Q^T}{(\bbeta, \beta')}$ and $\egd{\Q^T}{(\bbeta, \beta')} \egdle \egd{C_2(\beta') [\Q]}{(\bbeta, \beta')}$ for all $\beta' \in [\beta^*, 1)$.
Then the claim follows by \cref{lem:geo-sem-well-behaved} with $C(\beta') := C_1(\beta') \cdot C_2(\beta')$ by monotonicity of $C$.

For conditionals $P = \Ite{E}{P_1}{P_2}$, we have that $\evgeo{\egd{\P}{\balpha}}{E} \semgeo{P_1} \egd{\R}{\bgamma}$ and $\evgeo{\egd{\P}{\balpha}}{\lnot E} \semgeo{P_2} \egd{\S}{\bdelta}$.
By inductive hypothesis, there are $\R^T, \S^T$ and $\gamma^* \in (0, 1), \delta^* \in (0, 1)$ such that 
\[ \evgeo{\egd{[\P]}{(\balpha, \gamma')}}{E} \semgeo{P_1^T} \egd{\R^T}{(\bgamma, \gamma')} \egdle \egd{C_1(\gamma') \cdot [\R]}{(\bgamma, \gamma')} \]
for all $\gamma' \in [\gamma^*, 1)$ and
\[ \evgeo{\egd{[\P]}{(\balpha, \delta')}}{E} \semgeo{P_2^T} \egd{\S^T}{(\bdelta, \delta')} \egdle \egd{C_2(\delta') \cdot [\S]}{(\bdelta, \delta')} \]
for all $\delta' \in [\delta^*, 1)$.
Set $\beta^* = \max(\gamma^*, \delta^*)$ and $C(\beta') = \max(C_1(\beta'), C_2(\beta'))$.
Then $\egd{[\Q]}{(\bbeta, \beta')}$ is a join of $\egd{[\R]}{(\bgamma, \gamma')}$ and $\egd{[\S]}{(\bdelta, \delta')}$ and the claim follows from \cref{lem:join-well-behaved,lem:geo-sem-well-behaved}.

For loops $P = \Whilst{E}{B}$, we know that there is a $c$-contraction invariant $\egd{\R}{\bgamma} \egdge \egd{\P}{\balpha}$ with $\egd{\Q}{\bbeta} = \evgeo{\egd{\frac{\R}{1-c}}{\bgamma}}{\lnot E}$.
We claim that $\egd{[\R]}{(\bgamma, \gamma')}$ is a contraction invariant for $P^T$ for all $\gamma' \in [\gamma^*, 1)$ where $\gamma^*$ will be determined later.
By the inductive hypothesis, we have
\[ \evgeo{\egd{[\R]}{(\bgamma, \gamma')}}{E} \semgeo{B^T} \egd{[\S']}{(\bdelta, \delta')} \egdle \egd{c \cdot C(\gamma') \cdot [\R]}{(\bgamma, \gamma')} \]
Since $\lim_{\gamma' \to 1} C(\gamma') = 1$, we can choose $\gamma^*$ large enough such that $C(\gamma') < \frac{1}{c}$ for all $\gamma' \in [\gamma^*, 1)$ and thus $c'(\gamma') := c \cdot C(\gamma') < 1$ for all $\gamma' \in [\gamma^*,1)$.
Furthermore, we have $c'(\gamma') \to c$ as $\gamma' \to 1$.
Clearly, $\egd{[\P]}{(\balpha, \gamma')} \egdle \egd{[\R]}{(\bgamma, \gamma')}$ holds too.
Thus $\egd{[\R]}{(\bgamma, \gamma')}$ is a $c'(\gamma')$-contraction invariant for all $\gamma' \in [\gamma^*, 1)$.
So we obtain the bound:
\[ \egd{[\P]}{(\balpha, \gamma')} \semgeo{P} \evgeo{\egd{\left[\frac{\R}{1 - c'(\gamma')}\right]}{(\bgamma, \gamma')}}{\lnot E} \]
The ratio between $\frac{1}{1 - c'(\gamma')}$ and $\frac{1}{1 - c}$ tends to 1 as $\gamma' \to 1$, as well.
This finishes the proof.
\end{proof}

\ConvergenceUpperBounds*
\begin{proof}
Recall that there are two sources of imprecision in the geometric bound semantics: overapproximating loops and joining branches of conditionals.
The imprecision of loop approximations decreases exponentially in $u$ because by the definition of a contraction invariant, the bound on the probability distribution entering the loop decreases by the contraction factor $c < 1$ in each iteration.
So after $u$ unrollings, the imprecision is reduced by a factor of $c^u$.
The imprecision of joins can also be reduced by picking larger and larger expansions of the EGDs involved.
The complete proof that follows is technical, because we have to ensure that all these approximations compose in the right way, but these two ideas are the key ingredients.

First, note that $\sem{P^{(u)}} = \sem{P}$ since unrolling does not affect the semantics of a program.
The proof works by induction on the structure of $P$.
We strengthen the inductive hypothesis by asserting that $A$ and $C$ are monotonic functions of $\egd{\P}{\balpha}$.

\textbf{Exactly representable statements.}
The statements $P$ where $\egd{\P}{\balpha} \semgeo{P} \sem{P}(\egd{\P}{\balpha})$ are $\Skip$, $X_k \passign c$, $X_k \massign 1, X_k := 0$.
Since $\semgeo{P}$ is exact here, the claim holds trivially by setting $A = 0$ and $C = 0$.

\textbf{Sequential composition.}
If $P = P_1; P_2$, we argue as follows:
Since $\egd{\P}{\balpha} \semgeo{P} \egd{\Q}{\bbeta}$, there exist $\R, \bgamma$ such that $\egd{\P}{\balpha} \semgeo{P_1} \egd{\R}{\bgamma}$ and $\egd{\R}{\bgamma} \semgeo{P_2} \egd{\Q}{\bbeta}$.
By the induction hypothesis, there exist $A_1, C_1, \R^{(u)}$ such that $\egd{\P}{\balpha} \semgeo{P_1^{(u)}} \egd{\R^{(u)}}{\bgamma}$ and $\egd{\R^{(u)}}{\bgamma} - \sem{P_1}(\egd{\P}{\balpha}) \mle A_1 \cdot C_1^u \cdot \egd{\R}{\bgamma}$.
Also by the induction hypothesis, there exist $A_2, C_2, \Q^{(u)}, \Q'$ such that $\egd{\R^{(u)}}{\bgamma} \semgeo{P_2} \egd{\Q'}{\bbeta} \egdle \egd{\Q}{\bbeta}$ (by \cref{lem:geo-sem-well-behaved} since $\egd{\R^{(u)}}{\bgamma} \egdle \egd{\R}{\bgamma}$) and $\egd{\R^{(u)}}{\bgamma} \semgeo{P_2^{(u)}} \egd{\Q^{(u)}}{\bbeta}$ and $\egd{\Q^{(u)}}{\bbeta} - \sem{P_2}\left(\egd{\R^{(u)}}{\bgamma}\right) \mle A_2 \cdot C_2^u \cdot \egd{\Q'}{\bbeta}$.
Together, we find:
\begin{align*}
&\egd{\Q^{(u)}}{\bbeta} - \sem{P}(\egd{\P}{\balpha}) \\
&= \left(\egd{\Q^{(u)}}{\bbeta} - \sem{P_2}\left(\egd{\R^{(u)}}{\bgamma}\right)\right) + \left(\sem{P_2}\left(\egd{\R^{(u)}}{\bgamma}\right) - \sem{P}\left(\egd{\P}{\balpha}\right)\right) \\
&= \left(\egd{\Q^{(u)}}{\bbeta} - \sem{P_2}\left(\egd{\R^{(u)}}{\bgamma}\right)\right) + \sem{P_2}\left(\egd{\R^{(u)}}{\bgamma} - \sem{P_1}\left(\egd{\P}{\balpha}\right)\right) \\
&\mle A_2 \cdot C_2^u \cdot \egd{\Q'}{\bbeta} + \sem{P_2}\left(A_1 \cdot C_1^u \cdot \egd{\R}{\bgamma}\right) \\
&\mle A_2 \cdot C_2^u \cdot \egd{\Q}{\bbeta} + A_1 \cdot C_1^u \cdot \egd{\Q}{\bbeta} \\
&\mle (A_1 + A_2) \cdot \max(C_1, C_2)^u \cdot \egd{\Q}{\bbeta}
\end{align*}
Furthermore, monotonicity of $A := A_1 + A_2$ and $C := \max(C_1, C_2)$ in $\egd{\P}{\balpha}$ is preserved.

\textbf{Conditionals.}
The case $\Ite{E}{P_1}{P_2}$ is similar.
First, let $\evgeo{\egd{\P}{\balpha}}{E} \semgeo{P_1} \egd{\R}{\bgamma}$ and $\evgeo{\egd{\P}{\balpha}}{\lnot E} \semgeo{P_2} \egd{\S}{\bdelta}$.
Similarly, let $\evgeo{\egd{\P}{\balpha}}{E} \semgeo{P_1^{(u)}} \egd{\R^{(u)}}{\bgamma}$ and $\evgeo{\egd{\P}{\balpha}}{\lnot E} \semgeo{P_2^{(u)}} \egd{\S^{(u)}}{\bdelta}$ with
\begin{align*}
\egd{\R^{(u)}}{\bgamma} - \sem{P_1}\left(\evgeo{\egd{\P}{\balpha}}{E}\right) &\mle A_1 \cdot C_1^u \cdot \egd{\R}{\bgamma} \\
\egd{\S^{(u)}}{\bdelta} - \sem{P_2}\left(\evgeo{\egd{\P}{\balpha}}{\lnot E}\right) &\mle A_2 \cdot C_2^u \cdot \egd{\S}{\bdelta}
\end{align*}
Now $\egd{\Q}{\bbeta}$ is a join of $\egd{\R}{\bgamma}$ and $\egd{\S}{\bdelta}$ via expansions $\R_*, \S_*$ and $\egd{\Q^{(u)}}{\bbeta}$ is a join of $\egd{\R^{(u)}}{\bgamma}$ and $\egd{\S^{(u)}}{\bdelta}$ via expansions $\R_*^{(u)}$ and $\S_*^{(u)}$.
We know that $\egd{\Q}{\bbeta}$ is a join of $\egd{\R}{\bgamma}$ and $\egd{\S}{\bdelta}$.
By \cref{lem:join-well-behaved}, there is a join $\egd{\Q'}{\bbeta}$ of $\egd{\R^{(u)}}{\bgamma}$ and $\egd{\S^{(u)}}{\bdelta}$ with $\egd{\Q'}{\bbeta} \egdle \egd{\Q}{\bbeta}$.
By \cref{lem:join-approx}, we can find such a $\Q^{(u)}$ with
\[ \egd{\Q^{(u)}}{\bbeta} - \egd{\R^{(u)}}{\bgamma} - \egd{\S^{(u)}}{\bdelta} \mle A_0 \cdot C_0^u \cdot \egd{\Q'}{\bbeta} \mle A_0 \cdot C_0^u \cdot \egd{\Q}{\bbeta} \]
where $A_0$ and $C_0$ will be chosen very shortly.
Overall, we find
\begin{align*}
&\egd{\Q^{(u)}}{\bbeta} - \sem{P}(\egd{\P}{\balpha}) \\
&= \left(\egd{\Q^{(u)}}{\bbeta} - \egd{\R^{(u)}}{\bgamma} - \egd{\S^{(u)}}{\bdelta}\right) \\
&\qquad + \left(\egd{\R^{(u)}}{\bgamma} - \sem{P_1}\left(\evgeo{\egd{\P}{\balpha}}{E}\right)\right) + \left(\egd{\S^{(u)}}{\bdelta} - \sem{P_2}\left(\evgeo{\egd{\P}{\balpha}}{\lnot E}\right)\right) \\
&\mle A_0 \cdot C_0^u \cdot \egd{\Q}{\bbeta} + A_1 \cdot C_1^u \cdot \egd{\R}{\bgamma} + A_2 \cdot C_2^u \cdot \egd{\S}{\bdelta} \\
&\mle A_0 \cdot C_0^u \cdot \egd{\Q}{\bbeta} + \max(A_1, A_2) \cdot \max(C_1, C_2)^u \cdot \egd{\Q}{\bbeta} \\
&\mle 2(A_1 + A_2) \cdot \max(C_1, C_2)^u \cdot \egd{\Q}{\bbeta}
\end{align*}
by choosing $A_0 := \max(A_1, A_2)$ and $C_0 := \max(C_1, C_2)$.
Monotonicity of $A := 2(A_1 + A_2)$ and $C := \max(C_1, C_2)$ in $\egd{\P}{\balpha}$ also holds by induction hypothesis.

\textbf{Loops.}
The most interesting case is $P = \Whilst{E}{B}$.
Since $\egd{\P}{\balpha} \semgeo{P} \egd{\Q}{\bbeta}$, there exist $c, \R, \S, \bgamma, \bdelta$ such that $\egd{\P}{\balpha} \egdle \egd{\R}{\bgamma}$, and $\evgeo{\egd{\R}{\bgamma}}{E} \semgeo{B} \egd{\S}{\bdelta}$ and $\egd{\S}{\bdelta} \egdle \egd{c \cdot \R}{\bgamma}$ and $\evgeo{\egd{\R}{\bgamma}}{\lnot E} = \egd{\Q}{\bbeta}$.
By the induction hypothesis, $\evgeo{\egd{\R}{\bgamma}}{E} \semgeo{B^{(u)}} \egd{\S^{(u)}}{\bgamma}$ with $\egd{\S^{(u)}}{\bgamma} \egdle \egd{\S}{\bdelta} \egdle c \cdot \egd{\R}{\bgamma}$, so $\egd{\R}{\bgamma}$ is a $c$-contraction invariant for $\Whilst{E}{B^{(u)}}$ too.
Furthermore, the inductive hypothesis yields
\[ \egd{\S^{(u)}}{\bgamma} - \sem{B}\left(\evgeo{\egd{\R}{\bgamma}}{E}\right) \mle A_r \cdot C_r^u \cdot \egd{\S}{\bdelta} \mle A_r \cdot C_r^u \cdot c \cdot \egd{\R}{\bgamma} \]

\begin{figure}
\begin{subfigure}{0.5\textwidth}
\begin{align*}
&\Comment{\egd{\P_0}{\balpha_0} = \egd{\P}{\balpha}} \\
&\Ite{E}{ \\
&\quad B \\
&\quad \Comment{\egd{\P_1}{\balpha_1}} \\
&\quad \Ite{E}{ \\
&\qquad B \\
&\qquad \Comment{\egd{\P_2}{\balpha_2}} \\
&\qquad \dots \\
&\qquad \Comment{\egd{\P_{u}}{\balpha_{u}}} \\
&\qquad \Whilst{E}{B} \\
&\qquad \Comment{\egd{\Q_{u}}{\bbeta_{u}}} \\
&\qquad \dots \\
&\qquad \Comment{\egd{\Q_2}{\bbeta_2}} \\
&\quad }{\Skip} \\
&\quad \Comment{\egd{\Q_1}{\bbeta_1}} \\
&}{\Skip} \\
&\Comment{\egd{\Q_0}{\bbeta_0} = \egd{\Q}{\bbeta} }
\end{align*}
\end{subfigure}%
\begin{subfigure}{0.5\textwidth}
\begin{align*}
&\Comment{\egd{\P_0^{(u)}}{\balpha_0} = \egd{\P}{\balpha}} \\
&\Ite{E}{ \\
&\quad B^{(u)} \\
&\quad \Comment{\egd{\P^{(u)}_1}{\balpha_1}} \\
&\quad \Ite{E}{ \\
&\qquad B^{(u)} \\
&\qquad \Comment{\egd{\P^{(u)}_2}{\balpha_2}} \\
&\qquad \dots \\
&\qquad \Comment{\egd{\P^{(u)}_{u}}{\balpha_{u}}} \\
&\qquad \Whilst{E}{B^{(u)}} \\
&\qquad \Comment{\egd{\Q^{(u)}_{u}}{\bbeta_{u}}} \\
&\qquad \dots \\
&\qquad \Comment{\egd{\Q^{(u)}_2}{\bbeta_2}} \\
&\quad }{\Skip} \\
&\quad \Comment{\egd{\Q^{(u)}_1}{\bbeta_1}} \\
&}{\Skip} \\
&\Comment{\egd{\Q^{(u)}_0}{\bbeta_0} = \egd{\Q^{(u)}}{\bbeta} }
\end{align*}
\end{subfigure}
\caption{Notation for the loop case in the proof of \cref{thm:convergence}}
\label{fig:convergence-loop-case-notation}
\end{figure}

\textbf{Definition of $\P_v$ and $\P_v^{(u)}$.}
Next, we will define intermediate distributions within the unrolling, illustrated in \cref{fig:convergence-loop-case-notation}.
Define $\egd{\P_0}{\balpha_0} = \egd{\P}{\balpha}$ and $\evgeo{\egd{\P_{v-1}}{\balpha_{v-1}}}{E} \semgeo{B} \egd{\P_v}{\balpha_v}$.
This satisfies the inequality $\egd{\P_v}{\balpha_v} \egdle c^v \egd{\R}{\bgamma}$ by a simple inductive argument using the definition of a $c$-contraction invariant.
Define $\egd{\P_0^{(u)}}{\balpha_0} = \egd{\P}{\balpha}$ and choose $\egd{\P_v^{(u)}}{\balpha_v}$ to satisfy the following conditions:
\begin{enumerate}
  \item $\evgeo{\egd{\P_{v-1}^{(u)}}{\balpha_{v-1}}}{E} \semgeo{B^{(u)}} \egd{\P_v^{(u)}}{\balpha_v}$, which is possible by induction hypothesis.
  \item $\egd{\P_v^{(u)}}{\balpha_v} \egdle \egd{\P_v}{\balpha_v}$, which is satisfiable by a simple inductive argument using monotonicity of the semantics.
  \item $\evgeo{\egd{\P_{v-1}^{(u)}}{\balpha_{v-1}}}{E} \semgeo{B} G_v^{(u)}$ with $G_v^{(u)} \egdle \egd{\P_v}{\balpha_v}$, which can be satisfied by \cref{lem:geo-sem-well-behaved} because we have $\evgeo{\egd{\P_{v-1}^{(u)}}{\balpha_{v-1}}}{E} \egdle \evgeo{\egd{\P_{v-1}}{\balpha_{v-1}}}{E}$ and because $\evgeo{\egd{\P_{v-1}}{\balpha_{v-1}}}{E} \semgeo{B} \egd{\P_v}{\balpha_v}$.
\end{enumerate}
As a consequence of these conditions and the induction hypothesis, we find that
\begin{align*}
  \egd{\P_v^{(u)}}{\balpha_v} - \sem{B}\left(\evgeo{\egd{\P_{v-1}^{(u)}}{\balpha_{v-1}}}{E}\right) &\mle A_v \cdot C_v^u \cdot G_v^{(u)} \\
  &\mle A_v \cdot C_v^u \cdot \egd{\P_v}{\balpha_v} \\
  &\mle A_v \cdot C_v^u \cdot c^v \cdot \egd{\R}{\gamma} \\
  &\mle A_r \cdot C_r^u \cdot c^v \cdot \egd{\R}{\gamma}
\end{align*}
where the last steps follows from monotonicity of $A$ and $C$.

\textbf{Definition of $\Q_v^{(u)}$.}
Similarly, we want to define $\egd{\Q_v^{(u)}}{\bbeta}$ such that
\[ \egd{\P_u^{(u)}}{\balpha_u} \semgeo{(\Whilst{E}{B})_v^{(u)}} \egd{\Q_v^{(u)}}{\bbeta} \quad \text{and} \quad \egd{\Q_v^{(u)}}{\bbeta} \egdle \frac{c^u}{1-c} \egd{\Q}{\bbeta} \]
We achieve this by the following inductive definition:
For $v = u$, choose it so that
\[ \egd{\P_u^{(u)}}{\balpha_u} \semgeo{\Whilst{E}{B^{(u)}}} \egd{\Q_u^{(u)}}{\bbeta} \]
Since $\egd{\P_u^{(u)}}{\balpha_u} \egdle \egd{\P_u}{\balpha_u} \egdle c^u \cdot \egd{\R}{\bgamma}$, we can choose
$\Q_u^{(u)}$ such that
\[ \egd{\P_u^{(u)}}{\balpha_u} \semgeo{\Whilst{E}{B^{(u)}}} c^u \cdot \evgeo{\egd{\frac{\R}{1-c}}{\bgamma}}{\lnot E} \egdle \frac{c^u}{1-c} \cdot \egd{\Q}{\bbeta} \]
by \cref{lem:geo-sem-well-behaved}.
For $v < u$, use \cref{lem:join-approx} to choose $\Q_v^{(u)}$ such that $\egd{\Q_v^{(u)}}{\bbeta}$ is a join of $\egd{\Q_{v+1}^{(u)}}{\bbeta}$ and $\evgeo{\egd{\P_v^{(u)}}{\balpha}}{\lnot E}$ such that
\[ \egd{\Q_v^{(u)}}{\bbeta} - \evgeo{\egd{\P_v^{(u)}}{\balpha_v}}{\lnot E} - \egd{\Q_{v+1}^{(u)}}{\bbeta} \mle \frac{B_v}{1-c} \cdot c^v \cdot \egd{\Q}{\bbeta} \]
where $B_v = A_r \cdot C_r^u \cdot c(1-c)$.
The lemma can be applied because $\frac{c^v}{1-c} \cdot \egd{\Q}{\bbeta}$ is a bound on the join of $\evgeo{\egd{\P_v^{(u)}}{\balpha_v}}{\lnot E}$ and $\egd{\Q_{v+1}^{(u)}}{\bbeta}$.
This is because
\[ \evgeo{\egd{\P_v^{(u)}}{\balpha_v}}{\lnot E} \egdle \evgeo{\egd{\P_v}{\balpha_v}}{\lnot E} \egdle c^v \cdot \evgeo{\egd{\R}{\bgamma}}{\lnot E} = c^v \cdot \egd{\Q}{\bbeta} \]
and $\egd{\Q_{v+1}^{(u)}}{\bbeta} \egdle \frac{c^{v+1}}{1-c} \egd{\Q}{\bbeta}$ by inductive hypothesis, which implies that the sum $\left(c^v + \frac{c^{v+1}}{1-c}\right) \cdot \egd{\Q}{\bbeta}$ is a valid join and equals $\frac{c^v}{1-c} \cdot \egd{\Q}{\bbeta}$.
Hence a $\egd{\Q_v^{(u)}}{\bbeta} \egdle \frac{c^v}{1-c} \egd{\Q}{\bbeta}$ with the desired properties exists.
Finally, define $\Q^{(u)} := \Q_0^{(u)}$.

Inductively, it is easy to see that
\[ \egd{\P_v^{(u)}}{\balpha_u} \semgeo{(\Whilst{E}{B})_v^{(u)}} \egd{\Q_v^{(u)}}{\bbeta} \]
In the base case $v = u$, this is clear from the definition.
For the inductive step, we can assume $\egd{\P_{v+1}^{(u)}}{\balpha_{v+1}} \semgeo{(\Whilst{E}{B})_{v+1}^{(u)}} \egd{\Q_{v+1}^{(u)}}{\bbeta}$.
We have $\evgeo{\egd{\P_v^{(u)}}{\balpha_v}}{E} \semgeo{B^{(u)}} \egd{\P_{v+1}^{(u)}}{\balpha_{v+1}}$ by definition, so
\[ \evgeo{\egd{\P_v^{(u)}}{\balpha_v}}{E} \semgeo{B^{(u)}; (\Whilst{E}{B})_{v+1}^{(u)}} \egd{\Q_{v+1}^{(u)}}{\bbeta} \]
Since $\egd{\Q_v^{(u)}}{\bbeta}$ is a join of $\egd{\Q_{v+1}^{(u)}}{\bbeta}$ and $\evgeo{\egd{\P_v^{u}}{\balpha}}{\lnot E}$, the semantics of $\Ite{-}{-}{-}$ implies
\[ \egd{\P_v^{(u)}}{\balpha_v} \semgeo{(\Whilst{E}{B})_v^{(u)}} \egd{\Q_v^{(u)}}{\bbeta} \]
because $(\Whilst{E}{B})_v^{(u)} = \Ite{E}{B^{(u)}; (\Whilst{E}{B})_{v+1}^{(u)}}{\Skip}$.

\textbf{Actual proof for the loop case.}
Consider $P^{(u)} = (\Whilst{E}{B})_0^{(u)}$.
The proof for the while case is by induction on $v$ in $(\Whilst{E}{B})_v^{(u)}$.
We need to strengthen our induction hypothesis to the following:
\begin{align}
  \egd{\Q_v^{(u)}}{\bbeta} - \sem{P}\left(\egd{\P_v^{(u)}}{\balpha_v}\right) \mle \left(c^u + 2 A_r C_r^u \cdot \frac{c^{v+1}}{1-c}\right) \cdot \egd{\Q}{\bbeta} \label{eq:while-strengthened}
\end{align}

We first consider the base case $v = u$:
We have $(\Whilst{E}{B})_u^{(u)} = \Whilst{E}{B^{(u)}}$.
We remarked above that $\egd{\R}{\bgamma}$ is a $c$-contraction invariant for $\Whilst{E}{B^{(u)}}$.
Furthermore we have $\egd{\P_u^{(u)}}{\balpha_u} \egdle \egd{\P_u}{\balpha_u} \egdle c^u \cdot \egd{\R}{\bgamma}$ and thus:
\begin{align*}
  \egd{\P_u^{(u)}}{\balpha_u} \semgeo{(\Whilst{E}{B})_u^{(u)}} c^u \cdot \evgeo{\egd{\frac{\R}{1-c}}{\bgamma}}{\lnot E}
  &= c^u \cdot \egd{\Q}{\bbeta} \\
  &\mle \left(c^u + 2 A_r C_r^u \frac{c^{u+1}}{1-c}\right) \egd{\Q}{\bbeta}
\end{align*}

Now for the inductive step $v < u$:
By the loop unrolling property, we have:
\[ \sem{\Whilst{E}{B}}\left(\egd{\P_v^{(u)}}{\balpha_v}\right) = \sem{\Whilst{E}{B}}\left(\sem{B}\left(\evgeo{\egd{\P_v^{(u)}}{\balpha_v}}{E}\right)\right) + \evgeo{\egd{\P_v^{(u)}}{\balpha_v}}{\lnot E} \]
Furthermore, since $\Whilst{E}{B}_v^{(u)}$ has the same semantics as $\Whilst{E}{B}$, we find:
\begin{align*}
&\egd{\Q_v^{(u)}}{\bbeta} - \sem{\Whilst{E}{B}}\left(\egd{\P_v^{(u)}}{\balpha_v}\right) \\
&= \egd{\Q_v^{(u)}}{\bbeta} - \left(\sem{\Whilst{E}{B}}\left(\sem{B}\left(\evgeo{\egd{\P_v^{(u)}}{\balpha_v}}{E}\right)\right) + \evgeo{\egd{\P_v^{(u)}}{\balpha_v}}{\lnot E}\right) \\
&= \egd{\Q_v^{(u)}}{\bbeta} - \evgeo{\egd{\P_v^{(u)}}{\balpha_v}}{\lnot E} - \egd{\Q_{v+1}^{(u)}}{\bbeta} \\
&\quad + \egd{\Q_{v+1}^{(u)}}{\bbeta} - \sem{\Whilst{E}{B}}\left(\sem{B}\left(\evgeo{\egd{\P_v^{(u)}}{\balpha_v}}{E}\right)\right) \\
&= \egd{\Q_v^{(u)}}{\bbeta} - \evgeo{\egd{\P_v^{(u)}}{\balpha_v}}{\lnot E} - \egd{\Q_{v+1}^{(u)}}{\bbeta} \\
&\quad + \egd{\Q_{v+1}^{(u)}}{\bbeta} - \sem{\Whilst{E}{B}}\left(\egd{\P_{v+1}^{(u)}}{\balpha_{v+1}}\right) \\
&\quad + \sem{\Whilst{E}{B}}\left(\egd{\P_{v+1}^{(u)}}{\balpha_{v+1}} - \sem{B}\left(\evgeo{\egd{\P_v^{(u)}}{\balpha_v}}{E}\right)\right) \\
&\mle \frac{B_v}{1-c} \cdot c^v \cdot \egd{\Q}{\bbeta} \\
&\quad + \left(c^u + 2 A_r C_r^u \cdot \frac{c^{v+2}}{1-c}\right) \cdot \egd{\Q}{\bbeta} \\
&\quad + \sem{\Whilst{E}{B}}\left(A_r \cdot C_r^u \cdot c^{v+1} \cdot \egd{\R}{\bgamma}\right) \\
&\mle A_r \cdot C_r^u \cdot c^{v+1} \cdot \egd{\Q}{\beta} \\
&\quad + \left(c^u + 2 A_r C_r^u \cdot \frac{c^{v+2}}{1-c}\right) \cdot \egd{\Q}{\bbeta} \\
&\quad + A_r \cdot C_r^u \cdot c^{v+1} \cdot \evgeo{\egd{\R}{\bgamma}}{\lnot E} \\
&\mle \left(c^u + 2 A_r C_r^u \cdot \frac{c^{v+2}}{1-c}\right) \cdot \egd{\Q}{\bbeta} + 2 \cdot A_r \cdot C_r^u \cdot c^{v+1} \cdot \egd{\Q}{\bbeta} \\
&\mle \left(c^u + 2 A_r C_r^u \cdot \frac{c^{v+1}}{1-c}\right) \cdot \egd{\Q}{\bbeta} \\
\end{align*}
This proves \cref{eq:while-strengthened}.
As a consequence, we have (since $\P = \P_0^{(u)}$ and $\balpha = \balpha_0$):
\[ \egd{\Q^{(u)}}{\bbeta} - \sem{P}(\egd{\P}{\balpha}) \mle \left(c^u + 2 A_r C_r^u \cdot \frac{c}{1-c}\right) \cdot \egd{\Q}{\bbeta} \mle A \cdot C^u \cdot \egd{\Q}{\bbeta} \]
with $A := 1 + \frac{2A_r c}{1-c}$ and $C := \max(c, C_r)$.
Furthermore, $A, C$ can be chosen to be monotonic functions of $\egd{\P}{\balpha}$, by choosing $\egd{\R}{\bgamma}$ as a monotonic function of $\egd{\P}{\balpha}$.
This finishes the induction and the proof.
\end{proof}

The following lemma was needed in the proof.
\begin{lemma}
\label{lem:join-approx}
Let $A > 0$.
Suppose $\left(\egd{\R}{\bgamma}, \egd{\S}{\bdelta}, \egd{\Q}{\bbeta}\right) \in \JoinRel$.
Then there are $\P, \balpha$ such that $\egd{\P}{\balpha} \egdle \egd{\Q}{\bbeta}$ and $\left(\egd{\R}{\bgamma}, \egd{\S}{\bdelta}, \egd{\P}{\balpha}\right) \in \JoinRel$ and $\egd{\P}{\balpha} - \egd{\R}{\bgamma} - \egd{\S}{\bdelta} \mle A \cdot \egd{\Q}{\bbeta}$.
\end{lemma}
\begin{proof}
Let $\egd{\R^*}{\bgamma}, \egd{\S^*}{\bdelta}$ be the expansions of $\egd{\R}{\bgamma}$ and $\egd{\S}{\bdelta}$ such that $\R^* + \S^* = \Q$.
We set $\balpha = \max(\bgamma, \bdelta)$.
Clearly, $\balpha \le \bbeta$ and thus $\egd{\Q}{\balpha} \egdle \egd{\Q}{\bbeta}$.

The proof is by induction on $n = |\balpha|$.
If $n = 0$ then $\egd{\P}{\balpha}= \P = \R + \S \in \RR$ and $\egd{\P}{\balpha} - \egd{\R}{\bgamma} - \egd{\S}{\bdelta} = \P - \R - \S = 0$.

If $n > 0$, we assume without loss of generality that $\gamma_n \le \delta_n$ and thus $\alpha_n = \delta_n$.
Let $\balpha' := (\alpha_1, \dots, \alpha_{n-1})$, $\bbeta' := (\beta_1, \dots, \beta_{n-1})$, $\bgamma' := (\gamma_1, \dots, \gamma_{n-1})$, and $\bdelta' := (\delta_1, \dots, \delta_{n-1})$.
In this proof we will write $\T_j$ for $\T_{n: j}$, i.e. the $j$-the component in the $n$-th dimension of a tensor $t$.

Let $m > |\Q|_n$ to be chosen later and let $\egd{\R^{(m)}}{\bgamma}$, $\egd{\S^{(m)}}{\bdelta}$, $\egd{\Q^{(m)}}{\bbeta}$ be the expansions of $\egd{\R^*}{\bgamma}$, $\egd{\S^*}{\bdelta}$, and $\egd{\Q}{\bbeta}$ to size $m$ in dimension $n$.
Then for $j < m$, we can find by the inductive hypothesis some $\P_j$ with $\left(\egd{\R_j^{(m)}}{\bgamma'}, \egd{\S_j^{(m)}}{\bdelta'}, \egd{\P_j}{\balpha'}\right) \in \JoinRel$ and
\[ \egd{\P_j^{(m)}}{\balpha'} - \egd{\R_j^{(m)}}{\bgamma'} - \egd{\S_j^{(m)}}{\bdelta'} \mle A/2 \cdot \egd{\Q^{(m)}_j}{\bbeta'} \]
Let $[m_1] \times \cdots \times [m_{n-1}]$ be the maximum dimensions of $\P_j$.
Without loss of generality, we can assume that all $\R^{(m)}_j, \S^{(m)}_j, \P_j$ are expanded to these maximum dimensions.
Collecting the $\P_j$'s into a big vector yields $\P = (\P_0, \dots, \P_{|\Q|_n})$.
We claim $\egd{\P}{\balpha} - \egd{\R}{\bgamma} - \egd{\S}{\bdelta} \mle A \cdot \egd{\Q}{\bbeta}$.
We show this by looking at the $j$-th component, i.e. $\mu_j(\x') := \mu(\x', j)$ for $\x' \in \NN^{n-1}$.
For $j < m$, this follows from the definition.
For $j \ge m$, we analyze two cases.

If $\gamma_n = \delta_n$, then $\alpha_n = \gamma_n = \delta_n$ by definition.
We choose $m = |\Q|_n$ and find for $j \ge m = |\Q|_n$:
\begin{align*}
&\egd{\P_{m-1} \cdot \alpha_n^{j - m + 1}}{\balpha'} - \egd{\R_{m-1}^{(m)} \cdot \gamma_n^{j - m + 1}}{\bgamma'} - \egd{\S_{m-1}^{(m)} \cdot \delta_n^{j - m + 1}}{\bdelta'} \\
&= \left(\egd{\P_{m-1}}{\balpha'} - \egd{\R_{m-1}^{(m)}}{\bgamma'} - \egd{\S_{m-1}^{(m)}}{\bdelta'}\right) \cdot \alpha_n^{j - m + 1} \\
&\mle A \cdot \egd{\Q_{m-1}}{\bbeta'} \cdot \beta_n^{j - m + 1}
\end{align*}

If $\gamma_n < \delta_n = \alpha_n$, we find for $j \ge m$:
\begin{align*}
&\egd{\P_{m-1}}{\balpha'} \cdot \alpha_n^{j-m+1} - \egd{\R^{(m)}_{m-1}}{\bgamma'} \cdot \gamma_n^{j-m+1} - \egd{\S^{(m)}_{m-1}}{\bdelta'} \cdot \delta_n^{j-m+1} \\
&= \left(\egd{\P_{m-1}}{\balpha'} - \egd{\R^{(m)}_{m-1}}{\bgamma'} - \egd{\S^{(m)}_{m-1}}{\bdelta'}\right) \cdot \alpha_n^{j-m+1} \\
&\quad + \egd{\R^{(m)}_{m-1}}{\balpha'} \cdot \alpha_n^{j-m+1} - \egd{\R^{(m)}_{m-1}}{\bgamma'} \cdot \gamma_n^{j-m+1} \\
&\mle A/2 \cdot \egd{\Q_{m-1}}{\bbeta'} \cdot \beta_n^{j-m+1} + \egd{\R^{(m)}_{m-1}}{\balpha'} \cdot \alpha_n^{j-m+1} \\
&= A/2 \cdot \egd{\Q_{|\Q|_n-1} \cdot \beta_n^{m - |\Q|_n}}{\bbeta'} \cdot \beta_n^{j-m+1} + \egd{\R^*_{|\Q|_n - 1} \cdot \gamma_n^{m - |\Q|_n}}{\balpha'} \cdot \alpha_n^{j-m+1} \\
&\mle A/2 \cdot \egd{\Q_{|\Q|_n-1}}{\bbeta'} \cdot \beta_n^{j - |\Q|_n + 1} + \egd{\R^*_{|\Q|_n - 1}}{\bbeta'} \cdot \alpha_n^{j - |\Q|_n + 1} \cdot \left(\frac{\gamma_n}{\alpha_n}\right)^{m - |\Q|_n} \\
&\mle \left(A/2 + \left(\frac{\gamma_n}{\alpha_n}\right)^{m - |\Q|_n}\right) \cdot \egd{\Q_{|\Q|_n-1}}{\bbeta'} \cdot \beta_n^{j - |\Q|_n + 1} \\
&\mle A \cdot \egd{\Q_{|\Q|_n-1}}{\bbeta'} \cdot \beta_n^{j - |\Q|_n + 1}
\end{align*}
where the last step follows from $\left(\frac{\gamma_n}{\alpha_n}\right)^{m - |\Q|_n} \to 0$ as $m \to \infty$, so we can choose $m$ large enough such that the bound holds for all $j \ge m$.

So the bound holds for $j < m$ and $j \ge m$, and thus we find, as desired:
\[ \egd{\P}{\balpha} - \egd{\R}{\bgamma} - \egd{\S}{\bdelta} \mle A \cdot \egd{\Q}{\balpha} \qedhere \]
\end{proof}

\begin{example}[Failure of geometric bound semantics]
\Cref{ex:imprecise-tails} can easily be extended to an example where the geometric bound semantics fails even though an EGD bound for the whole program exists.
Recall that at the end of the program, the true distribution is $\Dirac{(0,0)}$, but the geometric bound semantics finds a bound with nonzero tails.
If we append $\Ite{X_1 = 0}{\Skip}{\Diverge}$ to the program, this does not affect the true distribution (because $X_1 = 0$ almost surely), but the geometric bound semantics will now fail because $\Diverge := \Whilst{\Flip(1)}{\Skip}$ does not admit a geometric bound since it has infinite expected running time.
\end{example}
\section{Supplementary Material for \cref{sec:impl}}
\label{apx:impl}

\subsection{Overapproximating the Measure Support}
\label{apx:approx-support}

The residual mass bound is a bound on the total mass of the distribution $\nu := \sem{P}(\mu) - \semlo{P}(\mu)$.
In practice, $\nu$ often has bounded support.
This is useful for our bounds because we know that for measurable sets $S \subseteq \estates$ outside the support of $\nu$, the lower bound is tight: $\semlo{P}(\mu)(S) = \esem{P}(\mu)(S)$.
Suppose the support of $\nu$ is contained in a Cartesian product of intervals $J := J_1 \times \dots \times J_n$.
Then we know that for discrete variables $X_k$, the lower bound on the probability mass $\Prob[X_k = i]$ is exact for $i \notin J_k$.
Even more, if $J_k$ is bounded, then we can even compute upper bounds on moments:
\[ \ExpVal_{\X \sim \sem{P}(\mu)}[X_k] - \ExpVal_{\X \sim \semlo{P}(\mu)}[X_k] = \ExpVal_{\X \sim \nu}[X_k] \le \nu(\estates) \cdot \max J_k \le \semres{P}(\mu) \cdot \max J_k \]

To compute such an overapproximation $J := J_1 \times \dots \times J_n$ of the support of $\nu$, we use the standard technique of abstract interpretation with interval arithmetic.
Since the lower bound ignores loops, computing $J$ amounts to bounding the range of values of variables after loops.
Such a program analysis is standard, so we describe it only briefly.

As the abstract domain, we use products of intervals $I_1 \times \dots \times I_n$ where each $I_j$ is an interval $[a_j, b_j]$ with $a_j \in \NN$ and $b_j \in \NN \cup \{\infty\}$.
This domain is extended with a special element $\bot = \emptyset \times \dots \times \emptyset$ representing an empty support, i.e. the zero measure.
The join operation is defined as the pointwise join of the intervals.
The abstract versions of addition and subtraction are standard.
As the widening operator $\widen$ on intervals, we use $[a, b] \widen [c, d] = [e, f]$ with $e = a$ if $a \le c$ else $e = 0$, and $f = b$ if $b \ge d$ else $f = \infty$.
This is extended pointwise to products of intervals.

\begin{example}
For the asymmetric random walk example (\cref{example:asym-rw}), the interval analysis finds that the support of $\nu := \sem{P^{(u)}}(\Dirac(0, 0)) - \semlo{P^{(u)}}(\Dirac(0, 0))$ is contained in $[0, 0] \times [u, \infty]$.
This tells us that the lower bounds on the probability masses $\Prob[X_2 = i]$ are exact for $i < u$.
\end{example}

\subsection{Details on the ADAM-based Solver and Optimizer}
\label{apx:impl-solvers}

The success of the (gradient-based) ADAM optimizer in high-dimensional problems in machine learning inspired us to apply it to optimizing our bounds as well.
Since ADAM works on unconstrained optimization problems, we first need to reduce the constrained optimization problem to an unconstrained one.
Our idea is to add penalty terms to the objective function for constraints that are violated or almost violated, to steer the solver towards feasible points.
Specifically, for each constraint $f(\x) \le 0$, we add a penalty term to the objective function:
$\exp(\lambda \mathsf{dist}_f(\x))$, where $\mathsf{dist}_f(\x) = \frac{f(\x)}{\| \nabla f(\x) \|}$ is an approximation of the (signed) distance to the constraint boundary.
This term is near zero if the constraint is comfortably satisfied, but large if the constraint is violated.
In that sense it is similar to \emph{penalty function methods} used in constrained optimization, except penalty functions are \emph{exactly} zero when the constraint is satisfied.
The parameter $\lambda$ is increased during the optimization: we take it to be the current iteration number.
A higher $\lambda$ leads to a ``sharper'' transition from near zero to a very high penalty at the constraint boundary.

How do we find an initial point for the optimization?
In many cases (see proof of \cref{thm:sufficient-cond}), there is a point with $\bbeta = \one_n$ that is close to a feasible point.
Thus we simply initialize $\bbeta$ at the point $0.999 \cdot \one_n$ (i.e. very close to one) in our implementation.
This initialization is typically good enough for our optimizer to quickly find a feasible point.
However, as mentioned in the main text, IPOPT is better at finding feasible points than our ADAM-based approach.
The primary strength of the ADAM-based optimizer is to improve the geometric bounds, particularly the tail bounds.

\section{Supplementary Material for \cref{sec:eval}}
\label{apx:eval}

\paragraph{Details on benchmarks selection}
The Polar benchmarks are all unnested loops without conditioning.
Since Polar is also used for loop invariant synthesis, most of their operate on loops of the form $\Whilst{\True}{\dots}$.
We excluded benchmarks with such infinite loops.
Since PSI cannot handle unbounded loops, the loops in its repository were all either artificially bounded or PSI could not solve them.
We remove the artificial bounds for the purposes of our benchmarks.
This explains why PSI cannot solve any of ``its own'' benchmarks.

\begin{table}
\centering
\footnotesize
\caption{Existence of geometric bounds: for each benchmark, we present some statistics about the program, the constraint problem arising from the geometric bound semantics, and the solving time, if successful. (*: cannot be solved by the respective tool; \#Var: number of program variables; \#Stmt: number of statements; Obs?: presence of observations; \#Constr: number of constraints; \#ConstrVar: number of constraint variables; Time: time to compute bounds (or error); timeout: 5 minutes exceeded; infeasible: no solution to the constraints exists; $\ExpVal[\text{runtime}] = \infty$: expected runtime is infinite, implies infeasible.)}
\label{table:benchmarks-applicability}
\alternaterowcolors
\begin{tabular}{lrrrrrr}
\toprule
Benchmark & \#Var & \#Stmt & Obs? & \#Constr & \#ConstrVar & Time \\
\midrule
\verb|polar/c4B_t303|* & 3 & 21 & \xmark & 6 & 13 & 0.02 s \\
\verb|polar/coupon_collector2| & 5 & 12 & \xmark & 2 & 10 & 0.02 s \\
\verb|polar/fair_biased_coin| & 2 & 5 & \xmark & 2 & 5 & 0.01 s \\
\verb|polar/geometric| (\cref{ex:power-of-two}) & 4 & 13 & \xmark & 21 & 18 & \xmark{} ($\ExpVal[\text{runtime}] = \infty$) \\
\verb|polar/las_vegas_search| & 3 & 7 & \xmark & 20 & 23 & 0.02 s \\
\verb|polar/linear01|* & 1 & 5 & \xmark & 2 & 3 & 0.02 s \\
\verb|polar/rabin|* & 4 & 19 & \xmark & 45 & 34 & \xmark{} (infeasible) \\
\verb|polar/random_walk_2d| & 3 & 11 & \xmark & 11 & 7 & \xmark{} ($\ExpVal[\text{runtime}] = \infty$) \\
\verb|polar/simple_loop| & 2 & 8 & \xmark & 1 & 3 & 0.01 s \\
\verb|prodigy/bit_flip_conditioning| & 5 & 13 & \cmark & 3 & 18 & 0.02 s \\
\verb|prodigy/brp_obs| & 2 & 9 & \cmark & 93 & 8 & 2.79 s \\
\verb|prodigy/condand|* & 2 & 6 & \xmark & 4 & 6 & 0.02 s \\
\verb|prodigy/dep_bern| & 3 & 8 & \xmark & 4 & 5 & 0.05 s \\
\verb|prodigy/endless_conditioning| & 1 & 7 & \cmark & 1 & 2 & 0.09 s \\
\verb|prodigy/geometric| & 2 & 6 & \xmark & 3 & 4 & 0.11 s \\
\verb|prodigy/ky_die| & 2 & 37 & \xmark & 42 & 43 & 0.03 s \\
\verb|prodigy/n_geometric| & 2 & 6 & \xmark & 3 & 4 & 0.01 s \\
\verb|prodigy/nested_while| & 3 & 11 & \xmark & 12 & 10 & \xmark{} (infeasible) \\
\verb|prodigy/random_walk| & 2 & 7 & \xmark & 2 & 4 & \xmark{} ($\ExpVal[\text{runtime}] = \infty$) \\
\verb|prodigy/trivial_iid| & 2 & 5 & \xmark & 12 & 4 & 0.03 s \\
\verb|psi/beauquier-etal3|* & 14 & 39 & \xmark & 64 & 8194 & 2.43 s \\
\verb|psi/cav-example7|* & 2 & 16 & \xmark & 4 & 6 & 0.15 s \\
\verb|psi/dieCond|* & 2 & 7 & \cmark & 2 & 8 & 0.02 s \\
\verb|psi/ex3|* & 1 & 5 & \xmark & 2 & 3 & 0.02 s \\
\verb|psi/ex4|* & 2 & 5 & \xmark & 1 & 2 & 0.06 s \\
\verb|psi/fourcards|* & 7 & 21 & \xmark & 112 & 257 & 0.09 s \\
\verb|psi/herman3|* & 8 & 33 & \xmark & 16 & 258 & 0.04 s \\
\verb|psi/israeli-jalfon3|* & 8 & 30 & \xmark & 4 & 130 & 0.03 s \\
\verb|psi/israeli-jalfon5|* & 12 & 48 & \xmark & 116 & 4098 & 2.99 s \\
\verb|ours/1d-asym-rw| (\cref{example:asym-rw}) & 2 & 7 & \xmark & 2 & 4 & 0.14 s \\
\verb|ours/2d-asym-rw| & 3 & 12 & \xmark & 4 & 5 & 0.03 s \\
\verb|ours/3d-asym-rw| & 4 & 17 & \xmark & 126 & 59 & 0.18 s \\
\verb|ours/asym-rw-conditioning| & 3 & 10 & \cmark & 6 & 6 & 0.02 s \\
\verb|ours/coupon-collector5| & 7 & 18 & \xmark & 75 & 162 & 0.04 s \\
\verb|ours/double-geo| (\cref{ex:double-geo}) & 2 & 5 & \xmark & 2 & 4 & 0.01 s \\
\verb|ours/geometric| (\cref{example:geo-counter}) & 1 & 3 & \xmark & 1 & 3 & 0.02 s \\
\verb|ours/grid| & 2 & 6 & \xmark & 11 & 12 & 0.22 s \\
\verb|ours/herman5| & 12 & 38 & \xmark & 214 & 6146 & \xmark{} (timeout) \\
\verb|ours/imprecise_tails| (\cref{ex:imprecise-tails}) & 2 & 8 & \xmark & 14 & 13 & 0.10 s \\
\verb|ours/israeli-jalfon4| & 10 & 39 & \xmark & 25 & 770 & 0.09 s \\
\verb|ours/nested| & 2 & 7 & \xmark & 12 & 10 & 0.02 s \\
\verb|ours/sub-geom| & 1 & 3 & \xmark & 3 & 4 & 0.01 s \\
\verb|ours/sum-geos| & 2 & 5 & \xmark & 3 & 5 & 0.03 s \\
\bottomrule
\end{tabular}
\end{table}

\paragraph{Detailed applicability results for \cref{sec:eval-applicability}}
The results of the geometric bound semantics, as implemented in Diabolo, on all benchmarks are shown in \cref{table:benchmarks-applicability}, along with statistics about the input program and the constraint problem arising from the geometric bound semantics.
For 5 benchmarks, Z3 could prove that no EGD bound exists (for invariant sizes $d=1,2,3,4$), these are marked as ``infeasible''.
For 3 of them, we manually determined that the expected runtime is infinite, which implies no EGD bound of any size exists (\cref{thm:necessary-runtime}).
We suspect that no EGD bound (for any invariant size) exists for the other 2 benchmarks either.

\paragraph{Detailed comparison with martingale-based methods.}
For the example of the asymmetric random walk (\cref{example:asym-rw}), we take $X_2$ as the ranking supermartingale $M$ and find that $M_{i+1} = M_i + r \cdot 1 + (1 - r) \cdot (-1) = M_i - (1 - 2r)$.
As discussed in \citep{ChatterjeeFNH16}, for a ranking supermartingale $M$ with $M_{i+1} \le M_i - \epsilon$, the tail of the termination time $T$ is bounded by $\Prob[T > n] \le \exp\left(-\frac{2(\epsilon((n-1)-M_1))^2}{(n-1)(b-a)^2}\right)$.
In our case, we have $\epsilon = 1 - 2r$ and $b - a = 2$.
Thus we find
\begin{align*}
\Prob[T > n] &\le \exp\left(-\frac{2(\epsilon ((n - 1) - M_1))^2}{2(n-1)(b-a)^2}\right)
= O\left(\exp\left(-\frac{((1 - 2r)(n-1))^2}{2(n-1)}\right)\right) \\
&= O\left(\exp\left(-\frac{(1 - 2r)^2(n-1)}{2}\right)\right)
= O\left(\left(\sqrt{\exp(-(1 - 2r)^2)}\right)^n\right)
\end{align*}
Note that $\exp(-(1 - 2r)^2) = \exp(-1 + 4r - 4r^2) > 4r(1-r)$ for $r \ne 0$ by the inequality $\exp(x) > 1 + x$ for $x \ne 0$.
By contrast, our method yields $\Prob[T > n] \le c^{n-1}(1-c) = O(c^n) = O((2\sqrt{r(1-r)} + \delta)^n)$ for any $\delta > 0$.
Hence our method yields better asymptotic tail bounds than the martingale-based method.

Similarly, for the geometric example (\cref{example:geo-counter}), our method yields the tail bound $O((\frac12 + \delta)^{n})$.
The martingale is given by the variable $M_i := \Flip(\frac{1}{2})$ implicit in the loop condition.
It satisfies
\[ \ExpVal[M_{i+1} \mid M_i] = \begin{cases}
\frac{1}{2} & \text{if } M_i = 1 \\
0 & \text{if } M_i = 0
\end{cases} \le M_i - \epsilon \cdot [M_i > 0] \]
for $\epsilon = \frac12$.
Since the martingale is bounded between 0 and 1, we have $b - a = 1$ and find the tail bound
\[
\Prob[T > n] \le \exp\left(-\frac{2(\epsilon ((n - 1) - M_1))^2}{(n-1)(b-a)^2}\right)
= O\left(\exp\left(-\frac{2(\frac{1}{2} (n - 1))^2}{(n-1)}\right)\right)
= O((\exp(-1/2))^n)
\]

\clearpage
\tableofcontents

\fi

\end{document}